\documentclass[a4paper, english]{article}

\usepackage{babel}[english]

\usepackage{amsmath, amssymb,amsthm}
\usepackage{mathrsfs}
\usepackage{hyperref}
\usepackage{cmll}
\usepackage{bm, stmaryrd, relsize, bussproofs}
\usepackage{thm-restate}
\usepackage{tikz}
\usetikzlibrary{positioning,decorations,decorations.pathmorphing}
\usepackage{adjustbox}

\title{Enumerating Error Bounded Polytime Algorithms Through Arithmetical Theories \\
(Extended Version)
} 

\author{Melissa Antonelli\\
{\small Helsinki Institute for Information Technology, Finland}\\
{\small\texttt{melissa.antonelli@helsinki.fi}}
\\ \ \\
Ugo Dal Lago\\
{\small Bologna University, Italy,}\\
{\small Inria, Universit\'e C\^ote d'Azur, France  }\\
{\small \texttt{ugo.dallago@unibo.it}}
\\ \ \\
Davide Davoli\\
{\small Inria, Universit\'e C\^ote d'Azur, France}\\
{\small \texttt{davide.davoli@inria.fr}}
\\ \ \\
Isabel Oitavem\\
{\small Nova University Lisbon, Portugal}
\\
{\small \texttt{oitavem@fct.unl.pt}}
\\ \ \\
Paolo Pistone\\
{\small Bologna University, Italy}
\\
{\small \texttt{paolo.pistone2@unibo.it}}
}
\date{}

\newtheorem{theorem}{Theorem}[section]
\newtheorem{example}{Example}[section]
\newtheorem{remark}{Remark}[section]
\newtheorem{corollary}{Corollary}[section]
\theoremstyle{definition}
\newtheorem{defn}{Definition}
\newtheorem{notation}{Notation}

\newtheorem{prop}{Proposition}

\newtheorem{lemma}[theorem]{Lemma}
\newtheorem{proposition}[theorem]{Proposition}

\newcommand{\POR}{\mathcal{POR}}

\newcommand{\Stm}{\stm}

\newcommand{\query}{\emph{Q}}
\newcommand{\Sf}{\emph{S}}
\newcommand{\Cf}{\emph{C}}

\newcommand{\apply}{\mathit{step}}
\newcommand{\dectape}{\mathit{dectape}}
\newcommand{\sa}{\mathit{sa}}

\newcommand{\rfp}[1]{\langle#1\rangle}
\newcommand{\LANG}[1]{ \mathrm{Lang}(#1)}
\newcommand{\THE}{ \mathsf{T}}

\newcommand{\BOX}{\mathbf{C}}

\newcommand{\ovverline}[2]{{\underline{#1}}_{#2}}

\newcommand{\reaches}[2]{\triangleright^{#1}_{#2}}
\newcommand{\Qs}{\mathcal{Q}}

\newcommand{\RFP}{\PPT}
\newcommand{\PPT}{\mathbf{RFP}}
\newcommand{\SFP}{\mathbf{SFP}}

\newcommand{\SIFP}{\mathbf{SIFP}}
\newcommand{\RA}{{\mathbf{RA}}}
\newcommand{\LA}{{\mathbf{LA}}}
\newcommand{\SIFPRA}{\mathbf{SIFP_{\RA}}}
\newcommand{\SFPOD}{\mathbf{SFP_{\mathbf{OD}}}}
\newcommand{\SIFPLA}{\mathbf{SIFP_{\LA}}}

\newcommand{\lang}[1]{{\mathcal L ({#1})}}

\newcommand{\id}{\mathsf{Id}}
\newcommand{\stm}{\mathsf{Stm}}
\newcommand{\xp}{\mathsf{Exp}}
\newcommand{\fl}[1]{\mathtt{Flip(}#1\mathsf{)}}
\newcommand{\rb}{\mathtt{RandBit()}}
\newcommand{\while}[2]{\mathtt{while}(#1)\{#2\}}

\newcommand{\takes}{\leftarrow}
\newcommand{\store}{\Sigma}
\newcommand{\as}[2]{[#1 \leftarrow #2]}
\newcommand{\ssos}{\triangleright}
\newcommand{\sred}{\rightharpoonup}

\newcommand{\LL}[1]{{\mathfrak L_{#1}}}

\newcommand{\leadstola}{{\leadsto_\LA}}
\newcommand{\leadstora}{{\leadsto_\RA}}

\newcommand{\RS}{{R}\Sigma^b_1\text{-NIA}}
\newcommand{\RSE}{{R}\Sigma^b_1\text{-NIA} + \mathrm{Exp}}

\newcommand{\PA}{\mathsf{PA}}

\newcommand{\Lpw}{\mathcal{RL}}

\newcommand{\MEAS}{\mathsf{TwoThirds}}
\newcommand{\Flipless}{\mathsf{NoFlip}}
\newcommand{\Thresh}{\mathsf{Threshold}}

\newcommand{\IDE}{\mathsf{I\Delta}_{0}+\mathrm{Exp}}

\newcommand{\zzero}{{0}}
\newcommand{\oone}{{1}}
\newcommand{\bbool}{{b}}
\newcommand{\cconc}{{\conc}}
\newcommand{\ttimes}{{\times}}
\newcommand{\ssubseteq}{{\subseteq}}

\newcommand{\eepsilon}{{\epsilon}}

\newcommand{\model}[1]{\llbracket #1\rrbracket}

\newcommand{\conc}{\frown}
\newcommand{\bool}{\mathtt{b}}
\newcommand{\zero}{\mathtt{0}}
\newcommand{\one}{\mathtt{1}}

\newcommand{\Flip}{\mathtt{Flip}}

\newcommand{\Bool}{\mathbb{B}}
\newcommand{\Ss}{\mathbb{S}}
\newcommand{\Os}{\mathbb{O}}
\newcommand{\Nat}{\mathbb{N}}
\newcommand{\Dist}[1]{\mathbb{D}(#1)}

\newcommand{\realize}{\circledR}

\newcommand{\BPP}{\mathbf{BPP}}
\newcommand{\RP}{\mathbf{RP}}

\newcommand{\PP}{\mathbf{PP}}
\newcommand{\PTM}{\mathscr{M}}
\newcommand{\STM}{\mathscr{S}}

\newcommand{\POLY}{\mathsf{POLY}}

\newcommand{\midd}{\; \; \mbox{\Large{$\mid$}}\;\;}

\newcommand{\NP}{\mathbf{NP}}

\newcommand{\nruns}{{37m}}

\newcommand{\zeroT}{\mathsf{0}}
\newcommand{\oneT}{\mathsf{1}}

\newcommand{\arrowT}{\Rightarrow}

\newcommand{\ooverline}[1]{\overline{\overline{#1}}}

\newtheorem{cor}{Corollary}

\begin{document}

\maketitle

\begin{abstract}
	We consider a minimal extension of the language of arithmetic, 
	such that the bounded formulas provably total in a suitably-defined theory \emph{à la Buss} (expressed in this new language)
	precisely capture polytime \emph{random} functions. 
	Then, we provide two new characterizations of the semantic class $\mathbf{BPP}$ obtained by internalizing 
	the error-bound check \emph{within} a logical system:
	the first relies on measure-sensitive quantifiers, while the second is based on standard first-order quantification. 
	This leads us to introduce a family of effectively enumerable 
	subclasses of $\BPP$, called $\mathbf{BPP}_{\mathsf{T}}$ and consisting of languages captured by those probabilistic Turing machines whose underlying error can be proved bounded in $\mathsf T$. 
	As a paradigmatic example of this approach, 
	we establish that polynomial identity testing is in 
	$\mathbf{BPP}_{\mathsf T}$, where $\mathsf T=\IDE$ is a well-studied theory based on bounded induction. 
\end{abstract}

\section{Introduction}\label{sec:introduction}\label{section1}
Since the early days of computer science,
numerous and profound interactions with mathematical logic have emerged (think of the seminal works by 
Turing~\cite{Turing} and Church~\cite{Church}). 
Among the sub-fields of computer science that have benefited the most from this dialogue, we should 
certainly mention the theory of programming languages (e.g.~through the 
Curry-Howard correspondence \cite{Curry34,Howard,SorUrz}), the theory of 
databases (e.g.~through Codd's theorem~\cite{Codd}) and computational complexity 
(e.g.~through descriptive 
complexity~\cite{BellantoniCook,TextbookDescriptiveComplexity}). 
In particular, 
this last discipline deals with complexity 
classes~\cite{HartmanisStearns,Cobham,AroraBarak},  the nature of which still 
remains today, more than fifty years after the introduction of $\mathbf{P}$ and 
$\mathbf{NP}$~\cite{Cook,HartmanisStearns}, somewhat obscure.

The possibility of describing fundamental  classes within the 
language of mathematical logic offered a better understanding of 
their 
nature: since the seventies~\cite{Fagin,CookReckhow}, but 
especially from the eighties and 
nineties~\cite{Buss86,GirardScedrovScott,BellantoniCook,TextbookDescriptiveComplexity,ProofComplexity},
 the logical characterization of several crucial classes has made it possible 
to consider them from a new 
viewpoint, less dependent on concrete machine 
models and explicit resource bounds. 
Characterizing complexity classes by way of 
a simple enough proof-of-recursion theoretical system also means being able to 
\emph{enumerate} the problems belonging to them,
and thus to devise sound and 
complete languages for the class, from which type systems and static analysis 
methodologies can be derived~\cite{HofmannSurvey}.

Among the various classes of problems considered in computational complexity,
those defined on the basis of \emph{randomized} algorithms~\cite{Motwani} have 
appeared difficult to capture with the tools of logic. 
These include important and well-studied classes like $\mathbf{BPP}$ or $\mathbf{ZPP}$.
The former, in particular, is often considered as \emph{the} class of feasible problems, 
and most complexity theorists conjecture that it actually coincides with $\mathbf{P}$.
One might thus expect it to be possible to obtain an enumeration of $\mathbf{BPP}$, along the lines of the many examples known for classes like $\mathbf{P}$, or even $\mathbf{PP}$ \cite{DalLagoKahleOitavem21,DalLagoKahleOitavem22}.
However, by simply looking at its definition, $\mathbf{BPP}$ looks pretty different from $\mathbf{P}$. 
Notably, the former, but not the latter, is an example of what is usually
called a \emph{semantic} class:
the definition of 
$\mathbf{BPP}$ relies on
algorithms which are both efficient and not \emph{too erratic}: once an input is fixed, 
one of the two possible outputs must clearly prevail over the other; in other words, 
there is some fixed probability $p$, bounded away from $\frac{1}{2}$, such that, on any input $x$, 
the machine outputs some value $b_{x}\in\{0,1\}$ with probability at least $p$.
The existence of an effective enumerable family of algorithms deciding \emph{all and only} the problems in $\BPP$ is still an open question.


In this paper we make a step towards a 
logical understanding of semantic complexity classes, and in particular of the 
logical and proof-theoretic complexity involved in keeping error-bounds under control.
Our contributions can be divided in three parts. 
First, we generalize to the probabilistic setting the
path indicated by \emph{bounded arithmetic} \cite{Buss86, Ferreira90}, a well-known approach to capture polynomial time algorithms, by extending usual arithmetical languages with a distinguished unary predicate $\mathtt{Flip}(x)$, playing the role of a 
source of randomness. %
We define a bounded 
theory $\RS$ as the randomized analogue of Buss'  
$S^1_2$~\cite{Buss86} and Ferreira's $\Sigma^b_1$-NIA~\cite{Ferreira90b},
and show that the functions which can be proved total in $\RS$ are precisely  the polytime \emph{random} functions~\cite{Santos69}, i.e.~those functions from 
strings to \emph{distributions} of strings which can be computed by polytime 
probabilistic Turing machines (PTM, for short).
Then, we move towards proper randomized classes by considering ways to keep the probability of error under control \emph{from within} the logic.
We first consider \emph{measure quantifiers} \cite{Morgenstern, 
MichalewskiMio, ADLP21}, well-studied second-order quantifiers capable of 
measuring \emph{the extent} to which a formula is true; we then show that these quantifiers, when applied to bounded formulas,
can be encoded via \emph{standard} first-order quantification.
This way we obtain two characterizations of the problems in $\BPP$, yet  still semantic in nature:
the error-bound check is translated into conditions 
which are not based on \emph{provability} in some formal system, 
but  rather on the \emph{truth} of some formula in the standard model
of first-order arithmetic. 

While these results, which rely on semantic conditions, do not shed light on the enumeration problem for $\BPP$ directly, they set the conditions for a proof-theoretic investigation of this class: our last contribution is the introduction of a family of new \emph{syntactic} 
subclasses of $\BPP$, each called $\BPP_{\mathsf T}$, and consisting of those languages for which the 
error-bounding condition is not only true, but also \emph{provable} in some (non necessarily bounded) theory $\mathsf T$.
This reduces the enumeration problem to that of finding a recursively enumerable (r.e., for short)~arithmetical theory $\mathsf T$ such that $\BPP=\BPP_{\mathsf T}$.
To witness the difficulty of this problem, we show that 
 the error-bounding condition is \emph{$\Pi^{0}_{1}$-complete} and that  establishing that $\BPP$ \emph{cannot} be enumerated would be at least 
\emph{as hard} as refuting the $\BPP=\mathbf P$ conjecture. At the same time, we show that \emph{polynomial identity testing} (PIT), 
one of the few problems in $\mathbf{BPP}$ 
currently not known to be in $\mathbf{P}$ lies in $\BPP_{\THE}$, where $\THE=\IDE$ is a well-studied \cite{Pudlak1998} sub-theory of $\PA$, thus identifying an interesting and effectively enumerable subclass of $\BPP$.

%
%
%

The main technical contributions of this paper can thus be summarized as follows:
\begin{itemize}
\item We introduce the arithmetical theory $\RS$ and prove that the random functions which are $\Sigma^b_1$-representable in it are precisely those which can be computed in polynomial time.
The proof of the correspondence goes through the definition of a class of \emph{oracle recursive}
functions, called $\POR$,
which is shown equivalent to the class of probabilistic polytime random functions $\PPT$. 
The overall structure of the proof 
is described in Section~\ref{sec:FAtoRC},
while further details can be found
in Appendix~\ref{app:1} and~\ref{app:2}.
\item We exploit this result to obtain a new syntactic characterization of $\PP$ and, more interestingly, two 
semantic characterizations of $\mathbf{BPP}$, the first based  on measure quantifiers and the second relying on standard, first-order 
quantification. This is in Section~\ref{sec:TBPP}.
%
\item
Finally, we introduce a family of syntactic subclasses 
$\BPP_{\mathsf T}\subseteq\mathbf{BPP}$ of \emph{provable} $\BPP$-problems, relative to a theory $\mathsf T$. 
After showing that the property of being non-erratic is $\Pi^{0}_{1}$-complete, we establish that PIT is in $\BPP_{(\IDE)}$.
We conclude by showing how our approach relates to existing works capturing $\BPP$ languages in bounded arithmetic \cite{Jerabek}.
All this can be found in 
Sections~\ref{sec:SSBPP} and~\ref{sec:Jerabek}.
\end{itemize}

\subparagraph*{Related Work}
While a recursion-theoretic characterization of the syntactic class $\mathbf{PP}$ can be found  in~\cite{DalLagoKahleOitavem21},
most existing characterizations of $\mathbf{BPP}$ are based on some  external, semantic condition~\cite{DalLagoParisenToldin,Scedrov}.
In particular,  
Eickmeyer and Grohe \cite{Grohe} provide a semantic characterization of $\BPP$ in a logic with fixed-point operators and a special counting quantifier, 
associated with a probabilistic semantics not too different from the quantitative interpretation we present in Section~\ref{sec:FAtoRC}.
On the other hand, \cite{Jerabek} and \cite{JerabekAPAL} 
uses bounded arithmetic to 
provide characterizations of (both syntactic and semantic) randomized classes, such as $\mathbf{ZPP}, \mathbf{RP}$ and $\mathbf{coRP}$, and also provides a semantic characterization of $\BPP$. An in-depth comparison is thus in order, and can be found in Section~\ref{sec:Jerabek}. 
Finally, \cite{Scedrov} defines a higher-order language for polytime oracle recursive functions based on an adaptation of Bellantoni-Cook's safe recursion.

\section{On the Enumeration of Complexity 
Classes}\label{sec:semantic}\label{section2}
Before delving into the technical details, it is worth spending a few 
words on 
the problem of enumerating complexity classes, and on 
the reasons why  it is 
more difficult for semantic classes than for syntactic ones.

First of all, it is worth observing that, although the distinction between syntactic and semantic classes appears in many popular textbooks (e.g. in~\cite{AroraBarak,Papadimitriou}), in the literature these
notions are not defined in a precise way. 
Roughly speaking, syntactic classes are those which can be defined via limitations on the \emph{amount of resources} (i.e.~units of either time or space) that the underlying algorithm is allowed to use. Typical examples are the class $\mathbf{P}$ of problems solvable in polynomial time and the class $\mathbf{PSPACE}$ of problems solvable in polynomial space.
Instead, the definition of a semantic class usually requires, beyond some resource condition, an additional condition, sometimes called a \emph{promise}, typically expressing that the underlying algorithm
returns the correct answer \emph{often enough}. A typical example here is the class $\mathbf{BPP}$ considered in this paper (cf.~Def.~\ref{def:bpp}), corresponding to problems solvable in polynomial time by probabilistic algorithms with some fixed error bound strictly smaller than $\frac{1}{2}$.
Sometimes the distinction between syntactic and semantic classes may be subtle. For instance, as we discuss in Section \ref{sec:TBPP}, the class $\mathbf{PP}$, whose definition also comprises a promise, is generally considered a syntactic class.

Notice that the sense of the terms ``syntactic'' and ``semantic'', as referred to complexity classes, is not 
clearly related to the sense that these terms have in mathematical logic. 
To a certain extent, the analysis that we develop in this paper with the tools of bounded arithmetic may help clarify this point. On the one hand, well-known results in bounded arithmetic (cf.~\cite{Buss86, Buss98}) provide a characterization of syntactic classes like $\mathbf{P}$ in terms of purely \emph{proof-theoretic} conditions (i.e.~provability in some weak fragment of Peano Arithmetic); on the other hand, we establish that, for a semantic class like $\mathbf{BPP}$, an arithmetical characterization can be obtained by employing \emph{both} proof-theoretic and \emph{model-theoretic} conditions (i.e.~truth in the standard model of Peano Arithmetic).

A natural question is whether such genuinely semantic (i.e.~model-theoretic) conditions can somehow be eliminated in favor of purely syntactic (i.e.~proof-theoretic) ones. In fact, this is a non-trivial problem, since, as proved in Section \ref{sec:SSBPP} (cf.~Proposition \ref{prop:pi01}), the promise underlying $\mathbf{BPP}$ is expressed by a \emph{$\Pi^{0}_{1}$-complete} arithmetical formula. 
One should of course recall, however, that the distinction between semantic and syntactic classes 
refers to \emph{how} a class is defined
and not to the underlying set of problems. It is thus of \emph{intensional} nature. In other words, even if $\mathbf{P}$ and $\mathbf{BPP}$ are defined in a different way, it could well be that someday we discover that $\mathbf{P}=\BPP$: in this case $\mathbf{BPP}$ would become a syntactic class, and, as we show in Section \ref{sec:SSBPP} (cf.~Proposition \ref{prop:p=bpp}), a purely proof-theoretic characterization of $\mathbf{BPP}$ would be available.

The problem of showing that a complexity class can be enumerated (i.e.~that one can devise a recursive enumerations of, say, Turing Machines solving all and only the problems in the class) provides a different, and useful, angle to look at the distinction between syntactic and semantic classes. 
Ordinary syntactic classes, such as $\mathbf{P}$, $\mathbf{PP}$, and 
$\mathbf{PSPACE}$, are quite simple to enumerate.  
%
While verifying resource bounds for \emph{arbitrary} programs is 
very difficult, it is surprisingly easy to define an enumeration of resource 
bounded algorithms containing at least \emph{one} algorithm for any problem in 
one of the aforementioned classes. 
To clarify what we mean, suppose we want to 
characterize the class $\mathbf{P}$.
On the one hand, the class of \emph{all} algorithms working in polynomial time 
is recursion-theoretically very hard, actually 
$\Sigma^0_2$-complete~\cite{Hajek}. 
On the 
other hand, the class of those programs consisting of a $\mathtt{for}$ loop 
executed a polynomial number of times, whose body itself consists of 
conditionals and simple enough instructions manipulating string variables, is 
both trivial to enumerate and big enough to characterize $\mathbf{P}$, at least in 
an extensional sense: every problem in this class is decided by \emph{at 
least one} program in the class and every algorithm in this class works in polytime. Many characterizations of 
$\mathbf{P}$ (and of other syntactic classes), as those based on 
safe-recursion~\cite{BellantoniCook,Leivant}, light and soft linear 
logic~\cite{GirardLafont,Girard,Lafont}, and bounded arithmetic~\cite{Buss86},
can be seen as instances of the just described pattern, 
where the precise class of 
polytime \emph{programs} varies, while the underlying class of \emph{problems} remains unchanged.

But what about semantic classes? 
Being resource bounded is not sufficient for an algorithm to solve a problem in some semantic class, since there can well be algorithms getting it wrong too often. For instance, it may well be that some  probabilistic Turing Machine running in polynomial time does not solve \emph{any} problem in $\mathbf{BPP}$. 
For this reason, unfortunately,  
the  enumeration strategy sketched above does not seem to be readily applicable to semantic classes.
How can we isolate a simple enough subclass of algorithms – which are not only resource bounded, but also not too erratic – at the same time saturating the class?

We think that the results in this paper, concerning proof-theoretic and model-theoretic characterizations of probabilistic complexity classes, may provide new insights on the nature of 
this problem, without giving a definite answer. 
Indeed, observe that the existence of a purely proof-theoretic characterization of some complexity class $\mathcal C$ via some recursively enumerable theory $\mathsf T$ directly leads to providing an enumeration of $\mathcal C$ (by enumerating the theorems of $\mathsf T$). In this way, the problem of enumerating a semantic class $\mathcal C$ is directly related to the existence of some strong enough theory $\mathsf T$. 

In the following sections we do \emph{not} prove $\BPP$ 
to be effectively enumerable, 
which is still out of reach.
On the one hand we
show that proving the non-enumerability of $\BPP$ is \emph{as hard as} proving that $\mathbf P\neq \BPP$. 
On the other hand, we 
show that there exist subclasses of $\BPP$ which are large enough to include interesting problems in $\BPP$ and still ``syntactic enough'' to be effectively enumerable via some arithmetical theory. 


\section{Bounded Arithmetic and Polytime Random Functions}\label{sec:FAtoRC}

\label{sec:LogicAndFunctions}


%

%

In this section we discuss our first result, namely, the characterization of polytime random functions via bounded arithmetic.

\subsection{From Arithmetic to Randomized Computation, Subrecursively}

We introduce the two main ingredients on which our characterization of polytime random functions relies:
a randomized bounded theory of arithmetic $\RS$, 
and a Cobham-style function algebra for polytime oracle recursive functions, 
called $\POR$.

\subparagraph{Recursive Functions and Arithmetical Formulas}
The study of so-called {bounded theories of arithmetic}, i.e.~subsystems of $\PA$ in which only \emph{bounded quantifications} are admitted, initiated by Parikh and Buss, has led to characterize several complexity classes~\cite{Parikh, Cook75, Buss86, Buss98, Ferreira90, KrajicekPudlakTakeuti}.
At the core of these characterizations lies the well-known fact (dating back to G\"odel's \cite{Godel}) that recursive functions can be \emph{represented} in $\PA$ by means of $\Sigma^{0}_{1}$-formulas,
i.e.~formulas of the form $\exists x_{1}.\dots. \exists x_{n}.A$, 
where $A$ is a bounded formula.
For example, the formula
\[
A(x_{1},x_{2},y) := \exists x_{3}. x_{1}\times x_{2} = x_{3}  \land y = \mathit{succ}(x_{3})
\]
represents the function $f(x_{1},x_{2})=(x_{1} \times x_{2})+1$. 
Indeed, in $\PA$ one can prove that 
$\forall x_{1}.\forall x_{2}.\exists!y. A(x_{1},x_{2},y)$, namely that  $A$ expresses a \emph{functional} relation, and check that for all $n_{1},n_{2},m\in \Nat$,
$A(\overline{n_{1}},\overline{n_{2}}, \overline{m})$ holds (in the standard model $	\mathscr{N}$) precisely when
$m=f(n_{1}, n_{2})$.
Buss' intuition was then that, by considering theories \emph{weaker} than $\PA$, 
it becomes possible to capture functions computable within given resource bounds~\cite{Buss86,Buss98}.

In order to extend this approach to classes of \emph{random} computable functions, we rely on a simple correspondence between first-order predicates over  natural numbers and \emph{oracles} from the Cantor space $\{0,1\}^{\Nat}$, following~\cite{ADLP21}.
Indeed, suppose the aforementioned recursive function $f$ has now the ability to observe
(part of) an infinite sequence 
of bits. 
For instance, $f$ might observe the first bit and return
$(x_{1}\times x_{2})+1$ if this is 0, and return $0$ otherwise. 
Our idea is that we can capture the call by $f$ to the oracle 
by adding to the standard language of $\PA$ a new unary predicate $\Flip(x)$,
to be interpreted as a stream of (random) bits. 
Our function $f$ can then be represented by the following formula:
\[
B(x_{1},x_{2},y) :=  \big(\Flip(\overline{0})\land \exists x_{3}. x_{1}\times x_{2} = x_{3}  \land y = succ(x_{3})\big) \lor \big(\lnot \Flip(\overline{0})  \land y= \overline 0\big)
\]
As in the case above, it is possible to prove that $B(x_{1},x_{2},y)$ is functional, that is, that $\forall x_{1}.\forall x_{2}.\exists !y. B(x_{1},x_{2},y)$. 
However, since $B$ now contains the unary predicate symbol $\Flip(x)$, the actual numerical function that $B$ represents depends on the choice of an interpretation for $\Flip(x)$, i.e.~on the choice of an oracle for $f$.

In the rest of this section we develop this idea in detail, establishing a correspondence between polytime random functions and a class of \emph{oracle-recursive} functions which are provably total in a suitable bounded theory relying on the predicate $\Flip$.

%

%
%

\subparagraph*{The Language $\Lpw$.}

We let $\Bool:=\{\zzero,\oone\}$, 
$\Ss:=\Bool^{*}$ indicate the set of finite words from $\Bool$,
and $\Os:= \Bool^{\Ss}$. 
We introduce a language for first-order arithmetic 
incorporating the new predicate symbol $\Flip(x)$ and its interpretation in the standard model. Following~\cite{FerreiraOitavem}, we consider a first-order signature for natural numbers \emph{in binary notation}. 
Consistently, formulas will be interpreted over  $\Ss$ rather than $\Nat$.
Working with strings is not essential
and all results below could be spelled out in a language for 
natural numbers.
Indeed, bounded theories may be formulated in both
ways equivalently, e.g.~Ferreira's $\Sigma^b_1$-NIA
and Buss' $S^1_2$~\cite{FerreiraOitavem}.

%
%

%


\begin{defn}
  \label{def:lpwterms}
The terms and formulas of $\Lpw$ are defined by the grammars below:
\begin{align*}
t,s& ::= x \midd \epsilon \midd \zero \midd \one
\midd t \conc s \midd t\times s \\
F, G& ::= \Flip(t) \midd t=s \midd t\subseteq s
\midd \neg F \midd F\wedge G \midd F\vee G
\midd
 \exists x.F \midd
\forall x.F.
\end{align*}
\end{defn}
The function symbol $\conc$ stands for string concatenation, while $t\times u$ indicates the concatenation of $t$ with itself
a number of times corresponding to the length of $u$. 
The binary predicate $\subseteq$ stands for the 
initial substring relation.
As usual, we let $A\to B:= \lnot A \lor B$.



%

We adopt the following abbreviations: $ts$ for
$t\conc s$; 
 $\one^t$ for $\one\times t$;
    $t\preceq s$ for $\one^t\subseteq \one^s$, i.e.~the length of $t$ is
   smaller than that of $s$; 
    $t|_r=s$ for  $
      (\one^r\subseteq \one^t
      \wedge s\subseteq t \wedge \one^r=\one^s)
      \vee (\one^t\subseteq \one^r \wedge s=t) 
      $, i.e~$s$ is the \emph{truncation}
    of $t$ at the length of $r$.
    For each string $\sigma \in \Ss$, we let $\overline{\sigma}$ be the term of $\Lpw$ representing it (e.g.~$\overline{\epsilon}=\epsilon$, $\overline{\sigma\zzero}=\overline{\sigma}\zero$ and 
    $\overline{\sigma\oone}=\overline{\sigma}\one$).
      
%
%
As for standard bounded arithmetics~\cite{Buss86,Ferreira88},
a defining feature of our theory is the focus on so-called \emph{bounded quantifications}.
In $\Lpw$, \emph{bounded quantifications} 
are of the forms
$\forall x. \one^x \subseteq \one^t \rightarrow F$ and
$\exists x.\one^x \subseteq \one^t \wedge F$,
abbreviated as $\forall x \preceq t.F$ and $\exists x\preceq t.F$.
Following~\cite{Ferreira88},
we adopt \emph{subword quantifications} as those quantifications of the forms
$\forall x.(\exists w\subseteq t.wx\subseteq t)\to F$ and
$\exists x.\exists w\subseteq t.wx \subseteq t \land F$, 
abbreviated as
$\forall x\subseteq^{*}t.F$ and $\exists x\subseteq^{*}t.F$.
An $\Lpw$-formula $F$ is said to be a \emph{$\Sigma^{b}_{1}$-formula} if it is of the form
$\exists x_{1}\preceq t_{1}.\dots. \exists x_{n}\preceq t_{n}.G$, where the only quantifications in $G$ are {subword} ones.
The distinction between bounded and subword quantifications is relevant for complexity reasons:
 if $\sigma\in\Ss$ is a string of length $k$,
the witness of a subword existentially quantified formula
$\exists y.y\subseteq^{*} \overline{\sigma} \land H$ is to be looked for among all possible sub-strings of 
$\sigma$, i.e.~within a space of size $\mathcal O(k^2)$,
while the witness of a bounded formula
$\exists y\preceq \overline{\sigma}.H$ is to be looked for among all possible strings of length $k$, 
i.e.~within a space of size $\mathcal O(2^{k})$.


\subparagraph*{The Borel Semantics of $\Lpw$.}


We introduce a \emph{quantitative} semantics for formulas
of $\Lpw$, inspired by~\cite{ADLP21}.
While the function symbols of $\Lpw$, as well as the predicate symbols ``$=$'' and ``$\subseteq$'', have a standard interpretation as relations over $\Ss$, the idea is that predicate symbol $\Flip$ may stand for \emph{an arbitrary} subset of $\Ss$,
that is, an arbitrarily chosen $\omega\in \Os$. For this reason, we take as the interpretation of a 
$\Lpw$-formula $F$ the set $\model{F}\subseteq \Os$ of \emph{all} possible interpretations of $\Flip$ satisfying $F$.
 Importantly, such sets $\model{F}$ can be proved to be  \emph{measurable}, a fact that will turn out essential in Section \ref{sec:TBPP}.
Indeed, the canonical first-order model of $\Lpw$ over $\Ss$ 
can be extended to a probability space $(\Os, \sigma(\mathscr{C}), \mu)$ defined in a standard way: here $\sigma(\mathscr{C})\subseteq \wp(\Os)$ is the Borel $\sigma$-algebra generated by \emph{cylinders}
$\mathsf C_{\sigma}^{b}=\{\omega \mid \omega(\sigma)={b}\}$, with ${b}\in\{0,1\}$, 
and $\mu$ is the \emph{unique} measure such that $\mu(\mathsf C_{\sigma}^{b})=\frac{1}{2}$ (see~\cite{Billingsley}). 
%
While the interpretation of terms is standard,
the interpretation of formulas is defined below.
%
%
%
%

\begin{defn}[Borel Semantics of $\Lpw$]\label{def:quantsem}
Given a term $t$, a formula
$F$ and an environment $\xi:\mathcal{G}\rightarrow
\Ss$, where $\mathcal{G}$
is the set of term variables,
the \emph{interpretation of F under $\xi$}
is the measurable set
$\llbracket F\rrbracket_\xi\in \sigma(\mathscr{C})$
inductively defined as follows:

\vskip-4mm
\hskip-11mm
\adjustbox{scale=0.9}{
\begin{minipage}[t]{0.2\linewidth}
\begin{align*}
\llbracket t=s\rrbracket_\xi &:= \begin{cases}
\Os \ \ &\text{if } \llbracket t\rrbracket_\xi = \llbracket s\rrbracket_\xi \\
\emptyset \ \ &\text{otherwise}
\end{cases} \\
\llbracket t\subseteq s\rrbracket_{\xi} &:=
\begin{cases}
\Os \ \ \ &\text{if } \llbracket t\rrbracket_\xi
\ssubseteq \llbracket s\rrbracket_\xi \\
\emptyset \ \ \ &\text{otherwise}
\end{cases}
\end{align*}
\end{minipage}
\begin{minipage}[t]{0.3\linewidth}
\begin{align*}
\llbracket \Flip(t)\rrbracket_{\xi} &:= \{\omega \ | \
\omega(\llbracket t\rrbracket_{\xi}) = \oone\} \\
\llbracket \neg G\rrbracket_\xi &:=
\Os - \llbracket G\rrbracket_\xi \\
\llbracket G\vee H\rrbracket_\xi &:= \llbracket G\rrbracket_\xi \cup \llbracket H\rrbracket_\xi \\
\llbracket G\wedge H\rrbracket_\xi &:=
\llbracket G\rrbracket_\xi \cap \llbracket H\rrbracket_\xi
%
\end{align*}
\end{minipage}
\begin{minipage}[t]{0.23\linewidth}
\begin{align*}
%
\llbracket \exists x.G\rrbracket_\xi &:=
\bigcup_{i\in\Ss} \llbracket G\rrbracket_{\xi\{x\leftarrow i\}} \\
\llbracket \forall x.G\rrbracket_\xi &:=
\bigcap_{i\in \Ss}\llbracket G\rrbracket_{\xi\{x\leftarrow i\}}.
\end{align*}
\end{minipage}
}

\end{defn}
\noindent
 This semantics is well-defined as
the sets $\llbracket \Flip(t)\rrbracket_\xi$,
$\llbracket t=s\rrbracket_\xi$
and $\llbracket t\subseteq s\rrbracket_\xi$
are measurable
and measurability is preserved by all
the logical operators.

Observe that an interpretation of the language $\Lpw$, in the usual first-order sense, requires some $\xi$ as above as well as an interpretation $\omega$ for $\Flip(x)$. 
One can easily check by induction that, for any formula $F$ and interpretation $\xi$, $\omega\in \model{F}_{\xi}$ precisely when $F$ is satisfied in the first-order environment formed by $\xi$ \emph{and} $\omega$.

\subparagraph*{The Bounded Theory $\RS$.}
We now introduce a bounded theory in the language $\Lpw$, called $\RS$,
which can be seen as a probabilistic counterpart to Ferreira's $\Sigma^b_1$-NIA \cite{Ferreira90b}.
%
The theory $\RS$ is defined by
axioms belonging to two classes:
\begin{itemize}
\itemsep0em
\item \emph{Basic axioms} (where $\bool\in \{\zero, \one\}$): \\

\vskip-7mm
\hskip-7mm
\begin{minipage}{0.17\textwidth}
\begin{align*}
x\epsilon& =x   \\
x(y\bool)&=(xy)\bool
\end{align*}
\end{minipage}
\begin{minipage}{0.2\textwidth}
\begin{align*}
x\times \epsilon&=\epsilon\\
x\times y\bool&=(x\times y)x
\end{align*}
\end{minipage}
\begin{minipage}{0.2\textwidth}
\begin{align*}
x\subseteq \epsilon & \leftrightarrow x=\epsilon \\
x\subseteq y\bool &\leftrightarrow x\subseteq y \lor x= y\bool
\end{align*}
\end{minipage}
\begin{minipage}{0.2\textwidth}
\begin{align*}
x\bool&=y\bool \rightarrow x=y \\
x\zero&\neq y\one \ \ \ \ 
x\bool\neq \epsilon
\end{align*}
\end{minipage}

\medskip

\item \emph{Axiom schema for induction on
notation}:
$
B(\epsilon) \wedge \forall x.\big(B(x)
\rightarrow B(x\zero)\wedge B(x\one)\big)
\rightarrow \forall x.B(x),
$
where $B$ is a $\Sigma^b_1$-formula
in $\Lpw$.
\end{itemize}
\noindent
The axiom schema for induction on notation adapts the usual induction schema of $\PA$ to the binary representation. As standard in bound arithmetic, restriction 
to $\Sigma^{b}_{1}$-formulas, 
is essential to characterize algorithms with \emph{bounded} resources. 
Indeed, more general instances of this schema would lead to represent functions which are not polytime computable.

%
%
%
%
%
%

%
%
%

\subparagraph{An Algebra of Polytime Oracle Recursive Functions}\label{sec:POR}

We now introduce a Cobham-style function algebra, called $\POR$, for polytime \emph{oracle} recursive functions, and show that it is captured by 
a class of bounded formulas provably representable
in the theory $\RS$. 
This algebra is inspired by Ferreira's 
{PTCA}~\cite{Ferreira88,Ferreira90b}. 
Yet, a fundamental difference is that the functions we define are of the form $f: \Ss^{k}\times \Os\to \Ss$, i.e.~they carry an additional argument
$\omega:\Ss \to \Bool$, to be interpreted as the underlying stream of random bits.
Furthermore, our class includes the basic  \emph{query} function,
 which can be used to observe any bit from $\omega$.
%
%

%

The \emph{class $\POR$} is the smallest class of
functions from
$\Ss^k \times \Os$ to $\Ss$, containing the \emph{empty} function $E(x,\omega)=\eepsilon$, the \emph{projection} functions $P^n_i(x_1,\dots, x_n,\omega)=x_i$, the \emph{word-successor}
function $\Sf_{b}(x,\omega)=x\bbool$, the \emph{conditional} function 
$\Cf(\eepsilon, y, z_0,z_1,\omega) = y$ and 
$\Cf(x\bbool, y,z_0, z_1,\omega) = z_b$,
where $\bbool\in \Bool$ (corresponding to $b\in\{0,1\}$), 
the \emph{query} function $\query(x,\omega)=\omega(x)$,
and closed under the following schemata:
\begin{itemize}
\item \emph{Composition}, where $f$
is defined from $g,h_1,\dots, h_k$
as
$
f(\vec{x},\omega)=g(h_1(\vec{x},\omega), \dots,
h_k(\vec{x},\omega),\omega).
$
\item \emph{Bounded recursion on notation},
where $f$ is defined from
$g,h_0,h_1$ as
\begin{align*}
f(\vec{x},\eepsilon, \omega) &= g(\vec{x},\omega); \\
f(\vec{x}, y\zzero,\omega) &= h_0\big(\vec{x},y,f(\vec{x},y,\omega),\omega\big)|_{t(\vec{x},y)}; \\
f(\vec{x},y\oone,\omega) &= h_1\big(\vec{x},y,f(\vec{x},y,\omega),\omega\big)|_{t(\vec{x},y)},
\end{align*}
with $t$ obtained from
$\eepsilon,\zzero,\oone,\cconc,
\ttimes$ by explicit definition, i.e.~by applying
 $\cconc$ and $\ttimes$ on  constants
 $\eepsilon$, $\zzero$, $\oone$, and  variables $\vec x$ and $y$.
\end{itemize}

We now show that  functions of $\POR$
are precisely those which are $\Sigma^{b}_{1}$-representable in $\RS$.
To do so, we slightly modify Buss' representability conditions
by adding a constraint relating the quantitative semantics
of formulas in $\Lpw$ and the additional functional parameter
$\omega$ of oracle recursive functions.

\begin{defn}\label{def:representability}
A function $f: \Ss^{k} \times \Os \to \Ss$ is \emph{$\Sigma^{b}_{1}$-representable in $\RS$} if there exists a $\Sigma^{b}_{1}$-formula $G(\vec x,y)$ of $\Lpw$ such that:
\begin{enumerate}
\item $\RS \vdash \forall \vec x.\exists! y.G(\vec x,y)$,

\item for all $\sigma_{1},\dots, \sigma_{k},\tau\in \Ss$ and $\omega\in \Os$,
$f(\sigma_{1},\dots, \sigma_{k},\omega)=\tau$ iff $\omega \in \model{G(\overline{\sigma_1}, \dots \overline{\sigma_k},
\overline \tau)}$.
\end{enumerate}
\end{defn}

Condition 1. of Definition~\ref{def:representability} does \emph{not} say that the unique value $y$ is obtained as a function of $\vec x$ \emph{only}. Indeed,
 the truth-value of a formula 
 depends both on the value of its first-order variables and on the value assigned to the random predicate $\Flip$. 
 Hence this condition says that $y$ is uniquely determined as a function \emph{both} of its first-order inputs and of an oracle from $\Os$, precisely as functions of $\POR$. 

\begin{theorem}\label{thm:RStoPOR}
For any $f: \Ss^{k} \times \Os \to \Ss$, $f$ is $\Sigma^{b}_{1}$-representable in $\RS$ iff $f\in \POR$.

\end{theorem}
\begin{proof}[Proof sketch]
$(\Leftarrow)$
The desired $\Sigma^{b}_{1}$-formula is constructed by induction on the
structure of oracle recursive functions.
Observe that the formula $\forall \vec x.\exists! y.G(\vec x,y)$ occurring in Condition 1.~of Definition~\ref{def:representability} is \emph{not} $\Sigma^{b}_{1}$,
since it is universally quantified while the existential quantifier is not bounded.
Hence, in order to apply the inductive steps (corresponding to functions defined by composition and bounded recursion on notation),
we need to adapt Parikh's theorem~\cite{Parikh}
(which holds for $S^{1}_{2}$ and $\Sigma^b_1$-NIA) to $\RS$, to state that if $\RS \vdash \forall \vec x.\exists y.G(\vec x,y)$, where $G(\vec x,y)$ is a
$\Sigma^{b}_{1}$-formula, then we can find a term $t$ such that
$\RS\vdash \forall \vec{x}. \exists y\preceq t.G(\vec x,y)$.
$(\Rightarrow)$
The proof consists in adapting 
Cook and Urquhart's argument for system IPV$^{\omega}$~\cite{CookUrquhart}, and
this goes through a \emph{realizability interpretation} of the
intuitionistic version of $\RS$, called $I\RS$.
Further details can be found in the Appendix~\ref{app:1}. 
\end{proof}

%

\subsection{Characterizing Polytime Random Functions}\label{sec:CPRF}
Theorem \ref{thm:RStoPOR} shows that it is possible to characterize $\POR$ by means of a system of bounded arithmetic. 
Yet, this is not enough to deal with classes, like $\BPP$ or $\RP$,
which are defined in terms of functions computed by PTMs.
%
Observe that there is a crucial difference in the way in 
which probabilistic machines
and oracle recursive functions access randomness,
so our next goal is to fill the gap,
by relating these classes of functions.

Let $\Dist{\Ss}$ indicate the set of \emph{distributions over $\Ss$}, that is, those functions $\lambda: \Ss \to [0,1]$ such that $\sum_{\sigma\in \Ss }\lambda(\sigma)= 1$.
By a \emph{random function} we mean a function of the form
$f: \Ss^k \to \Dist{\Ss}$. Observe that any (polytime) PTM $\PTM$ computes a random function $f_{\PTM}$, where, for every
$\sigma_{1},\dots, \sigma_{k},\tau\in \Ss$,
$f_{\PTM}(\sigma_{1},\dots, \sigma_{k})(\tau)$ coincides with the probability that $\PTM(\sigma_{1}\sharp\dots \sharp \sigma_{k})\Downarrow \tau$.
However, a random function needs not be computed by a PTM in general. 
We define the following class of \emph{polytime random functions}:
\begin{defn}[Class $\PPT$]
The class $\PPT$ is made of all random functions $f: \Ss^{k}\to \Dist{\Ss}$ such that $f=f_{\PTM}$, for some PTM $\PTM$ running in polynomial time.
%
%
\end{defn}
%
\noindent
Functions of
$\PPT$ are closed under \emph{monadic composition} $\diamond$, where
$(g\diamond f)(\sigma)(\tau)=\sum_{\rho\in \Ss }g(\rho)(\tau)\cdot f(\sigma)(\rho)$ (one can check $f_{\PTM'}\diamond f_{\PTM}= f_{\PTM'\circ \PTM}$, where $\circ$ indicates PTM composition).

Since functions of $\PPT$ have a different shape from 
those of $\POR$,
we must adapt the notion of $\Sigma^{b}_{1}$-representability for them,  relying on the fact that any closed $\Lpw$-formula $F$ generates a \emph{measurable} set $\model{F}\subseteq \Bool^{\Nat}$.

\begin{defn}\label{def:rfprepresentable}
A function $f: \Ss^{k} \to \Dist{\Ss}$ is \emph{$\Sigma^{b}_{1}$-representable in $\RS$} if there exists a $\Sigma^{b}_{1}$-formula $G(\vec x,y)$ of $\Lpw$ such that:
\begin{enumerate}
\item $\RS \vdash \forall \vec x.\exists! y.G(\vec x,y)$,
\item for all $\sigma_{1},\dots, \sigma_{k},\tau\in \Ss$,
$f(\sigma_{1},\dots, \sigma_{k},\tau)= \mu\big( \model{G(\overline{\sigma_1}, \dots,
\overline{\sigma}_k, \overline \tau)}\big)$.
\end{enumerate}
\end{defn}

Notice that any $\Sigma^{b}_{1}$-formula $G(\vec x, y)$ satisfying Condition 1.~from Definition~\ref{def:rfprepresentable} actually defines a random function $\rfp{G}:\Ss\to \Dist{\Ss}$ given by $\rfp{G}(\vec \sigma)(\tau)=
\mu(\model{G(\overline{\vec \sigma}, \overline\tau)})$, where $\rfp{G}$ is $\Sigma^{b}_{1}$-represented by $G$.
Moreover, if $G$ represents some $f\in \PPT$, then $f=\rfp{G}$.
In analogy with Theorem \ref{thm:RStoPOR}, we can now prove the following result:

\begin{theorem}\label{thm:RStoPPT}
For any $f: \Ss^{k}  \to \Dist{\Ss}$, $f$ is $\Sigma^{b}_{1}$-representable in $\RS$ iff $f\in \PPT$.

\end{theorem}
\noindent
Thanks to 
Theorem \ref{thm:RStoPOR}, the proof of the result above 
simply consists in showing that $\POR$ and $\PPT$ can be related as stated below.

\begin{lemma}\label{lemma:Fundamental}
For all functions $f: \Ss^{k}\times \Os \to \Ss$ in $\POR$ there exists $g:\Ss^{k}\to \Dist{\Ss}$ in $\PPT$ such that for all $\sigma_{1},\dots, \sigma_{k},\tau\in \Ss$,
$\mu(\{\omega \mid f(\vec \sigma,\omega)=\tau\})=g(\sigma_{1},\dots, \sigma_{k},\tau)$, and conversely.
%
\end{lemma}
\begin{proof}[Proof sketch]
The first step of our proof consists in replacing the class $\PPT$ by an intermediate class $\SFP$ corresponding to functions computed by polytime \emph{stream Turing machines} (STM, for short). 
These are defined as deterministic TM with one extra read-only tape: at the beginning of the computation the extra tape is sampled from $\Bool^{\Nat}$, and at each computation step the machine reads one new bit from this tape. 
Then we show that for any function $f: \Ss^{k}\to \Dist{\Ss}$ computed by some polytime PTM there is a function $g:\Ss^{k}\times \Bool^{\Nat}\to \Ss$ computed by a polytime STM such that for all $\sigma_{1},\dots, \sigma_{k},\tau\in\Ss$,
and $\eta\in\Bool^\Nat$,
 $f(\sigma_{1},\dots, \sigma_{k},\tau)=\mu(\{\eta\mid g(\sigma_{1},\dots, \sigma_{k},\eta)=\tau\})$, and conversely.
To conclude, we prove the correspondence between the classes $\POR$ and $\SFP$:
\begin{description}

\item[($\SFP\Rightarrow\POR$)] 
The encoding relies on the remark that, given an input $x\in \Ss$ and an extra-tape $\eta\in \Bool^{\Nat}$, an STM $\STM$ running in polynomial time can only access a \emph{finite} portion of $\eta$, bounded by some polynomial $p(|x|)$. This way the behavior of $\STM$ is encoded by a $\POR$-function $h(x,y)$, where the second string $y$ corresponds to $\eta_{p(|x|)}$, and we can define $f^{\sharp}(x,\omega)=h(x, e(x,\omega))$, where  $e:\Ss\times \Os\to \Ss$ is a function of $\POR$ which mimics the prefix extractor $\eta\mapsto\eta_{p(|x|)}$, in the sense that its outputs have the same distributions of all possible $\eta$'s prefixes (yet over $\Os$ rather than $\Bool^{\Nat}$).

\item[($\POR\Rightarrow\SFP$)] 
Here we must consider that these two models not only invoke oracles of different shape,
but also that functions of $\POR$ can manipulate such oracles in a much more liberal way
than STMs. 
Notably, the STM accesses oracle bits in a \emph{linear} way: each bit is used exactly once and cannot be re-invoked. Moreover, at each step of computation the STM queries a new oracle bit, while 
functions of $\POR$ can access the oracle, so to say, \emph{on demand}.
The argument rests then on a chain of simulations, making use of a class of imperative languages inspired by Winskell's IMP \cite{Winskell}, each one taking care of one specific oracle access policy: first non-linear and
on-demand (as for $\POR$), then linear but still on-demand, and finally linear and not on-demand (as for STMs).
\end{description}
\end{proof}

\section{Semantic Characterizations of $\BPP$}\label{sec:TBPP}

We now turn our attention to randomized complexity classes. 
This requires us to consider how random functions (and thus PTMs) may correspond to languages, i.e.~subsets of $\Ss$.
The language computed by a random function can naturally be defined via a majority rule:
\begin{defn}
Let $f: \Ss\to \Dist{\Ss}$ be a random function. The language $\LANG{f}\subseteq\Ss$ is defined by 
$\sigma \in \LANG{f}$ iff $f(\sigma)(\epsilon)>\frac{1}{2}$.
%
%
\end{defn}

It is instructive to first take a look at the case of the class $\PP$, recalled below:

\begin{defn}[$\PP$]\label{def:pp}
Given a language $L\subseteq \Ss$, $L\in\PP$ iff there is a polynomial time PTM 
$\PTM$ such that for any $\sigma \in \Ss$,
$\mathrm{Pr}[\PTM(\sigma)=\chi_L(\sigma)] > \frac{1}{2}$, 
where, 
$\chi_{L}:\Ss \to \{0,1\}$ is the characteristic function of $L$.
\end{defn}
At first glance, $\PP$ might be considered a semantic class, since its definition comprises \emph{both} a resource condition and a promise. However, $\PP$ is generally considered a syntactic class, due to the fact that, when trying to capture the machines solving languages in $\PP$, the promise condition can actually be eliminated. Indeed, \emph{any} PTM $\PTM$ running in polynomial time recognizes \emph{some} language in $\PP$, namely the language 
$L=\LANG{f}$, where $f$ is the polytime random function computed by $\PTM$.
Furthermore, the class $\PP$ \emph{can} be enumerated (see e.g.~\cite{DalLagoKahleOitavem21}).

Using Theorem \ref{thm:RStoPPT}, the remarks above readily lead to a proof-theoretic characterization of $\PP$ via $\RS$.
\begin{proposition}[Syntactic Characterization of $\PP$]\label{prop:pp}
For any language $L\subseteq \Ss$, $L\in \PP$ iff there is a $\Sigma^{b}_{1}$-formula $G( x, y)$ such  that: 
\begin{enumerate}
\item $\RS\vdash \forall x.\exists !y.G(x,y)$,
\item $L= \LANG{\rfp{G}}$.
\end{enumerate}
\end{proposition}

The characterization above provides an enumeration of $\PP$ (by enumerating the pairs made of a formula $G$ and 
a proof in $\RS$ of Condition 1). However, while a majority rule is enough to capture the problems in $\PP$, the definition of a semantic class like $\BPP$ requires a different condition. 

\begin{defn}[$\BPP$]\label{def:bpp}
Given a language $L\subseteq \Ss$, $L\in\BPP$ iff there is a polynomial time PTM 
$\PTM$ such that for any $\sigma \in \Ss$,
$\mathrm{Pr}[\PTM(\sigma)=\chi_L(\sigma)] \ge \frac{2}{3}$.
\end{defn}
\noindent
The class $\BPP$ can be captured by ``non-erratic'' probabilistic algorithms, i.e.~such that, for a fixed input, one possible output is definitely more likely than the others.

\begin{defn}
A random function $f: \Ss \to \Dist{\Ss}$ is \emph{non-erratic} if for all $\sigma\in \Ss$, $f(\sigma)(\tau)\geq \frac{2}{3}$ holds for some value $\tau\in \Ss$.
\end{defn}

\begin{lemma}\label{lemma:nonerratic}
For any language $L\subseteq \Ss$, $L\in\BPP$ iff $L=\LANG{f}$, for some non-erratic random function 
$f\in \PPT$.
\end{lemma}
\begin{proof}
For any non-erratic $\PPT$-function $f$, let $\PTM$ be the PTM computing ${k\diamond f}$, where $k(\epsilon)=1$ and $k(\sigma\neq \epsilon)=0$; then $\PTM$ computes $\chi_{\LANG{f}}$ with error $\leq \frac{1}{3}$. 
Conversely, if $L\in \BPP$, 
let $\PTM$ be a PTM accepting $L$ with error $\leq \frac{1}{3}$; then
$L=\LANG{h\diamond f_{\PTM}}$, where $h(1)=\epsilon$ and $h(\sigma\neq 1)=0$.
\end{proof}
\noindent
Lemma \ref{lemma:nonerratic} suggests that, in order to characterize $\BPP$ in the spirit of Proposition \ref{prop:pp}, a new condition has to be added, corresponding to the fact that  $G$ represents a non-erratic random function.
In the rest of this section we discuss two approaches to measure error bounds for probabilistic algorithms, leading to two different characterizations of $\BPP$: first via measure quantifiers~\cite{ADLP21},
then by purely arithmetical means.
While both such methods ultimately consist
in showing that the \emph{truth} of a formula
in the standard model of $\RS$, they also suggest a more proof-theoretic approach, that we explore in Section~\ref{sec:SSBPP}.
%


\subparagraph{$\BPP$ via Measure Quantifiers.}

%
%
%

As we have seen, any $\Lpw$-formula $F$ is associated with a measurable
set $\model{F}\subseteq \Os$.
So, a natural idea, already explored in \cite{ADLP21}, consists in enriching $\Lpw$ with \emph{measure-quantifiers} \cite{Morgenstern, MichalewskiMio}, that is, second-order quantifiers of the form $\BOX^q F$,
where $q\in[0,1]\cap \mathbb{Q}$,
intuitively expressing that the measure of $\model F$ is greater than
(or equal to)
$q$.
%
%
%
Then, 
let $\Lpw^{\textsf{MQ}}$ be the extension of $\Lpw$
with measure-quantified formulas $\BOX^{t/s}F$,
where $t,s$ are terms.
The Borel semantics of $\Lpw$ naturally extends to 
$\Lpw^{\textsf{MQ}}$ letting
$
\model{\BOX^{t/s}F}_{\xi} =\Os$ when $|\model{s}_{\xi}| > 0 $ and $ \mu(\model{F}_\xi)
\ge \frac{|\model{t}_\xi|}{|\model{s}_\xi|}$ both hold, and  
$\model{\BOX^{t/s}F}_{\xi} =\emptyset$ otherwise.
To improve readability, for all $n,m\in \mathbb{N}$,
we abbreviate $\BOX^{\one^n/\one^m} F$ 
 as $\BOX^{n/m}F$. %
%
%
%

Measure quantifiers can now be used to express that the formula representing a random function is non-erratic, as shown below.

\begin{theorem}[First Semantic Characterization of $\BPP$]\label{theorem:char1}
For any language $L\subseteq \Ss$,
$L\in \BPP$ iff there is a $\Sigma^b_1$-formula
$G(x,y)$ such that:  
\begin{enumerate}
\itemsep0em
\item $\RS \vdash \forall x\exists !y.G(x,y)$,
\item $\vDash \forall x.\exists y.\BOX^{2/3}G(x, y)$,
\item $L=\LANG{\rfp{G}}$.
\end{enumerate}
\end{theorem}
\begin{proof}
Let $L\in\BPP$ and $g:\Ss\to \mathbb{D}(\Ss)$
be a function of $\RFP$ 
computing $L$ with uniform error-bound
(which, thanks to Lemma \ref{lemma:nonerratic}, we can suppose to be non-erratic).
By Theorem~\ref{thm:RStoPPT},
there is a $\Sigma^b_1$-formula $G(x,y)$ such that $g= \rfp{G}$.
So, for all $\sigma\in \Ss$, 
$ \mu(\model{G(\overline \sigma, \overline \tau)})=g(\sigma)(\tau)\geq \frac{2}{3}$ holds for some $\tau\in \Ss$, which shows that Condition 2.~holds.
Conversely, if Conditions 1.-3.~hold, then $\rfp{G}$ computes $L$ with the desired error bound, so $L\in \BPP$.
%
\end{proof}

\subparagraph{Arithmetizing Measure Quantifiers.}
Theorem~\ref{theorem:char1} relies on the tight correspondence between arithmetic and probabilistic computation;
yet, Condition 2.~involves formulas which are not in the language of first-order arithmetic. 
Lemma~\ref{lemma:RLexp} below shows that measure quantification over
bounded formulas of $\Lpw$ can be expressed 
arithmetically.

\begin{lemma}[De-Randomization of Bounded Formulas]\label{lemma:RLexp}
For any $\Sigma^{b}_{1}$-formula $F(\vec{x})$ of $\Lpw$,
there exists a  $\Flip$-free $\Pi^{0}_{1}$-formula $\mathsf{TwoThirds}[F](\vec{x})$ such that for any $\vec{\sigma}\in\Ss$,
$
\vDash \mathsf{TwoThirds}[F](\overline{\vec{\sigma}})$ holds iff $\mu(\model{F(\overline{\vec{\sigma}})}) \ge \frac{2}{3}$.
\end{lemma}
\begin{proof}
%
First, observe that for any bounded $\Lpw$-formula 
$F(\vec{x})$, strings $\vec \sigma$ and $\omega\in\Os$,
to check whether $\omega\in\model{F(\overline{\vec{\sigma}})}$ 
only a \emph{finite} portions of bits of $\omega$ has to be observed.
More precisely, we can construct a $\Lpw$-term $t_{F}(\vec{x})$ such that for any $\vec{\sigma}\in \Ss$ and $\omega,\omega'\in\Os$,
if $\omega$ and $\omega'$ agree on all strings shorter than $t_{F}(\vec{\sigma})$,
then $\omega \in \model{F(\sigma)}$ iff $\omega'\in \model{F(\sigma)}$.
Now, all finitely many relevant bits $\omega(\tau)$, for $|\tau|\leq t_{F}(\vec{\sigma})$ can be encoded as a \emph{unique} string $w$ of length $\leq 2^{|t_{F}(\vec{\sigma})|}$ where the bit $w_{i}$ corresponds to the value $\omega(\tau)$, where $\tau$ is obtained by stripping the right-most bit from the binary representation of $i$.
We obtain in this way a $\Flip$-free formula $F^{*}(\vec x, y)$ such that measuring $\model{F(\sigma)}$ corresponds to 
\emph{counting} the strings $y$ of 
length $\leq 2^{|t_{F}(\vec{\sigma})|}$ making $F^{*}(\vec x, y)$ true, i.e.~to showing
\begin{align}\label{eq:star}
\Big\vert \left \{ \tau\preceq  \mathsf 2^{|t_{F}(\vec \sigma)|} \ \Big \vert \   F^{*}(\vec \sigma,\tau)\right\}\Big \vert \geq 
\frac{2}{3}\cdot N 
\tag{$\star$}
\end{align}
where $\mathsf 2^{\eepsilon}=\oone$ and $\mathsf 2^{\sigma b}=\mathsf 2^{\sigma}\mathsf 2^{\sigma}$ is an exponential function on strings and 
$N=2^{\left(2^{|t_{F}(\sigma)|}\right)}$ is the total amount of the strings to be counted.
\eqref{eq:star} can be encoded in a standard way yielding a bounded formula $F^{\sharp}(\vec x)$ in the language of arithmetic \emph{extended} with the function symbol $\mathsf 2^{x}$. Finally, the function symbol $\mathsf 2^{x}$ can be eliminated using a $\Delta^{0}_{0}$-formula $\mathrm{exp}(x,y)$ defining the exponential function (see \cite{Gaifman1982}), yielding a $\Flip$-free $\Pi^{0}_{1}$-formula of $\Lpw$ of the form $\forall z_{1}.\dots. \forall z_{k}. \mathrm{exp}(t_{1},z_{1})\land \dots \land \mathrm{exp}(t_{k},z_{k}) \to F^{\sharp}(\vec x, z_{1},\dots, z_{k})$.
\end{proof}
\begin{remark}\label{remark:beyond}
It is important to observe at this point that the elimination of $\Flip$ via counting 
takes us \emph{beyond} the usual machinery of bounded arithmetic, since we employ some operation which is not polytime. 
This is indeed not surprising, since the counting problems associated with polytime problems (generating the class ${\sharp \mathsf P}$) are not even known to belong to the polynomial hierarchy $\mathsf{PH}$ (while, by Toda's theorem, we know that $ \mathsf{PH}\subseteq \mathsf P^{\sharp \mathsf P}$).
\end{remark}

\noindent
Theorem~\ref{theorem:char1} and Lemma~\ref{lemma:RLexp}
together yield a purely arithmetical characterization of $\BPP$.
Let $\mathsf{NotErratic}[G]$ indicate the arithmetical formula $\forall x.\exists y\preceq 0. \mathsf{TwoThirds}[G](x,y)$.

\begin{theorem}[Second Semantic Characterization of $\BPP$]\label{def:bpp2}
\label{thm:characterization2}
For any language $L\subseteq \Ss$, $L\in\BPP$
when there is a $\Sigma^b_1$-formula $G(x,y)$ such that:
\begin{enumerate}
\itemsep0em

\item $\RS\vdash \forall x.\exists ! y.G(x,y)$,
\item $\vDash \mathsf{NotErratic}[G]$,
\item $L=\LANG{\rfp{G}}$.
\end{enumerate}
\end{theorem}

\section{Provably $\BPP$ Problems}\label{sec:SSBPP}
The characterization provided by Theorem \ref{thm:characterization2} is still 
semantic in nature, as it provides no way to effectively enumerate $\BPP$: 
 the crucial Condition 2~is not checked within a formal 
system, but over the standard model of $\Lpw$. Yet, since the  
condition is now expressed in purely arithmetical terms, it makes sense to 
consider \emph{syntactic} variants of Condition 2, where the model-theoretic check  
is replaced by provability in some sufficiently expressive theory.

We will work in extensions of $\RSE$, 
where $\mathrm{Exp}=\forall x.\exists y. \mathrm{exp}(x,y)$ is the formula expressing the totality of the exponential function (which is used in the de-randomization of Lemma \ref{lemma:RLexp}).
 This naturally leads to the 
following definition:

\begin{defn}[Class $\BPP_{\THE}$]\label{def:bppt}
\label{defn:BPPT}
	Let $\THE\supseteq \RSE$ be a theory in the language $\Lpw$.
	The \emph{class $\BPP$ relative to $\THE$},
	denoted $\BPP_{\THE}$, contains all languages $L\subseteq 
	\Ss$ such that
	for some $\Sigma^{b}_{1}$-formula $G(x,y)$ the following hold:
	\begin{enumerate}
		\item $\RS\vdash \forall x.\exists! y. G(x,y)$,
		\item $ \THE \vdash \mathsf{NotErratic}[G]$,

		\item $L=\LANG{\rfp{G}}$.

	\end{enumerate}
	
\end{defn}
\noindent
Whenever $\THE$ is sound (i.e.~$\THE\vdash F$ implies that $F$ is true in the standard model), it is clear that $\BPP_{\THE}\subseteq \BPP$. 
However, a crucial difference between the \emph{syntactic} class $\BPP_{\THE}$ and the semantic class $\BPP$ is that, 
when $\THE$ is recursively enumerable, $\BPP_{\THE}$ can be 
\emph{enumerated} (by enumerating the proofs of Condition 1. and 2. in $\THE$).
Hence, the  enumerability problem for $\BPP$ translates 
into the question whether one can find a sound r.e.~theory $\THE$ such that $\BPP_{\THE}=\BPP$.
Let us first observe that the relevance of this problem is tightly related to the question $\BPP=\mathbf P$:
\begin{proposition}\label{prop:p=bpp}
If $\BPP=\mathbf P$, then there exists a r.e.~theory $\THE$ such that 
$\BPP=\BPP_{\THE}$.
\end{proposition}
\begin{proof}
If $\BPP=\mathbf P$, and $L\in \BPP$, then there is a
polytime deterministic TM $\mu$ accepting it. $\mu$ yields then a PTM $\mu^{*}$ in a trivial way. Since the corresponding formula $G$ of $\Lpw$ does not contain $\Flip$, $\mathsf{NotErratic}[G]$ can be proved in e.g.~$\RSE$. 
\end{proof}

The counter-positive of the result above is even more interesting, as it says that establishing that \emph{no} r.e.~theory $\THE$ is such that $\BPP_{\THE}=\BPP$ is \emph{at least as hard as} establishing that $\BPP \neq \mathbf{P}$. 
Yet, without knowing whether $\BPP=\mathbf P$, how hard may it be to find a theory $\THE$ such that $\BPP=\BPP_{\THE}$?

Observe that, when $G$ is $\Sigma^{b}_{1}$, $\mathsf{NotErratic}[G]$ is expressed by a $\Pi^{0}_{1}$-formula: as $\mathsf{TwoThirds}[G](\vec{x})$ is of the form
$\forall \vec z. \wedge_{i} \mathrm{exp}(t_{i},z_{i}) \to F^{\sharp}(\vec x, \vec z)$, the condition is expressed by the $\Pi^{0}_{1}$-formula
$\forall x.\forall \vec z. \wedge_{i} \mathrm{exp}(t_{i},z_{i}) \to \exists y\preceq \zero.F^{\sharp}(\vec x, \vec z)$. Hence, if we us fix some recursive enumeration $(\PTM_{n})_{n\in \mathbb N}$ of polytime PTM as well as a recursive coding $\sharp \PTM$ of such machines as natural numbers, the fact that $\PTM$ is non-erratic is expressed by some $\Pi^{0}_{1}$-formula $\varphi_{\mathsf{NotErratic}}(\sharp\PTM)$. The $\Pi^{0}_{1}$-set $\mathsf{NotErratic}= \{ e \mid \varphi_{\mathsf{NotErratic}}(e)\}
$ indicates then the sets of codes corresponding to non-erratic machines.

The possibility of finding a theory strong enough to prove \emph{all} positive instances of Condition 2 is then ruled out by the following result.

\begin{proposition}\label{prop:pi01}
$\mathsf{NotErratic}$ is $\Pi^{0}_{1}$-complete.
\end{proposition}
\begin{proof}
We reduce to $\mathsf{NotErratic}$ the $\Pi^{0}_{1}$-complete problem 
$\mathsf{HALT}_{n^{2}}$ consisting of codes of TM halting in time at most $n^{2}$ (see \cite{Gajser2016}). With any TM $\mu$ associate a polytime PTM $\mu^{*}$ that,  on input $x$, yields $\mathrm{TRUE}$ with prob.~$\frac{1}{2}$, and otherwise simulates $\mu(x)$ on $|x|^{2}$ steps, yielding $\mathrm{TRUE}$ if the computation of $\mu(x)$ terminated, and $\mathrm{FALSE}$ otherwise. Then it is easily seen that $\mu\in 
\mathsf{HALT}_{n^{2}}$ iff $\mu^{*}\in \mathsf{NotErratic}$.
\end{proof}
\begin{corollary}
Being (the code of) a PTM solving some $\BPP$-problem is $\Sigma^{0}_{2}$-complete.
\end{corollary}
\begin{proof}
As we say, for a PTM, solving some $\BPP$-problem is equivalent to being polytime and non-erratic. 
Being the code of a polytime (P)TM is a $\Sigma^{0}_{2}$-complete property \cite{Hajek}.
By Proposition~\ref{prop:pi01}, checking non-erraticity does not increase the logical complexity.
\end{proof}

Proposition \ref{prop:pi01} implies that for any consistent theory $\THE$ one can always find some non-erratic polytime PTM whose non-erraticity is \emph{not provable} in $\THE$.
Indeed, since $\mathsf{NotErratic}$ is $\Pi^{0}_{1}$-complete, we can reduce to it the $\Pi^{0}_{1}$-set of codes of \emph{consistent} r.e.~theories. Hence, if $\THE$ is some consistent theory such that for \emph{any} code $e\in \mathsf{NotErratic}$, $\THE$ proves
$\varphi_{\mathsf{NotErratic}}(e)$, then $\THE$ can prove \emph{all} $\Pi^{0}_{1}$-statement expressing the consistency of some consistent r.e.~theory, and thus, in particular, the one expressing its own consistency, contradicting (Rosser's variant of) G\"odel's second incompleteness theorem. 

Observe that this result suggests that the enumerability problem \emph{might} be very difficult, but it
 \emph{does not} provide a negative answer to it.
  Indeed, recall that what we are interested in is not an enumeration of \emph{all} 
  non-erratic polytime PTM,
  but an enumeration containing \emph{at least one} machine for each problem in $\BPP$. 
In other words, the question remains open whether, for any non-erratic polytime PTM, 
it is possible to find a machine solving the same problem but whose non-erratic behavior can be proved in some fixed theory $\THE$.
While we do not know the answer to this question, we can still show that a relatively weak arithmetical theory is capable of 
proving the non-erraticity of a machine solving one of the (very few) problems in $\BPP$ which are currently not known to be in $\mathbf P$.

%
%
%

\section{Polynomial Zero Testing is Provably $\BPP$}

In this section we establish that PIT is in 
$\BPP_{(\IDE)}$. We recall that $\IDE$ is the fragment of Peano Arithmetics with induction restricted to \emph{bounded formulas}, together with the totality of the exponential function. 
\begin{remark}
While $\IDE$ is a theory in the usual language of $\PA$, here we work in a language for binary strings. 
Indeed, what we here call $\IDE$ is actually the corresponding theory $\mathsf{\Delta}_{0}\text{-NIA}+ \mathrm{Exp}$, 
formulated for the language $\Lpw$ \emph{without} $\Flip$, and defined as $\Sigma^b_1\text{-NIA} + \mathrm{Exp}$ with induction extended to \emph{all} bounded formulas, plus the axiom $\mathrm{Exp}$. 
Based on~\cite{FerreiraOitavem} $\mathsf{\Delta}_{0}\text{-NIA}$ corresponds to Buss' theory $S_{2}$, which, in turn, is known to correspond to $\mathsf{I\Delta}_{0}+\Omega_{1}$, indeed a sub-theory of $\IDE$. 
%
\end{remark}

The PIT
problem asks to decide the identity of the polynomial computed by two arithmetical circuits. These are basically DAGs whose nodes can be labeled so as to denote an input, an output, the constants $0, 1$ or an arithmetic operation. These structures can easily be encoded, e.g. using lists, as terms of $\Lpw$.   

\begin{defn}[cf.~\cite{AroraBarak}]
  The problem PIT asks to decide 
  whether two arithmetical circuits $p, q$ encoded as lists of nodes describe
  the same polynomial, i.e. $\mathbb Z \models p=q$. 
\end{defn}

Usually, PIT is reduced to another problem: the so-called Polynomial Zero Testing (PZT) problem, which asks to decide whether a polynomial computing a circuit over $\mathbb Z$ is zero, i.e. to check whether $\mathbb Z \models p=0$. Indeed, $\mathbb Z \models p=q$ if and only if $\mathbb Z \models p-q =0$.
%
%
Our proof of the fact that the language PZT is in $\BPP_{(\IDE)}$ is structured as follows:
\begin{itemize}
\item We identify a $\Sigma^b_1$-formula $G(x,y)$ of $\Lpw$ 
characterizing the polytime algorithm $\texttt{PZT}$ from~\cite{AroraBarak}, and we turn it into a $\Flip$-free formula $G^{*}(x,y,z)$ as in Lemma \ref{lemma:RLexp}, where the variable $z$ stands for the source of randomness;

\item We identify  a $\Flip$-free $\Delta^{0}_{0}$-formula $H(x, y)$ which represents the na\"ive deterministic algorithm for PZT. 
\item 
We show that $\IDE$ proves a statement showing that the formulas 
$G^{*}$ and $H$ are equivalent in at least $\frac{2}{3}$ of all (finitely many) relevant values of $z$. In other words, we establish 
$\IDE\vdash \forall x.\forall y.\MEAS[ G(x,y)\leftrightarrow H_{}(x,y)]$. 
\end{itemize}
From the last step, since the totality of $H$ is provable in $\IDE$, we can deduce 
$\IDE \vdash \forall x.\exists y. \MEAS[G](x,y)$, as  required in Definition \ref{defn:BPPT}.

Each of the aforementioned steps will be described in one of the forthcoming paragraphs, although the details are discussed in the Appendix~\ref{sec:appPIT}.
%
\subparagraph*{The Randomized Algorithm.}
\label{par:ABalgo}

Our algorithm for PZT takes an input $x$, which encodes a circuit $p$ of size $m$ on the variables $v_1, \ldots, v_n$, it draws $r_1, \ldots, r_n$ uniformly at random from $\{0, \ldots, 2^{m+3}-1\}$ and $k$ from $\{1, \ldots, 2^{2m}\}$, then it computes the value of $p(r_1, \ldots, r_n) \bmod k$, so to ensure that during the evaluation no overflow can take place. This is done linearly
many times in $|x|$ (we call this value $s$), as to ensure that, if the polynomial is not identically zero, the probability to evaluate $p$ on values witnessing this property  at least once grows over $\frac 2 3$.
%
Finally,  
if all the evaluations returned $0$ as output
the input is accepted; otherwise, it is rejected. 

The procedure 
described above is correct only when the size of the input circuit $x$ is 
greater than some constant $\varrho$. If this is not the case, our algorithm 
queries a table $T$ storing all the pairs $(x_i, \chi_{\text{PZT}}(x_i))$ for 
$|x_i| < \varrho$, to obtain $\chi_{\text{PZT}}(x_i)$. The table $T$ can be 
pre-computed, having just a constant number of entries.
%
%
%
This algorithm, which we call \texttt{PZT}, is inspired by~\cite{AroraBarak} and described in detail in Appendix~\ref{sec:appABalgo}.

As the input circuit is evaluated modulo some $k\in \mathbb Z$, the algorithm works in time polynomial with respect to $|x|$. Therefore, as a consequence of Theorem \ref{thm:RStoPOR} and Lemma \ref{lemma:Fundamental}, there is a $\Sigma^b_1$-formula $G(x,y)$ of $\Lpw$ that represents it. 
Appendix~\ref{sec:appABalgo} also contains a lower-level description of this formula $G$. 
\subparagraph*{The Underlying Language.}
We show that there is a predicate $H$ of $\Lpw$ such that $H(x, \epsilon)$ holds if and only if $x$ is the encoding of a circuit in PZT; otherwise, $H(x, \zzero)$ holds. 
This predicate realizes the function $h$ described by the following algorithm:
\begin{enumerate}
\item Take in input $x$, and check whether it is a polynomial circuit with one output; if it is not, reject it. Otherwise:
\item Compute the polynomial term $p$ represented by $x$, and reduce it to a normal form $\overline p$.
\item Check whether all the coefficients of the terms are null. If this is true, output $\epsilon$, otherwise output $\oone$ and terminate.
\end{enumerate}
For reasonable encodings of polynomial circuits and expressions, $h$ is elementary recursive and therefore there is a predicate $H$ which characterizes it, and $\IDE$ proves the totality of $h$. Moreover,
we have 
 $h =\chi_\text{PZT}$,
 as for every polynomial $p$ with coefficients in $\mathbb Z$, $\mathbb Z \models \forall \vec x. p(\vec x) = 0$ iff all the monomials in the normal form of $p$ have zero as coefficient.
%
%
%
\subparagraph*{Proving the Error Bound.}
We now show that the formula $G$ is not-erratic and that it decides $\LANG{\rfp{G}}$. 
With the notations $G^{*}$ and $t_{G}$ from the the proof of Lemma \ref{lemma:RLexp}, 
this can be reduced to proving in $\IDE$ the following two claims:
%
\begin{align}
  \label{eq:claim5.1}\tag{$\dag$}
    &\vdash\forall z. |z|=t_{G}(x) \land G^{*}(x, \zzero , z) \rightarrow H(x,\zzero ),\\
      \label{eq:claim5.2}
 &\vdash \forall x. \Big \vert \left \{z \preceq \mathsf 2^{t_{G}(x))} \ \Big \vert \ G^{*}(x, \epsilon, z) \to H(x,\epsilon))\right\}\Big \vert \ge \frac{2}{3}\cdot \mathsf 2^{|\mathsf 2^{t_{G}(x)}|}. \tag{$\ddag$}
\end{align}
\noindent
\eqref{eq:claim5.1} states that whenever the randomized algorithm rejects an input, then so does the deterministic one, while \eqref{eq:claim5.2}, which is reminiscent of \eqref{eq:star}, 
states that in at least $\frac 2 3$ of all possible cases, if the randomized algorithm accepts the circuit, the deterministic one accepts it too.
Jointly, \eqref{eq:claim5.1} and \eqref{eq:claim5.2} imply that the equivalence 
$G^{*}(x,y,z)\leftrightarrow H(x,y)$ holds in at least $2/3$ of all possible cases.

While Claim \eqref{eq:claim5.1} is a consequence of the compatibility of the $\bmod\ k$ function with addition and multiplication, which are easily proved in $\IDE$, the proof of Claim \eqref{eq:claim5.2} is more articulated and relies on the Schwartz-Zippel Lemma, providing a lower bound to the probability of evaluating the polynomial on values witnessing that it is not identically zero, and the Prime Number Theorem (whose provability in $\IDE$ is known~\cite{CornarosDimitracopoulos}) which bounds the probability to choose a \emph{bad} value for $k$, i.e.~one of those values causing \texttt{PZT} to return the wrong value. Detailed arguments are provided in the Appendix.

\subparagraph*{Closure under Polytime Reduction}

Only assessing that a problem belongs to $\BPP_\THE$ does not tell us anything about other languages of this class; for this reason, we are interested in showing that $\BPP_\THE$ is closed under polytime reduction. 
This allows us to start from $\text{PZT}\in \BPP_{(\IDE)}$ to conclude that all problems which can be reduced to PZT in polynomial time belong to this class, and in particular  that $\text{PIT} \in \BPP_{(\IDE)}$. This is assessed by the following proposition, proved in the Appendix:

\begin{prop}
  \label{prop:bpppaclosure}
  For any theory $\THE\supseteq \RSE$, language $L \in \BPP_\THE$ and language $M\subseteq \Ss$, if there is a polytime reduction from $M$ to $L$, then $M\in \BPP_\THE$.
\end{prop}

\begin{cor}
  PIT is in $\BPP_{(\IDE)}$.
\end{cor}


\section{On Je\v{r}\'abek's Characterization of $\BPP$}
\label{sec:Jerabek}

As mentioned in Section~\ref{section1}, a semantic characterization of $\BPP$ based on bounded arithmetic was already provided by Je\v{r}\'abek in \cite{Jerabek}. 
This approach relies on checking, against the standard model, the truth of a formula which, rather than expressing that some machine is non-erratic, expresses what can be seen as a second totality condition (beyond the formula expressing the totality of the algorithm).
Hence, also within this approach we think it makes sense to investigate which problems can be proved to be in $\BPP$ within some given theory.

In this section, we relate the two approaches by showing that the problems in $\BPP_{\THE}$ are provably definable $\BPP$ problems, in the sense of \cite{Jerabek}, \emph{within some suitable extension} of the bounded theory PV$_{1}$\cite{CookUrquhart}. 



A PTM is represented in this setting by two provably total functions $(A,r)$, where the machine accepts on input $x$ with probability less than $ p/q$ when $\mathrm{Pr}_{w< r(x)}(A(x,w)) \leq p/q$. 
Je\v{r}\'abek focuses on the theory PV$_{1}$, extended with an axiom schema dWPHP (PV$_{1}$)called the \emph{dual weak pigeonhole principle} (cf.~\cite[pp. 962ff.]{Jerabek}) for PV$_{1}$ (i.e.~the axiom stating that for every PV$_{1}$-definable function $f$, $f$ is not a surjection from $x$ to $x^{2}$). The reason is that this theory is capable of proving \emph{approximate} counting formulas of the form $\mathrm{Pr}_{w< r(x)}(A(x,w)) \preceq_{0} p/q$, where ``$\preceq_{0}$'' is a relation equivalent to ``$\leq$'' up to some polynomially small error (recall that, in order to establish \emph{exact} counting results, we were forced to use non-polytime operations, cf.~Remark \ref{remark:beyond}). 
%
%
The representation of $\BPP$ problems hinges on the definition, for any probabilistic algorithm $(A,r)$, of $ L_{A,r}^{+}(x)  := \mathrm{Pr}_{w< r(x)}(\lnot A(x,w))\leq 1/3$ and $ L_{A,r}^{-}(x) :=\mathrm{Pr}_{w< r(x)}(A(x,w))\leq 1/3$.
 Checking if the algorithm $(A,r)$ solves some problem in $\BPP$ reduces then to checking the ``totality'' formula $\vDash\forall x.L_{A,r}^{+}(x) \lor L_{A,r}^{-}(x)$.
%
 
Now, first observe that, modulo an encoding of strings via numbers, everything which is provable in $\RS$ \emph{without} the predicate $\Flip$ can be proved in the theory $S^{1}_{2}(PV)$ \cite{CookUrquhart}, which extends both PV$_{1}$ and Buss' $S^{1}_{2}$.
Moreover, by arguing as in the proof of Lemma \ref{lemma:nonerratic}, in our characterization of $\BPP$ we can w.l.o.g.~suppose that the formula $G$ satisfies
$\mathrm{EpsZero}[G]:=\forall x.\forall y. G(x,y) \to y=\epsilon \lor y=\zzero$. Under this assumption, the de-randomization procedure described in the proof of Lemma \ref{lemma:RLexp} turns $G$ into a pair $(A,r)$, where $A=G^{*}$ is $\Flip$-free and $r(x)=t_{G}(x)$, and the languages $L_{A,r}^{+}(x)$ and $L_{A,r}^{-}(x)$ correspond then to the formulas $
L_{G}^{+}(x):=\MEAS[G(-,\epsilon)](x)$, and 
$L_{G}^{-}(x):=\MEAS[G(-, \zzero)](x)$.

%

%
%
%

Now, since from $\mathsf T\vdash \forall x.\exists y.\MEAS[G](x,y)$ and 
$\mathrm{EpsZero}[G]$ one can deduce  
$\mathsf T\vdash\forall x.L_{G}^{+}(x) \lor L_{G}^{-}(x)$, we arrive at the following:

\begin{proposition}
Let $L$ be a language with $L=\LANG{\rfp{G}}$. If $L\in \BPP_{\THE}$, then $\forall x. L_{G}^{+}(x)\lor L_{G}^{-}(x)$ is provable in some recursively enumerable extension of PV$_{1}$. Conversely, if 
 PV$_{1}$+dWPHP(PV$_{1}$) $\vdash\forall x. L_{G}^{+}(x)\lor L_{G}^{-}(x)$, then $L\in \BPP_{\RSE}$.
\end{proposition}
The second statement above relies on the fact that approximate counting can be replaced by exact counting in $\RSE$ (i.e.~``$\preceq_{0}$'' can be replaced by ``$\leq$'').

%

%



\section{Conclusion}\label{section6}
The logical characterization of randomized complexity classes, in particular those having a semantic nature, is a great challenge. This paper contributes to the understanding of this problem by showing not only how resource bounded randomized computation can be captured within the language of arithmetic, but also that the latter offers convenient tools to control error bounds, the essential ingredient in the definition of classes like $\mathbf{BPP}$ and $\mathbf{ZPP}$.

We believe that the main contribution of this work is a first example of a sort of \emph{reverse} computational complexity for probabilistic algorithms. 
As we discussed in Section \ref{sec:SSBPP}, while the restriction to bounded theories is crucial in order to capture polytime algorithms via a totality condition, it is not necessary  to prove error bounds for probabilistic (even polynomial time) algorithms.
 In particular, the (difficult) challenge of enumerating $\BPP$ translates into the challenge of proving $\BPP=\BPP_{\THE}$ for some strong enough r.e.~theory $\THE$. 
 So, it is worth exploring how much can be proved within expressive arithmetical theories.
For this reason we focused here on a well-known problem, PIT, which is known to be in $\BPP$,
but not in $\mathbf{P}$, showing that 
the whole argument for $\mathrm{PIT}\in \BPP$ can be formalized in a fragment of $\PA$, namely $\IDE$.

\subparagraph*{Future Work}

The authors see this work as a starting point for a long-term study on the logical nature of semantic classes. From this point of view, many ides for further work naturally arise.

An exciting direction is the study of the expressiveness of the
 new syntactic classes 
$\mathbf{BPP}_{\THE}$, that is, an investigation on the kinds of error bounds which can be proved in the arithmetical theories lying \emph{between} standard bounded theories like $S^{1}_{2}$ or PV and $\mathsf{PA}$, but also in theories which are \emph{more expressive} than $\mathsf{PA}$ (like e.g.~second-order theories).
Surely, classes of the form $\BPP_{\THE}$ could be analyzed also as for the existence of complete problems and hierarchy theorems for them, since such results do not hold for $\BPP$ itself \cite{Fortnow2001}.

Our approach to $\BPP$ suggests that extensions to other complexity classes of randomized algorithms like $\mathbf{ZPP}, \mathbf{RP}$ and $\mathbf{coRP}$ could make sense. Notice that this requires to deal not only with beyond error-bounds, but also with either average class complexity or with failure in decision procedures.

%
%
%
Finally, given the tight connections between bounded arithmetics and proof complexity,
another natural direction is the study of applications of our work to randomized variations on the theme,
for example recent investigations on \emph{random resolution refutations}
\cite{Jerabek, Buss2014, Pudlak}, i.e.~resolution systems where proofs may make errors but are correct most of the time.
%

\bibliography{main.bib}

\appendix


\section{Proofs from Section~\ref{sec:POR}}\label{app:1}

\newcommand{\termO}{\textsf{t}}
\newcommand{\termT}{\textsf{u}}
\newcommand{\termF}{\textsf{v}}



$\mathbf{Theorem~\ref{thm:RStoPOR}.\Leftarrow.}$
As anticipated, in order to apply inductive steps
(namely, composition and bounded recursion on notation)
we need to adapt Parikh's theorem~\cite{Parikh} to
$\RS$.\footnote{The theorem is usually presented in
the context of Buss' bounded theories, as stating that given 
a bounded formula $F$ in $\mathcal{L}_{\Nat}$
such that $S^1_2\vdash \forall \vec{x}.\exists y.F$,
then there is a term $t(\vec{x})$ such that also
$S^1_2 \vdash \forall \vec{x}.\exists y\leq t(\vec{x}).F(\vec{x},y)$~\cite{Buss86,Buss98}.
Furthermore, due to~\cite{FerreiraOitavem},
Buss' syntactic proof can be adapted to $\Sigma^b_1$-NIA
in a natural way.
The same result holds for $\RS$,
as  not containing specific rules concerning $\Flip(\cdot)$.}

\begin{prop}[``Parikh''~\cite{Parikh}]\label{prop:Parikh}
Let $F(\vec{x},y)$ be a bounded $\Lpw$-formula 
such that $\RS \vdash \forall \vec{x}.\exists y.F(\vec{x},y)$.
Then, there is a term $t$ such that,
$
\RS \vdash \forall \vec{x}. \exists y\preceq t(\vec{x}). F(\vec{x},y).
$
\end{prop}

\begin{proof}[Proof for Theorem~\ref{thm:RStoPOR}($\Leftarrow$)]
The proof is by induction on the structure of functions in $\POR$.

\emph{Base Case.}
Each basic function is $\Sigma^b_1$-representable in $\RS$.
\begin{itemize}
\itemsep0em
\item The empty function $f=E$ is $\Sigma^b_1$-represented in $\RS$ by
the formula:
$$
F_E(x,y) : x=x \wedge y=\epsilon.
$$
\begin{enumerate}
\itemsep0em

\item Existence is proved considering $y=\epsilon$.
For the reflexivity of identity
both $\RS \vdash x=x$ and $\RS\vdash \epsilon
= \epsilon$ hold.
So, by rules for conjunctions, we obtain
$\RS\vdash x=x \wedge \epsilon=\epsilon$,
and conclude
$
\RS \vdash \forall x.\exists y.(x=x \wedge y=\epsilon).
$
Uniqueness is proved assuming $\RS \vdash x=x
\wedge z=\epsilon$.
By rules for conjunction,
in particular $\RS \vdash z=\epsilon$,
and since $\RS \vdash y=\epsilon$,
by the transitivity of identity,
we conclude
$
\RS \vdash y=z.
$

\item Assume $E(\sigma,\omega^*)=\tau$.
If $\tau=\eepsilon$, then
$\model{\overline{\sigma} = \overline{\sigma}
\wedge \overline{\tau}=\epsilon}$
=
$\model{\overline{\sigma}=\overline{\sigma}}
\cap \model{\overline{\tau}=\epsilon}$
= $\Os \cap \Os$
= $\Os$.
So, for any $\omega^*,\omega^*\in\model{\overline{\sigma}
= \overline{\sigma}\wedge \overline{\tau}=\epsilon}$,
as clearly $\omega^*\in\Os$.
If $\tau\neq \eepsilon$, then
$\model{\overline{\sigma} =\overline{\sigma}
\wedge \overline{\tau}=\epsilon}$
=
$\model{\overline{\sigma} =\overline{\sigma}}
\cap \model{\overline{\tau}=\epsilon}$
= $\Os \cap \emptyset$
= $\emptyset.$
So, for any $\omega^*,\omega^* \not\in\model{\overline{\sigma}
= \overline{\sigma} \vee \overline{\tau}=\epsilon}$,
as clearly $\omega^* \not\in\emptyset$.

\end{enumerate}

\item Functions $f=P^n_i$, 
$f=S_b$ and $f=C$
are $\Sigma^b_1$-represented in $\RS$ by respectively
the formulas:
\begin{align*}
F_{P^n_i}(x_1,\dots, x_n,y) &: \bigwedge_{j\in J} (x_j=x_j) \wedge
y=x_i, \\
F_{S_b}(x,y) &: y=x\bool, \\
F_C(x,v,z_0,z_1,y) &: (x=\epsilon \wedge y=v) \vee 
\exists x' \preceq x.(x=x'\zero \wedge y=z_0) \\
& \ \ \ \ \ \ \ \ \ \ \
\ \ \ \ \ \ \ \ \ \ \  \ \vee \exists x'\preceq x. (x=x'\one \wedge y=z_1).
\end{align*}
where $1\leq i\leq n$, $J=\{1,\dots, n\}\setminus \{i\}$,
%
%
%
and $\bool\in\{\zero,\one\}$ corresponding to (resp.) $b\in\{0,1\}$.
Proofs are omitted as straightforward.
%
%

\item $f=Q$ is $\Sigma^b_1$-represented in $\RS$ 
by the formula:
$$
F_Q(x,y) : (\Flip(x) \wedge y=\one) \vee
(\neg \Flip(x) \wedge y=\zero).
$$
Observe that, in this case, the proof crucially relies on the fact
that oracle functions invoke \emph{exactly one} oracle:
\begin{enumerate}

\itemsep0em
\item Existence is proved by cases.\footnote{For the formal proof, see~\cite{RBA}.}
Since our underlying logic is classical, 
$\RS \vdash \Flip(x) \vee \neg \Flip(x)$ holds.
Then, if $\RS \vdash \Flip(x)$,
let $y=\one$.
By the reflexivity of identity,  
$\RS \vdash (\Flip(x) \wedge \one = \one)
\vee (\neg \Flip(x) \wedge \zero = \one)$.
So, by rules for disjunction,
$\RS \vdash (\Flip(x) \wedge \one = \one)
\vee (\neg \Flip(x) \wedge \zero=\one)$
and we conclude:
$$
\RS \vdash \exists y.\big((\Flip(x) \wedge y=\one) \vee
(\neg \Flip(x) \wedge y=\zero)\big).
$$
If $\RS \vdash \neg \Flip(x)$,
let $y=\zero$. 
By the reflexivity of identity $\RS \vdash \zero=\zero$ holds.
Thus, by the rules for conjunction,
$\RS \vdash \neg \Flip(x) \wedge \zero=\zero$
and for disjunction,
we conclude $\RS\vdash (\Flip(x) \wedge \zero=\one)
\vee (\neg \Flip(x) \wedge \zero=\zero)$
and so,
$$
\RS \vdash \exists y.\big((\Flip(x) \wedge y=\one)
\vee (\neg \Flip(x) \wedge y=\zero)\big).
$$
Uniqueness is established relying on the transitivity
of identity.

\item Finally, it is shown that for every $\sigma,\tau\in\Ss$
and $\omega^* \in \Os$,
$Q(\sigma,\omega^*) = \tau$ when 
$\omega^* \in \model{F_Q(\overline{\sigma},\overline{\tau})}$.
Assume $Q(\sigma,\omega^*)=\oone$,
which is $\omega^*(\sigma)=\oone$,
\begin{align*}
\model{F_Q(\overline{\sigma},\overline{\tau})}
&= 
\model{\Flip(\overline{\sigma}) \wedge \overline{\tau} = \one}
\cup
\model{\neg \Flip(\overline{\sigma})\wedge
\overline{\tau}=\zero} \\
&= (\model{\Flip(\overline{\sigma})}
\cap \model{\one = \one})
\cup
(\model{\neg \Flip(\overline{\sigma})} \cap
\model{\one=\zero}) \\
&= (\model{\Flip(\overline{\sigma})}
\cap \Os) \cup (\model{\neg \Flip(\overline{\sigma})}\cap \emptyset) \\
&=
\model{\Flip(\overline{\sigma})} \\
&= \{\omega \ | \ \omega(\sigma)=\one\}.
\end{align*}
Clearly, $\omega^*\in\model{(\Flip(\overline{\sigma}) \wedge
\overline{\tau}=\one) \vee (\neg \Flip(\overline{\sigma})
\wedge \overline{\tau}=\zero)}$.
The case $Q(\sigma,\omega^*)=\zzero$
and the opposite direction are proved in a similar way.
\end{enumerate}

\end{itemize}

\emph{Inductive Case.}
If $f$ is defined by composition or bounded recursion from
$\Sigma^b_1$-representable functions,
then $f$ is $\Sigma^b_1$-representable in $\RS$:

\begin{itemize}
\itemsep0em

\item \emph{Composition.}
Assume that $f$ is defined by composition from functions
$g,h_1,\dots, h_k$ so that
$
f(\vec{x},\omega)= g(h_1(\vec{x},\omega), \dots,
h_k(\vec{x},\omega),\omega)
$
and that $g,h_1,\dots, h_k$ are represented in $\RS$
by the $\Sigma^b_1$-formulas $F_g,F_{h_1},\dots,
F_{h_k}$, respectively.
By Proposition~\ref{prop:Parikh}, there exist suitable terms
$t_g,t_{h_1},\dots, t_{h_k}$ such that (the existential part of) Condition 1. 
can be strengthened to $\RS \vdash \forall\vec{x}.
\exists y\preceq t_i.F_i(\vec{x},y)$ for each $i\in\{g,h_1,\dots, h_k\}$.
We conclude that $f(\vec{x},\omega)$ is $\Sigma^b_1$-represented
in $\RS$ by the following formula:
\begin{align*}
F_f(x,y) : \exists z_1\preceq t_{h_1}(\vec{x}).
\dots \exists z_k \preceq t_{h_k}(\vec{x}).
(F_{h_1}(\vec{x},z_1) &\wedge  \dots
F_{h_k}(\vec{x},z_k) \\
&\wedge F_{g}(z_1,\dots, z_k,y)).
\end{align*}
Indeed, by IH, $F_g,F_{h_1},\dots, F_{h_k}$ are $\Sigma^b_1$-formulas.
Then, also $F_f$ is in $\Sigma^b_1$.
Conditions 1.-2. are proved to hold by slightly modifying standard proofs.

\item \emph{Bounded Recursion.} Assume that $f$
is defined by bounded recursion from $g,h_0,h_1$, and $t$,
so that:
\begin{align*}
f(\vec{x},\eepsilon,\omega) &= g(\vec{x},\omega)  \\
f(\vec{x}, y\bool, \omega) &= h_i(\vec{x},y,f(\vec{x},y,\omega),
\omega)|_{t(\vec{x},y)},
\end{align*}
where $i\in\{0,1\}$ and $\bbool=\zzero$ when $i=0$
and $\bbool = \oone$ when $i=1$.
Let $g,h_0,h_1$ be represented in $\RS$ by, respectively,
the $\Sigma^b_1$-formulas $F_g,F_{h_0},$ and $F_{h_1}$.
Moreover, by Proposition~\ref{prop:Parikh}, there exist
suitable terms 
$t_g,t_{h_0},$ and $t_{h_1}$ such that
the existential part of condition 1.
can be strengthened to its 
``bounded version''.
Then, it can be proved that $f(\vec{x},y)$ is $\Sigma^b_1$-represented
in $\RS$ by the formula below:
\begin{align*}
F_f(x,y) : \ &\exists v\preceq t_g(\vec{x})t_f(\vec{x})(y\times
t(\vec{x},y) t(\vec{x},y)\one\one).(F_{lh}(v,\one\times y\one) \\
&\wedge \exists z\preceq t_g(\vec{x}).
(F_{eval}(v,\epsilon,z) \wedge F_g(\vec{x},z)) \\
&\wedge \forall u\subset y.\exists z.(\tilde{z} \preceq t(\vec{x},
y))(F_{eval}(v,\one\times u,z) \wedge
F_{eval}(v,\one\times u\times,\tilde{z}) \\
&\wedge (u\zero \subseteq y\rightarrow
\exists z_0\preceq t_{h_0}(\vec{x},u,z).
(F_{h_0}(\vec{x},u,z,z_0) \wedge z_{0}|_{t(\vec{x},u)}=\tilde{z})) \\
&\wedge (u\one \subseteq y \rightarrow
\exists z_1 \preceq t_{h_1}(\vec{x},u,z).(F_{h_1}
(\vec{x},u,z,z_1) \wedge
z_1|_{t(\vec{x},u)} = \tilde{z})))),
\end{align*}
where $F_{lh}$ and $F_{eval}$ are $\Sigma^b_1$-formulas
defined as in~\cite{Ferreira88}.
Intuitively, $F_{lh}(x,y)$ states that the number of $\one$s
in the encoding of $x$ is $yy$,
while $F_{eval}(x,y,z)$ is a ``decoding''
formula (strongly resembling G\"odel's $\beta$-formula),
expressing that the ``bit'' encoded in $x$ as its $y$-th
bit is $z$.
Moreover $x\subset y$ is an abbreviation for $x\subseteq y
\wedge x\neq y$.
Then, this formula $F_f$ satisfies all the requirements to
$\Sigma^b_1$-represent in $\RS$ the function $f$,
obtained by bounded recursion form $g,h_0,$ and $h_1$.
In particular, Condition 1. concerning existence and
uniqueness, have already been proved to hold by  Ferreira~\cite{Ferreira88}.
Furthermore, $F_f$ expresses that, given the desired encoding sequence
$v$: (i.) the $\eepsilon$-th bit of $v$ is (the encoding of)
$z'$ such that $F_g(\vec{x},z')$ holds,
where (for IH) $F_g$ is the $\Sigma^b_1$-formula
representing the function $g$,
and (ii.) given that for each $u\subset y$,
$z$ denotes the ``bit'', encoded in $v$ at position
$\one \times u\one$,
then if $u\bool\subseteq y$
(that is, if we are considering the initial substring of $y$
the last bit of which correspond to $\bool$),
then there is a $z_b$ such that $F_{h_b}(\vec{x},y,z,z_b)$,
where $F_{h_b}$
$\Sigma^b_1$-represents the function $f_{h_b}$
and the truncation of $z_b$ at $t(\vec{x},u)$
is precisely $\tilde{z}$, with $b=0$ when $\bool=\zero$
and $b=1$ when $\bool=\one$.
\end{itemize}
\end{proof}

$\mathbf{Theorem~\ref{thm:RStoPOR}.(\Rightarrow).}$
The proof is obtained by adapting that by Cook and Urquhart 
for $IPV^\omega$~\cite{CookUrquhart},
and is structured as follows:
\begin{enumerate}
\itemsep0em
\item We define $\POR^\lambda$
a basic equational theory for a simply typed $\lambda$-calculus
endowed with primitives corresponding to functions of $\POR$. 

\item We introduce a first-order \emph{intuitionistic}
theory $I\POR^\lambda$, which extends $\POR^\lambda$
with the usual predicate calculus as well as an $\NP$-induction
schema.
It is shown that $I\POR^\lambda$ is strong enough to prove
all theorems of $I\RS$.

\item We develop a realizability interpretation of $I\POR^\lambda$
(inside itself), showing that for any derivation of
$\forall x.\exists y.F(x,y)$ (where $F$ is a $\Sigma^b_0$-formula)
one can extract a $\lambda$-term $\termO$ of $\POR^\lambda$,
such that $\forall x.F(x,\termO x)$
is provable in $I\POR^\lambda$.
From this we deduce that every function which is $\Sigma^b_1$-representable in $I\RS$ is in $\POR$.

\item We extend this result to classical $\RS$
showing that any $\Sigma^b_1$-formula provable in
$I\POR^\lambda$ + Excluded Middle (EM, for short)
is already provable in $I\POR^\lambda$.
\end{enumerate}

\paragraph*{The System $\POR^\lambda$.}
We define an equational theory for a simply typed $\lambda$-calculus
augmented with primitives for functions of $\POR$.
Actually, these do not exactly correspond to the ones of $\POR$,
although the resulting function algebra is proved equivalent.

\begin{defn}
\emph{Types of $\POR^\lambda$} are defined by the grammar
below:
$$
\sigma := s \midd \sigma \arrowT \sigma.
$$
\end{defn}

\begin{defn}\label{df:termsPORl}
\emph{Terms of $\POR^\lambda$} are standard, simply
typed $\lambda$-terms plus the constants:
\begin{align*}
\zeroT,\oneT,\epsilon &: s \\
\circ, \textsf{Trunc} &: s \arrowT s \arrowT s \\
\textsf{Tail}, \textsf{Flipcoin} &: s \arrowT s \\
\textsf{Cond} &: s \arrowT s \arrowT s \arrowT s \arrowT s \\
\textsf{Red} &: s \arrowT (s \arrowT s \arrowT s) 
\arrowT (s \arrowT s \arrowT s)
\arrowT (s \arrowT s)
\arrowT s 
\arrowT s.
\end{align*}
\end{defn}
\noindent
Intuitively, $\textsf{Tail}(x)$
computes the string obtained by deleting the first
digit of $x$;
$\textsf{Trunc}(x,y)$ computes the string obtained
by truncating $x$ at the length of $y$;
$\textsf{Cond}(x,y,z,w)$ computes the function
that yields $y$ when $x=\eepsilon$,
$z$ when $x=x'\zzero$,
and $w$ when $x=x'\oone$;
$\textsf{Flipcoin}(x)$ indicates a random $\zzero/\oone$
generator;
$\textsf{Rec}$ is the operator for bounded recursion on notation.

\begin{notation}
We abbreviate $x\circ y$ as $xy$
and being $\textsf{T}$ any constant $\textsf{Tail},
\textsf{Trunc},\textsf{Cond},\textsf{Flipcoin},\textsf{Rec}$
of arity $n$,
we indicate $\textsf{T}\termT_1,\dots, \termT_n$
as $\textsf{T}(\termT_1,\dots, \termT_n)$.
\end{notation}

We also introduce the following abbreviations for composed
functions:
\begin{itemize}
\itemsep0em
\item $\textsf{B}(x):= \textsf{Cond}(x,\epsilon,\zeroT,\oneT)$
denotes the function computing the last digit of $x$.

\item $\textsf{BNeg}(x):= \textsf{Cond}(x,\epsilon,\zeroT,\oneT)$
denotes the function computing the Boolean negation
of $\textsf{B}(x)$.

\item $\textsf{BOr}(x,y):= \textsf{Cond}(\textsf{B}(x),
\textsf{B}(y), \textsf{B}(y),\oneT)$
denotes the function that coerces $x$ and $y$
to Booleans and then performs the OR operation.

\item $\textsf{BAnd}(x,y) := \textsf{Cond}(\textsf{B}(x),
\epsilon,\zeroT,\textsf{B}(y))$ denotes the function
that coerces $x$ and $y$ to Booleans and then
performs the AND operation.

\item $\textsf{Eps}(x) := \textsf{Cond}(x,\oneT,\zeroT,\zeroT)$
denotes the characteristic function of ``$x=\zero$''.

\item $\textsf{Bool}(x):= \textsf{BAnd}(\textsf{Eps}(\textsf{Tail}(x)),
\textsf{BNeg}(\textsf{Eps}(x)))$
denotes the characteristic function of ``$x=\zero \vee x=\one$''.

\item $\textsf{Zero}(x):= \textsf{Cond}(\textsf{Bool}(x),\zeroT,
\textsf{Cond}(x,\zeroT,\zeroT,\oneT),\zeroT)$ denotes the characteristic
function of predicate ``$x=\zero$''.

\item $\textsf{Conc}(x,y)$ denotes the concatenation
function defined as:
$$
\textsf{Conc}(x,\epsilon) := x \ \ \ 
\ \ \
\ \ \ 
\textsf{Conc}(x,y\textsf{b}) := \textsf{Conc}(x,y)\textsf{b},
$$
with $\textsf{b}\in\{\zeroT,\oneT\}$.

\item $\textsf{Eq}(x,y)$ denotes the characteristic
function of ``$x=y$''
and defined by double recursion by the equations below:
\begin{align*}
\textsf{Eq}(\epsilon,\epsilon) &:= \oneT 
\ \ \ 
\ \ \ 
\textsf{Eq}(\epsilon, y\textsf{b}) := \zeroT \\
\textsf{Eq}(x\textsf{b},\epsilon)
= \textsf{Eq}(x\zeroT,y\oneT)
= \textsf{Eq}(x\oneT,y\zeroT)
&:= \zeroT 
\ \ \
\ 
\textsf{Eq}(x\textsf{b},y\textsf{b})
:= \textsf{Eq}(x,y),
\end{align*}
with $\textsf{b}\in\{\zeroT,\oneT\}$.

\item $\textsf{Times}(x,y)$ denotes the function for
self-concatenation, $x,y\mapsto x\ttimes y$
and is defined by the equations below:
$$
\textsf{Times}(x,\epsilon) := \epsilon 
\ \ \ 
\ \ \ 
\ \ \ 
\textsf{Times}(x,y\textsf{b}) :=
\textsf{Conc}(\textsf{Times}(x,y),x),
$$
with $\textsf{b}\in\{\zeroT,\oneT\}$.

\item $\textsf{Sub}(x,y)$ denotes the
initial substring function,
$x,y \mapsto S(x,y)$, and is defined by bounded
recursion as follows:
$$
\textsf{Sub}(x,\epsilon) := \textsf{Eps}(x)
\ \ \ 
\ \ \ 
\ \ \
\textsf{Sub}(x,y\textsf{b}) := \textsf{BOr}(\textsf{Sub}(x,y),
\textsf{Eq}(x,y\textsf{b})),
$$
with $\textsf{b}\in\{\zeroT,\oneT\}$.
\end{itemize}

\begin{defn}
\emph{Formulas of $\POR^\lambda$}
are equations $\termO=\termT$,
where $\termO$ and $\termT$ are terms of type $s$.
\end{defn}

\begin{defn}[Theory $\POR^\lambda$]
Axioms of $\POR^\lambda$ are the following ones:
\begin{itemize}

\item Defining equations for the constants of $\POR^\lambda$:
\begin{align*}
\epsilon x = x\epsilon &= x \ \ \ 
\ \ \ \ \ \
\ \ \
\ \ \
\ \ \ \ \
x(y\textsf{b}) = (xy)\textsf{b} \\
\textsf{Tail}(\epsilon) &= \epsilon 
\ \ \
\ \ \
\ \ \
\ \ \ \
\ \ \  \ \ 
\textsf{Tail}(x\textsf{b}) = x  \\
\textsf{Trunc}(x,\epsilon) =
\textsf{Trunc}(\epsilon,x) &= \epsilon 
\ \ \ \ \ \
\ \ \ \ \ 
\textsf{Trunc}(x\textsf{b}, y\textsf{b}) = \textsf{Trunc}(x,y)\textsf{b} \\
\textsf{Cond}(\epsilon, y,z,w) &= y 
\ \ \ \ \ \ \ 
\textsf{Cond}(x\zeroT, y, z,w) = z 
\ \ \ \ \ \ \ 
\textsf{Cond}(x\oneT, y,z,w) = w \\
\textsf{Bool}(\textsf{Flipcoin}(x)) &= \oneT \\
\textsf{Rec}(x,h_0,h_1,k,\epsilon) &= x 
\ \ \ \ \ 
\textsf{Rec}(x,h_0,h_1,k,y\textsf{b}) =
\textsf{Trunc}(h_b y(\textsf{Rec}(x,h_0,h_1,k,y)),ky),
\end{align*}
where $\textsf{b}\in\{\zeroT,\oneT\}$
and $b\in\{0,1\}$ (correspondingly).

\item The $(\beta)$- and $(\nu)$-axioms:
\begin{align*}
\textsf{C}[(\lambda x.\termO)\termT] &= 
\textsf{C}[\termO\{\termT/x\}] 
\tag{$\beta$} \\
\textsf{C}[\lambda x.\termO x] &= \textsf{C}[\termO].
\tag{$\nu$} \\
\end{align*}
where $\textsf{C}[\cdot]$ indicates a context with a unique
occurrence of the hole $[ \ ]$,
so that $\textsf{C}[\termO]$ denotes the variable capturing 
replacement of $[ \ ]$ by $\termO$ in $\textsf{C}[ \ ]$.
\end{itemize}
The inference rules of $\POR^\lambda$
are the following ones:
\begin{align*}
\termO = \termT &\vdash
\termO = \termT \tag{$\textsf{R1}$} \\
\termO = \termT,
\termT = \termF
&\vdash 
\termO = \termF
\tag{$\textsf{R2}$} \\
\termO = \termT &\vdash
\termF\{\termO/x\} =
\termF\{\termT/x\}
\tag{$\textsf{R3}$} \\
\termO=\termT &\vdash
\termO\{\termF/x\} =
\termT\{\termF/x\}.
\tag{$\textsf{R4}$}
\end{align*}
\end{defn}
\noindent
As predictable,
$\vdash_{\POR^\lambda}\termO = \termT$
expresses that the equation $\termO=\termT$
is deducible using instances of the axioms above 
plus inference rules $(\textsf{R1})-(\textsf{R4})$.
Similarly, given any set $T$ of equations,
$T\vdash_{\POR^\lambda} \termO = \termT$
expresses that the equation $\termO=\termT$
is deducible using instances of the quoted axioms
and rules together with equations from $T$.

For any string $\sigma\in\Ss$,
let $\ooverline{\sigma}:s$ denote the term of 
$\POR^\lambda$ corresponding to it, that is:
$$
\ooverline{\eepsilon} = \epsilon \ \ \ 
\ \ \ 
\ \ \
\ \ \
\ooverline{\sigma \zzero} = \ooverline{\sigma}
\zeroT \ \ \ 
\ \ \ 
\ \ \ 
\ \ \ 
\ooverline{\sigma \oone} = \ooverline{\sigma}
\oneT.
$$
For any $\omega\in\Os$,
let $T_\omega$ be the set of all equations of the form
$\textsf{Flipcoin}(\ooverline{\sigma})
= \ooverline{\omega(\sigma)}$.

\begin{defn}[Provable Representability]
Let $f:\Os \times \Ss^j \to \Ss$.
A term $\termO:s\arrowT \dots \arrowT s$
of $\POR^\lambda$ \emph{provably represents f}
when for all strings $\sigma_1,\dots, \sigma_j,\sigma\in \Ss$
and $\omega\in\Os$,
$$
f(\sigma_1,\dots, \sigma_j,\omega) \ \ \
\text{ iff } \ \ \ 
T_\omega \vdash_{\POR^\lambda} \termO
\ooverline{\sigma_1}\dots \ooverline{\sigma_j}
= \ooverline{\sigma}.
$$
\end{defn}

\begin{example}\label{ex:flip}
The term $\textsf{Flipcoin} : s \arrowT s$
provably represents the query function 
$Q(x,\omega)=\omega(x)$ of $\POR$,
since for any $\sigma \in \Ss$
and $\omega \in \Os$,
$$
\textsf{Flipcoin}(\ooverline{\sigma}) =
\ooverline{\omega(\sigma)}
\vdash_{\POR^\lambda}
\textsf{Flipcoin}(\ooverline{\sigma})
=\ooverline{Q(\sigma,\omega)}.
$$
\end{example}
We consider some of the terms described
above
and show them to provably represent
the intended functions.
Let $Tail(\sigma,\omega)$
indicate the string obtained by chopping
the first digit of $\sigma$,
and $Trunc(\sigma_1,\sigma_2,\omega)
= \sigma_1|_{\sigma_2}$.

\begin{lemma}
Terms $\textsf{Tail}, \textsf{Trunc}$
and $\textsf{Cond}$
provably represent the functions $Tail$,
$Trunc$ and $C$, respectively.
\end{lemma}

\begin{theorem}\label{thm:provRepr}
\begin{enumerate}
\itemsep0em
\item Any function $f\in\POR$
is provably represented by a term $\termO\in\POR^\lambda$.

\item For any term $\termO \in \POR^\lambda$,
there is a function $f\in\POR$
such that $f$ is provably repesented by $\termO$.
\end{enumerate}
\end{theorem}

\begin{proof}[Proof Sketch]
1. The proof is by induction on the structure of $f\in\POR$.

\emph{Base Case.}
Each base function is provably represented.
Let us consider two examples:
\begin{itemize}
\item $f=E$ is provably represented by
$\lambda x.\epsilon$.
For any string $\sigma\in\Ss$, 
$\ooverline{E(\sigma,\omega)}=\ooverline{\eepsilon}
= \epsilon$ and
$\vdash_{\POR^\lambda} (\lambda x.\epsilon)
\ooverline{\sigma}=\epsilon$ is an instance of
$(\beta)$-axiom.
We conclude, $\vdash_{\POR^\lambda}(\lambda x.\epsilon)
\ooverline{\sigma}=\ooverline{E(\sigma,\omega)}$.

\item $f=Q$ is provably represented by the term 
$\textsf{Flipcoin}$, as observed in Example~\ref{ex:flip} above.
\end{itemize}

\emph{Inductive Case.}
Each function defined by composition or bounded
recursion from provably represented functions
is provably represented as well.
We consider bounded recursion.
Let $f$ be defined as:
\begin{align*}
f(\sigma_1,\dots, \sigma_n,\eepsilon,\omega) &=
g(\sigma_1,\dots, \sigma_n,\omega) \\
f(\sigma_1,\dots, \sigma_n,\sigma\bbool,\omega)
&=
h_b(\sigma_1,\dots, \sigma_n,\sigma, f(\sigma_1,
\dots, \sigma_n,\sigma,\omega), \omega) |_{k(\sigma_1,\dots,
\sigma_n,\sigma)}. \\
\end{align*}
By IH, $g,h_0,h_1$ and $k$ are provably represented
by the corresponding terms 
$\textsf{t}_{g}, \textsf{t}_{h_0},
\textsf{t}_{h_1}, \textsf{t}_k$.
So, for any $\sigma_1,\dots, \sigma_{n+2},\sigma \in \Ss$
and $\omega \in \Os$:
\begin{align*}
T_\omega &\vdash_{\POR^\lambda}
\termO_{\textsf{g}} \ooverline{\sigma_1} \dots
\ooverline{\sigma_n} = \ooverline{g(\sigma_1,\dots,
\sigma_n,\omega)}
\tag{$\textsf{t}_g$} \\
T_\omega &\vdash_{\POR^\lambda} \textsf{t}_{h_0}
\ooverline{\sigma_1} \dots \ooverline{\sigma_{n+2}}
= \ooverline{h_0 (\sigma_1,\dots, \sigma_{n+2},\omega)}
\tag{$\textsf{t}_{h_0}$} \\
T_{\omega} &\vdash_{\POR^\lambda}
\text{t}_{h_1} \ooverline{\sigma_1}\dots
\ooverline{\sigma_{n+2}}
= \overline{h_1(\sigma_1,\dots, \sigma_{n+2},\omega} 
\tag{$\textsf{t}_{h_1}$}  \\
T_\omega &\vdash_{\POR^\lambda}
\textsf{t}_k \ooverline{\sigma_1} \dots 
\ooverline{\sigma_n} = \ooverline{k(\sigma_1,\dots,
\sigma_n,\omega)}.
\tag{$\textsf{t}_k$}
\end{align*}
We prove by induction on $\sigma$,
that $T_\omega \vdash_{\POR^\lambda} \termO_f
\ooverline{\sigma_1}\dots
\ooverline{\sigma_n \sigma}
=\ooverline{f(\sigma_1,\dots, \sigma_n,\omega)}$,
where 
$\termO_f = \lambda x_1\dots \lambda x_n
\lambda x.\textsf{Rec}(\termO_g x_1
\dots x_n, \termO_{h_0} x_1\dots x_n,
\termO_{h_1} x_1\dots x_n,
\termO_{k} x_1\dots x_n,x)$.
Then,
\begin{itemize}
\itemsep0em

\item if $\sigma=\eepsilon$,
then $f(\sigma_1,\dots, \sigma_n,\sigma,\omega)
= g(\sigma_1,\dots, \sigma_n,\omega)$.
Using the $(\beta)$-axiom
we deduce,
$\vdash_{\POR^\lambda} \termO_f\ooverline{\sigma_1}
\dots \ooverline{\sigma_1} \dots
\ooverline{\sigma_n\sigma}
= \textsf{Rec}(\termO_g\ooverline{\sigma_1}\dots
\ooverline{\sigma_n},
\termO_{h_0} \ooverline{\sigma_1} \dots
\ooverline{\sigma_n},
\termO_{h_1}\ooverline{\sigma_1}\dots
\ooverline{\sigma_n},
\termO_{k}\ooverline{\sigma_1}\dots
\ooverline{\sigma_n},\ooverline{\sigma})$
and
using the axiom
$\textsf{Rec}(\termO_g x_1 \ddots x_n,
\termO_{h_0} x_1\dots x_n,
\termO_{h_1} x_1\dots x_n,
\termO_{k} x_1\dots x_n,\epsilon=
\termO_g x_1\dots x_n$,
we obtain
$\vdash_{\POR^\lambda}
\termO_f \ooverline{\sigma_1} \dots
\ooverline{\sigma_n\sigma}=
\termO_g\ooverline{\sigma_1} \dots \ooverline{\sigma_n}$,
by ($\textsf{R2}$) and
($\textsf{R3}$).
We conclude using ($\textsf{t}_g$)
together with ($\textsf{R2}$).

\item $\sigma=\sigma_m\zzero$,
then $f(\sigma_1,\dots,\sigma_n,\sigma,\omega)=
h_0(\sigma_1,\dots, \sigma_n,\sigma_m,f(
\sigma_1,\dots, \sigma_n,\sigma,\omega),\omega)|_{k(\sigma_1,\dots, \sigma_n,\sigma_m)}$.
By IH,
suppose $T_\omega\vdash_{\POR^\lambda}
\termO_f\ooverline{\sigma_1}\dots
\ooverline{\sigma_n\sigma_m}
= \ooverline{f(\sigma_1,\dots, \sigma_n,\sigma',\omega)}$.
Thus, using the $(\beta)$-axiom $\termO_f\ooverline{\sigma_1}
\dots \ooverline{\sigma_n\sigma}=
\textsf{Rec}(\termO_g \ooverline{\sigma} \dots
\ooverline{\sigma_n},
\termO_{h_0} \ooverline{\sigma_1} \dots
\ooverline{\sigma_n},
\termO_{h_1}\ooverline{\sigma_1}\dots
\ooverline{\sigma_n},
\termO_{k} \ooverline{\sigma_1}\dots
\ooverline{\sigma_n}, \ooverline{\sigma}$
the axiom $\textsf{Rec}(g,$ $h_0,h_1,k,x\zeroT)$ =
$\textsf{Trunc}(\termO_{h_0}x(\textsf{Rec}(g,h_0,h_1,k,\zeroT)),
kx)$ and IH we deduce,
$
\vdash_{\POR^\lambda} \termO_f \ooverline{\sigma_1}
\dots \ooverline{\sigma_n\sigma}
= \textsf{Trunc} (\termO_{h_0} \ooverline{\sigma_1}
\dots \ooverline{\sigma_n \sigma_m}
\ooverline{f(\sigma_1,\dots, \sigma_n,\sigma_m,\omega)},
\termO_k\ooverline{\sigma_1}\dots \ooverline{\sigma_n})
$,
by ($\textsf{R2}$) and ($\textsf{R3}$).
Using ($\termO_{h_0}$) and
($\termO_k$) we conclude using
($\textsf{R3}$) and ($\textsf{R2}$):
$
\vdash_{\POR^\lambda} \termO \ooverline{\sigma_1}\dots
\ooverline{\sigma_n\sigma}$
= 
$\ooverline{h_0(\sigma_1,\dots, \sigma_n, \sigma_m,}$
$\ooverline{f(\sigma_1,\dots, \sigma_,\sigma_m,\omega)|_{k(\sigma_1,\dots
\sigma_n,\sigma_m)}}.
$

\item the case $\sigma=\sigma_m\oone$ is proved in a 
similar way.
\end{itemize}
2. It is a consequence of the normalization property
for the simply typed $\lambda$-calculus:
a $\beta$-normal term $\termO: s \arrowT \dots
\arrowT s$ cannot contain variables of higher types.
By exhaustively inspecting  possible normal forms,
representability is checked.
\end{proof}

\begin{cor}
For any function $f:\Ss^j \times \Os \to \Ss$,
$f\in \POR$ when $f$ is provably represented
by some term $\termO:s\arrowT \dots \arrowT
s \in \POR^\lambda$.
\end{cor}

\emph{The Theory $I\POR^\lambda$.}
We introduce a first-order \emph{intuitionistic}
theory $I\POR^\lambda$,
which extends $\POR^\lambda$ with basic predicate calculus
and a restricted induction principle.
We also define $I\RS$ as a variant of $\RS$
having the intuitionistic predicate calculus as its logical basis.
All theorems of $\POR^\lambda$ and $I\RS$
are provable in $I\POR^\lambda$.
In fact, $I\POR^\lambda$ can be seen as an extension
of $\POR^\lambda$ and provides a language to
associate derivations in $I\RS$ with poly-time
computable functions, corresponding to terms
of $I\POR^\lambda$.

The language of $I\POR^\lambda$ extends that of 
$\POR^\lambda$
with (a translation for) all expressions of $\RS$.
In particular, the grammar for terms of $I\POR^\lambda$
is precisely the same as that of Definition~\ref{df:termsPORl},
while that for formulas is defined below.

\begin{defn}
\emph{Formulas of $I\POR^\lambda$}
are defined as follows:
i. all equations of $\POR^\lambda$ $\termO=\termT$,
are formulas of $I\POR^\lambda$;
ii. for any (possibly open) $\POR^\lambda$-term 
$\termO,\termT$, $\termO\subseteq \termT$ and
$\Flip(\termO)$ are formulas of $I\POR^\lambda$;
iii. formulas of $I\POR^\lambda$ are closed under
$\wedge,\vee,\rightarrow,\forall,\exists$.
\end{defn}
We adopt the standard conventions:
$\bot:= \zeroT=\oneT$ and
$\neg F:= F\rightarrow \bot$.
The notions of $\Sigma^b_0$- and $\Sigma^b_1$-formula
of $I\POR^\lambda$ are precisely as those for $\RS$.

\begin{remark}
Any formula of $\RS$ can be seen as a formula
of $I\POR^\lambda$, where
each occurrence of $\zero$ is replaced by $\zeroT$,
of $\one$ by $\oneT$,
of $\conc$ by $\circ$ (usually omitted),
of $\times$ by $\textsf{Times}$.
In the following, we assume that any formula of
$\RS$ is a formula of $I\POR^\lambda$,
modulo the substitutions defined above.
\end{remark}

\begin{defn}
The axioms of $I\POR^\lambda$ include
standard rules of the intuitionistic first-order predicate
calculus, usual rules for the equality symbol,
plus the following axioms:
1. all axioms of $\POR^\lambda$,
2. $x\subseteq y \leftrightarrow \textsf{Sub}(x,y)=\oneT$,
3. $x=\epsilon \vee x=\textsf{Tail}(x)\zeroT \vee
x=\textsf{Tail}(x)\oneT$,
4. $\zeroT=\oneT \rightarrow x=\epsilon$,
5. $\textsf{Cond}(x,y,z,w) = w' \leftrightarrow
(x=\epsilon \wedge w'=y) \vee
(x=\textsf{Tail}(x) \zeroT \wedge w'=z) \vee
(x=\textsf{Tail}(x)\oneT \wedge w'=w)$,
6. $\Flip(x) \leftrightarrow \textsf{Flipcoin}(x)=\oneT$,
7. any formula of the form:
$$
(F(\epsilon) \wedge \forall x.(F(x) \rightarrow F(x\zero))
\wedge \forall x.(F(x) \rightarrow F(x\one)))\rightarrow
\forall y.F(y),
$$
where $F$ is of the form
$\exists z \preceq \termO.\termT=\termF$,
with $\termO$ containing only first-order open variables. 
\end{defn}

\begin{notation}
We refer to a formula of the form
$\exists z\preceq \termO.\termT = \termF$,
with $\termO$ containing only first-order open variables,
as an \emph{$\NP$-predicate}.
\end{notation}

Now that $I\POR^\lambda$ has been introduced we show
that all theorems of both $\POR^\lambda$ and 
the intuitionistic version of $\RS$ are derived in it.
First, Proposition~\ref{prop:PORtoIPOR}
is established inspecting all rules of $\POR^\lambda$.

\begin{prop}\label{prop:PORtoIPOR}
Any theorem of $\POR^\lambda$ is a theorem of 
$I\POR^\lambda$.
\end{prop}
\noindent
Then, we show that every theorem of $I\RS$
is derivable in $I\POR^\lambda$.
To do so, we prove a few properties concerning
$I\POR^\lambda$.
In particular, its recursion schema differs from
that of $I\RS$ as dealing with formulas of the
form $\exists y\preceq \termO.\termT=\termF$
and not with all the $\Sigma^b_1$-ones.
The two schemas are related by Proposition~\ref{prop:IPORprop}
proved by induction on the structure of formulas.

\begin{prop}\label{prop:IPORprop}
For any $\Sigma^b_0$-formula $F(x_1,\dots, x_n)$
in $\Lpw$, there is a term $\termO_F(x_1,\dots, x_n)$
of $\POR^\lambda$ such that:
1. $\vdash_{I\POR^\lambda} F \leftrightarrow \termO_F=\zeroT$,
2. $\vdash_{I\POR^\lambda} \termO_F=\zeroT \vee
\termO_F =\oneT$.
\end{prop}
\noindent
This leads us to the following corollary and
to Theorem~\ref{theorem:IRStoIPOR},
realting $I\POR^\lambda$ and $I\RS$.

\begin{cor}
i. For any $\Sigma^b_0$-formula $F$,
$\vdash_{I\POR^\lambda} F\vee \neg F$;
ii. For any closed $\Sigma^b_0$-formula of $\Lpw$
$F$ and $\omega \in \Os$,
either $T_\omega \vdash_{I\POR^\lambda} F$
or $T_\omega \vdash_{I\POR^\lambda} \neg F$.
\end{cor}

\begin{theorem}\label{theorem:IRStoIPOR}
Any theorem of $I\RS$ is a theorem of $I\POR^\lambda$.
\end{theorem}

\begin{proof}
First, observe that, as a consequence of Proposition~\ref{prop:IPORprop},
for any $\Sigma^b-1$-formula $F=\exists x_1\preceq t_1.
\dots \exists x_n\preceq t_n.G$ in $\Lpw$,
$
\vdash_{I\POR^\lambda} F \leftrightarrow \exists x_1\preceq
\termO_1.\dots \exists x_n\preceq \termO_n.\termO_G=\zeroT,
$
any instance of the $\Sigma^b_1$-recursion schema
of $I\RS$ is derivable in $I\POR^\lambda$
from the $\NP$-induction schema.
Then, we conclude noticing that basic
axioms of $I\RS$
are provable in $I\POR^\lambda$.
\end{proof}

\begin{cor}\label{cor:Sigma0b}
For any closed $\Sigma^b_0$-formula $F$
and $\omega\in \Os$, either $T_\omega \vdash_{I\POR^\lambda}
F$ or $T_\omega \vdash_{I\POR^\lambda}\neg F$.
\end{cor}

Due to Corollary~\ref{cor:Sigma0b}
we establish the following Lemma~\ref{lemma:IPOR}.

\begin{lemma}\label{lemma:IPOR}
For any closed $\Sigma^b_0$-formula $F$
and $\omega \in \Os$,
either $T_\omega\vdash_{I\POR^\lambda} F$ 
iff $\omega\in\model{F}$.
\end{lemma}
\begin{proof}
($\Rightarrow$) By induction on the structure of rules for 
$I\POR$.
($\Leftarrow$) For Corollary~\ref{cor:Sigma0b}, either
$T_\omega \vdash_{I\POR^\lambda}F$ or
$T_\omega \vdash_{I\POR^\lambda} \neg F$.
Hence, if $\omega\in\model{F}$, then
it cannot be $T_{\omega} \vdash_{I\POR^\lambda} \neg F$
(by soundness).
We conclude $T_\omega \vdash_{I\POR^\lambda} F$.
\end{proof}

\emph{Realizability.}
We introduce realizability internal to $I\POR^\lambda$.
As a corollary, we obtain that from any derivation in $I\RS$
– actually, in $I\POR^\lambda$ – of a formula
in the form $\forall x.\exists y.F(x,y)$,
one can extract a functional term of $\POR^\lambda$
$\textsf{f}:s\arrowT s$, such that 
$\vdash_{I\POR^\lambda}\forall x.F(x,\textsf{f}x)$.
This allows us to conclude that if $f$ is $\Sigma^b_1$-representable
in $I\RS$,
then $f\in\POR$.

\begin{notation}
Let $\mathbf{x},\mathbf{y}$ denote finite sequences
of term variables, (resp.) $x_1,\dots, x_n$ and $y_1,\dots,
y_k$ and $\mathbf{x}(\mathbf{y})$ be an abbreviation
for $y_1(\mathbf{x}), \dots, y_k(\mathbf{x})$.
Let $\Lambda$ be a shorthand for the empty sequence
and $y(\Lambda):= y$.
\end{notation}

\begin{defn}
\emph{Formulas $x \ \realize \ F$}
are defined by induction as follows:
\begin{align*}
\Lambda \ \realize \ F &:= F \\
\mathbf{x,y} \ \realize \ (G\wedge H) &:= (\mathbf{x} \ \realize \ 
G) \wedge
(\mathbf{y} \ \realize \ H) \\
z,\mathbf{x}, \mathbf{y} \ \realize \ (G \vee H)
&:=
(z=\zeroT \wedge \mathbf{x} \ \realize \ G)
\vee 
(z\neq \zeroT \wedge \mathbf{y} \ \realize \
H) \\
\mathbf{y} \ \realize \ 
(G \rightarrow H) &:= 
\forall \mathbf{x}.
(\mathbf{x} \ \realize \ G \rightarrow
\mathbf{y}(\mathbf{x}) \ \realize \ H) \wedge
(G\rightarrow H) \\
z,\mathbf{x} \ \realize \ \exists y.G &:=
\mathbf{x} \ \realize \ G\{z/y\} \\
\mathbf{x} \ \realize \ \forall y.G &:=
\forall y.(\mathbf{x}(y) \ \realize \ G),
\end{align*}
where no variable in $\mathbf{x}$ is free in
$F$. Given terms $\mathbf{t}=\termO_1,\dots, \termO_n$
we let $\mathbf{t} \ \realize \ F := (\mathbf{x} \ \realize \ F)
\{\mathbf{t}/\mathbf{x}\}$.
\end{defn}
\noindent
We relate the derivability of these new formulas with that of
formulas of $I\POR^\lambda$.
Proofs below are by induction (resp.) on the structure
of $I\POR^\lambda$-formulas and 
on the height of derivations.

\begin{theorem}[Soundness]\label{thm:Rsound}
If $\vdash_{I\POR^\lambda}\mathbf{t} \ \realize \ F$,
then $\vdash_{I\POR^\lambda} F$.
\end{theorem}

\begin{notation}
Given $\Gamma=F_1,\dots, F_n$,
let $\mathbf{x} \ \realize \ \Gamma$ be a shorthand 
for $\mathbf{x}_1 \ \realize F_1, \dots, \mathbf{x}_n
\ \realize \ F_n$.
\end{notation}

\begin{theorem}[Completeness]\label{thm:Rcomp}
If $\vdash_{I\POR^\lambda} F$,
then $\mathbf{t}$ such that $\vdash_{I\POR^\lambda}
\mathbf{t} \ \realize \ F$.
\end{theorem}

\begin{proof}
We prove that if $\Gamma \vdash_{I\POR^\lambda}F$,
there exist terms $\mathbf{t}$ such that
$\mathbf{x} \ \realize \ \Gamma \vdash_{I\POR^\lambda}
\mathbf{t x}_1 \dots, \mathbf{x}_n \ \realize \ F$.
The proof is by induction on the derivation of 
$\Gamma \vdash_{I\POR^\lambda} F$.
Let us consider just one example:
\begin{prooftree}
\AxiomC{$\vdots$}
\noLine
\UnaryInfC{$\Gamma\vdash G$}
\RightLabel{$\vee R_1$}
\UnaryInfC{$\Gamma\vdash G\vee H$}
\end{prooftree}
By IH, there exist terms $\mathbf{u}$, such that
$\mathbf{t} \ \realize \ \Gamma \vdash_{I\POR^\lambda}
\mathbf{tu} \ \realize \ G$.
Since $x,y \ \realize \ G\vee H$
is defined as $(x=\zeroT \wedge
y \ \realize \ G) \vee
(x\neq \zeroT \wedge y \ \realize \ H)$,
we can take $\mathbf{t}=\zeroT,\mathbf{u}$.
\end{proof}

\begin{cor}
Let $\forall x.\exists y.F(x,y)$ be a closed
term of $I\POR^\lambda$,
where $F$ is a $\Sigma^b_1$-formula.
Then, there is a closed term $\termO:s \arrowT s$
of $\POR^\lambda$ such that
$\vdash_{I\POR^\lambda}\forall x.F(x,\termO x)$.
\end{cor}

\begin{proof}
By Theorem~\ref{thm:Rcomp},
there exist $\mathbf{t}=\termO,w$
such that $\vdash_{I\POR^\lambda} \mathbf{t} \ \realize
\ \forall x.\exists y.F(x,y)$.
So, $\mathbf{t} \ \realize \ \forall x.\exists y.F(x,y)
\equiv \forall x.(\mathbf{t}(x) \ \realize \ \exists y.
F(x,y)) \equiv \forall x.(w(x) \ \realize \ F(x,\termO x))$.
From this, by Theorem~\ref{thm:Rsound},
we deduce $\vdash_{I\POR^\lambda}\forall x.F(x,\termO x)$.
\end{proof}

Now, we have all the ingredients to prove that if a function
is $\Sigma^b_1$-representable in $I\RS$,
then it is in $\POR$.

\begin{cor}\label{cor:IPOR}
For any function $f:\Os \times \Ss \to \Ss$,
if there is a closed $\Sigma^b_1$-formula
in $\Lpw$ $F(x,y)$, such taht:
\begin{enumerate}
\itemsep0em

\item $I\RS \vdash \forall x.\exists !y. F(x,y)$
\item $\model{F(\ooverline{\sigma_1}, \ooverline{\sigma_2})}
= \{\omega \ | \ f(\sigma_1,\omega) = \sigma_2\}$,
\end{enumerate}
then $f\in\POR$.
\end{cor}

\begin{proof}
Since $\vdash_{I\RS} \forall x.\exists !y.F(x,y)$,
by Theorem~\ref{theorem:IRStoIPOR}
$\vdash_{I\POR^\lambda} \forall x.\exists !y.F(x,y)$.
Then, from $\vdash_{I\POR^\lambda} \forall x.\exists y.
F(x,y)$, we deduce $\vdash_{I\POR^\lambda}
\forall x.F(x,\textsf{g}x)$ for some closed term
$\textsf{g}\in\POR^\lambda$,
by Corollary~\ref{cor:IPOR}.
Furthermore, by Theorem~\ref{thm:provRepr}.2,
there is a $g\in\POR$ such that for any
$\sigma_1,\sigma_2\in\Ss$ and $\omega\in\Os$,
$g(\sigma_1,\omega)=\sigma_2$, when 
$T_\omega \vdash_{I\POR^\lambda} \textsf{g} 
\ooverline{\sigma_1} =\ooverline{\sigma_2}$.
So, by Proposition~\ref{prop:PORtoIPOR},
for any $\sigma_1,\sigma_2\in\Ss$ and $\omega\in\Os$
if $g(\sigma_1,\omega)=\sigma_2$,
then $T_\omega \vdash_{I\POR^\lambda} \textsf{g}\ooverline{\sigma_1}=\ooverline{\sigma_2}$
and so $T_\omega \vdash_{I\POR^\lambda} 
F(\ooverline{\sigma_1},\ooverline{\sigma_2})$.
By Lemma~\ref{lemma:IPOR},
$T_\omega \vdash_{I\POR^\lambda} F(\ooverline{\sigma_1},
\ooverline{\sigma_2})$,
when $\omega \in \model{F(\ooverline{\sigma_1},\ooverline{\sigma_2})}$,
that is $f(\sigma_1,\omega)=\sigma_2$.
But then $f=g$, so since $g\in\POR$ also $f\in \POR$.
\end{proof}

\emph{$\forall \NP$-Conservativity of $I\POR^\lambda+EM$
over $I\POR^\lambda$.}
Corollary~\ref{cor:IPOR} is already close to the result
we are looking for.
The remaining step to conclude our proof is its extension
from intuitionistic $I\RS$ to classical $\RS$,
showing that any function which is $\Sigma^b_1$-representable
in $\RS$ is also in $\POR$.
The proof adapts method by~\cite{CookUrquhart}.
We start by considering an extension of $I\POR^\lambda$
via EM and show that the realizability interpretation
extends to it so that for any of its closed theorems
$\forall x.\exists y\preceq \termO. F(x,y)$,
being $F$ a $\Sigma^b_1$-formula,
there is a closed term $\termO : s \arrowT s$
of $\POR^\lambda$ such that 
$\vdash_{I\POR^\lambda} \forall x.F(x,\termO x)$.

Let EM be the excluded-middle schema $F\vee \neg F$,
and Markov's principle 
be defined as follows $\neg \neg(\exists x.F\rightarrow (exists x)F)$
where $F$ is a $\Sigma^b_1$-formula.

\begin{prop}\label{prop:12}
For any $\Sigma^b_1$-formula $F$,
if $\vdash_{I\POR^\lambda+ EM}F$,
then $\vdash_{I\POR^\lambda +(Markov)}F$.
\end{prop}

\begin{proof}[Proof Sketch]
The proof is by double negation translation with
the following two remarks:
1. for any $\Sigma^b_0$-formula $F$,
$\vdash_{I\POR^\lambda} \neg \neg F \rightarrow F$;
2. using (Markov), the double negation of an
instance of the $\NP$-induction can  be shown equivalent
the $\NP$-induction schema. 
\end{proof}
\noindent
We conclude by  that the realizability interpretation
defined above extends to $I\POR^\lambda$+(Markov),
that is for any closed theorem
$\forall x. \exists y\preceq \termO. F(x,y)$
with $F$ $\Sigma^b_1$-formula
of $I\POR^\lambda$+(Markov) there is a closed term
of $\POR^\lambda$ $\termO:s\arrowT s$ such that
$\vdash_{I\POR^\lambda}\forall x.F(x,\termO x)$.

Let assume given  an encoding $\sharp:(s\arrowT s)
\arrowT s$ in $I\POR^\lambda$ of first-order
unary functions as strings, together with a
``decoding'' function $\textsf{app}:s\arrowT s\arrowT s$
satisfying
$\vdash_{I\POR^\lambda} \textsf{app}(\sharp \textsf{f},x)
=\textsf{f}x$.
Moreover, let $x * y := \sharp(\lambda z.\textsf{BAnd}(\textsf{app}(x,z), \textsf{app}(y,z)))$
and $T(x):= \exists y.(\textsf{B}(\textsf{app}(x,y))=\zeroT)$.
There is a \emph{meet semi-lattice} structure on 
the set of terms of type $s$ defined by $\termO \sqsubseteq
\termT$ when $\vdash_{I\POR^\lambda} T (\termT) 
\rightarrow T(\termO)$
with top element $\underline{\oone} := \sharp (\lambda x.\oneT)$
and meet given by $x*y$.
Indeed, from $T(x*\oneT) \leftrightarrow T(x)$,
$x\sqsubseteq \underline{\oone}$ follows,
Moreover, from $\textsf{B}(\textsf{app}(x,\termT)=\zeroT$,
we obtain $\textsf{B}(\textsf{app}(x*y,\termT))=\textsf{BAnd}(\textsf{app}(x,\termT),\textsf{app}(y,\termT))=\zeroT$,
whence $T(x) \rightarrow T(x*y)$,
i.e.~$x*y\sqsubseteq x$.
One can similarly prove $x*y\sqsubseteq y$.
Finally, from $T(x) \rightarrow T(v)$
and $T(y) \rightarrow T(v)$,
we deduce $T(x*y) \rightarrow T(v)$,
by observing that $\vdash_{I\POR^\lambda}T(x*y)
\rightarrow T(y)$.
Notice that the formula $T(x)$ is \emph{not}
a $\Sigma^b$-one,
as its existential quantifier is not bounded.

\begin{defn}
For any $I\POR^\lambda$-formula $F$
and fresh variable $x$,
we define formulas $x\Vdash F$:
\begin{align*}
x \Vdash F &:= F \vee T(x) \ \ \ 
\ \ \ 
\ \ \ 
\ \ \ 
\ \ \ 
\ \ \ 
\ \ \ \text{(}F \text{ atomic)} \\
x \Vdash G \wedge H &:= 
x \Vdash G \wedge x \Vdash H \\
x \Vdash G \vee H &:=
x \Vdash G \vee
x \Vdash H \\
x \Vdash G \rightarrow H &:=
\forall y.(y\Vdash G \rightarrow x*y \Vdash H) \\
x \Vdash \exists y.G &:=
\exists y.x \Vdash G \\
x \Vdash \forall y.G &:=
\forall y.x \Vdash G.
\end{align*}
\end{defn}

\begin{lemma}\label{lemma:13}
If $F$ is provable in $I\POR^\lambda$
without using $\NP$-induction,
then $x\Vdash F$ is provable in $I\POR^\lambda$.
\end{lemma}
\begin{proof}[Proof Sketch]
By induction on the structure of formulas of  
$I\POR^\lambda$ as in~\cite{CoquandHofmann}.
\end{proof}

\begin{lemma}\label{lemma:14}
Let $F=\exists x\preceq \termO.G$,
where $F$ is a $\Sigma^b_0$-formula.
Then, there is a term $\termT_F : s$
with $FV(\termT_F)=FV(G)$ such that
$\vdash_{I\POR^\lambda} F \leftrightarrow T(\termT_F)$.
\end{lemma}

\begin{proof}
Since $G(x)$ is a $\Sigma^b_0$-formula, 
for all terms $\termT:s$,
$\vdash_{I\POR^\lambda} G(x) \leftrightarrow
\termT_{x\preceq \termO \wedge G}(x)=\zeroT$,
where $\termO_{x\preceq \termO\wedge G}$
has the free variables of $t$ and $G$.
Let $H(x)$ be a $\Sigma^b_0$-formula, it
is shown by induction on its structure that for any term
$\termF: s$,
$\termO_{H(\termF)} = \termO_{H}(\termF)$.
Then, $\vdash_{I\POR^\lambda}
\vdash F \leftrightarrow \exists x. \termO_{x\preceq \termT
\wedge G}(x) = \zeroT \leftrightarrow
\exists x.T(\sharp(\lambda x.
\termO_{x\preceq \termO\wedge G}(x)))$.
So, we let $\termT_F = \sharp(\lambda x.
\termO_{x\preceq \termT \wedge G}
(x))$.
\end{proof}
From which we deduce the following three properties:
i. $\vdash_{I\POR^\lambda} (x\Vdash F) \leftrightarrow
(F\vee T(x))$;
ii. $\vdash_{I\POR^\lambda} (x\Vdash F)
\leftrightarrow (F\rightarrow T(x))$;
iii. $\vdash_{I\POR^\lambda}(x\Vdash \neg\neg F)
\leftrightarrow (F\vee T(x))$,
where $F$ is a $\Sigma^b_1$-formula.

\begin{cor}[Markov's Principle]
If $F$ is a $\Sigma^b_1$-formula, then
$\vdash_{I\POR^\lambda} x \Vdash \neg \neg F \rightarrow F$.
\end{cor}

To define the extension $(I\POR^\lambda)^*$
of $I\POR^\lambda$, we introduce PIND($F$) as:
$$
(F(\epsilon) \wedge (\forall x.
(F(x) \rightarrow F(x\zeroT))
\wedge 
\forall x.(F(x) \rightarrow F(x\oneT))))
\rightarrow \forall x.F(x).
$$
Observe that if $F(x)$ is a formula of the form
$\exists y\preceq \termO.\termT=\termF$,
then $z\Vdash$ PIND$(F)$ is of the form PIND$(F(x)
\vee T(z))$,
which is \emph{not} an instance of the $\NP$-induction
schema.

\begin{defn}
Let $(I\POR^\lambda)^*$ be the theory extending 
$\POR^\lambda$ with all instances of the induction schema
PIND$(F(x) \vee G)$,
where $F(x)$ is of the form
$\exists y\preceq \termO.\termT=\termF$,
and $G$ is an arbitrary formula with $x\not\in FV(G)$.
\end{defn}
\noindent
The following Proposition relates derivability in $I\POR^\lambda$
and in $(I\POR^\lambda)^*$.

\begin{prop}
For any $\Sigma^b_1$-formula $F$,
if $\vdash_{I\POR^\lambda} F$,
then $\vdash_{(I\POR^\lambda)^*}x\Vdash F$.
\end{prop}

Finally, we extend realizability to $(I\POR^\lambda)^*$ by
constructing a realizer fo PIND($F(x) \vee G$).

\begin{lemma}\label{lemma:PIND}
Let $F(x):\exists y\preceq \termO.\termT=\zeroT$
and $G$ be any formula not containing free occurrences
of $x$.
Then, there exist terms $\mathbf{t}$ such that
$\vdash_{I\POR^\lambda} \mathbf{t} \ \realize$ 
PIND$(F(x) \vee G)$.
\end{lemma}
So, by Theorem~\ref{thm:Rsound},
for any $\Sigma^b_1$-formula $F$
and formula $G$, with $x\not\in FV(F)$,
$\vdash_{I\POR^\lambda}$ PIND$(F(x) \vee G)$.

\begin{cor}[$\forall \NP$-Conservativity]
Let $F$ be a $\Sigma^b_1$-fromula.
If $\vdash_{I\POR^\lambda + EM} \forall x.\exists y\preceq \termO.
F(x,y)$,
then $\vdash_{I\POR^\lambda}\forall x.\exists y\preceq \termO.
F(x,y)$.
\end{cor}

We conclude the proof establishing the following 
Proposition~\ref{prop:Markov}.

\begin{prop}\label{prop:Markov}
Let $\forall x.\exists y\preceq \termO.F(x,y)$
be a closed term of $I\POR^\lambda$+(Markov),
where $F$ is a $\Sigma^b_1$-formula.
Then, there is a closed term of $\POR^\lambda$
$\termO:s \arrowT s$ such that
$\vdash_{I\POR^\lambda} \forall x.F(x,\termO x)$.
\end{prop}

\begin{proof}
If $\vdash_{I\POR^\lambda + (Markov)} \forall x.\exists
y.F(x,y)$,
then by Proposition~\ref{prop:Parikh},
also $\vdash_{I\POR^\lambda+(Markov)} \exists y\preceq
\termO.F(x,y)$.
Moreover, $\vdash_{(I\POR^\lambda)^*}z\Vdash \exists y\preceq
\termO.F(x,y)$.
Then, let us consider $G=\exists y\preceq \termO.F(x,y)$.
Taking $\termF=\termT_G$, by Lemma~\ref{lemma:14},
we deduce $\vdash_{(I\POR^\lambda)^*}G$
and, thus, by Lemma~\ref{lemma:13} and~\ref{lemma:PIND},
we conclude that there exist $\mathbf{t},\mathbf{u}$
such that $\vdash_{I\POR^\lambda} \mathbf{t}, \mathbf{u} 
\ \realize \ G$,
which implies $\vdash_{I\POR^\lambda}F(x,\mathbf{t}x)$
and so $\vdash_{I\POR^\lambda} \forall x.(F(x),\mathbf{t}x)$.
\end{proof}

So, by Proposition~\ref{prop:12},
if $\vdash_{I\POR^\lambda +EM} \forall x.\exists y\preceq
\termO. F(x,y)$,
being $F$ a closed $\Sigma^b_1$-formula,
then there is a closed term of $\POR$
$\termO:s\arrowT s$ such that
$\vdash_{I\POR^\lambda} \forall x.F(x,\termO x)$.
Finally, we conclude the desired Corollary~\ref{cor:fine}.

\begin{cor}\label{cor:fine}
Let $\RS \vdash \forall x. \exists y\preceq t.F(x,y)$,
where $F$ is a $\Sigma^b_1$-formula with only $x,y$ free.
For any $f:\Ss \times \Os \to \Ss$,
if $\forall x.\exists y\preceq t.F(x,y)$ represents $f$ so that:
\begin{enumerate}
\itemsep0em
\item $\RS \vdash \forall x.\exists !y. F(x,y)$
\item $\model{F(\ooverline{\sigma_1},\ooverline{\sigma_2})} =
\{\omega \ | \ f(\sigma_1,\omega)=\sigma_2\}$,
\end{enumerate}
then $f\in\POR$.
\end{cor}

\section{Proofs from Section~\ref{sec:FAtoRC}}\label{app:2}

\subsection{From $\RFP$ to $\POR$}

The goal of this section is to establish a correspondence between
$\RFP$ and $\POR$. This passes through Proposition \ref{prop:RFPtoSFP},
which assesses the equivalence between $\RFP$ and the
intermediate class $\SFP$ and Proposition \ref{prop:SFPtoPOR},
which concludes the proof showing the equivalence between $\POR$
and $\SFP$.
To establish rigorously the results mentioned above,
we fix some definitions. We start with those of STMs and their
configurations:

\begin{defn}[Stream Turing Machine]\label{df:streamMachine}
A \emph{stream Turing machine} is a quadruple
$M:= \langle {\Qs}, {q_0}, {\Sigma}, {\delta} \rangle$, where:
\begin{itemize}
\itemsep0em
\item ${\Qs}$ is a finite set of states ranged over by
${q_i}$ and similar meta-variables;
\item ${q_0} \in {\Qs}$ is an initial state;
\item ${\Sigma}$ is a finite set of characters
ranged over by ${c_i}$ \emph{et simila};
\item ${\delta}: {\hat{\Sigma}}
\times {\Qs} \times {{\hat{\Sigma}}} \times
\Bool
\longrightarrow {\hat{\Sigma}} \times {\Qs} \times {\hat{\Sigma}} \times \{{L},{R}\}$
is a transition function describing the new configuration
reached by the machine.
\end{itemize}
{$L$} and {$R$} are two fixed and distinct symbols,
e.g.~{$\zzero$} and {$\oone$},
${\hat{\Sigma}}={\Sigma} \cup \{{\circledast}\}$
and {$\circledast$} represents the
\emph{blank character}, such that ${\circledast} \not \in
{\Sigma}$.
Without loss of generality, in the following, we will use STMs with $\Sigma=\Bool$.
\end{defn}

\begin{defn}[Configuration of STM]\label{df:STMConfiguration}
The \emph{configuration of an STM}
is a quadruple $\langle {\sigma},
{q}, {\tau}, {\eta}\rangle$,
where:
\begin{itemize}
\itemsep0em
\item ${\sigma} \in \{\zzero, \oone, \circledast\}^*$
is the portion of the work tape on the left of the head;
\item ${q}\in \Qs$ is the current state of the machine;
\item ${\tau} \in \{\zzero, \oone, \circledast\}^*$ is the portion of the
work tape on the right of the head;
\item ${\eta} \in \Bool^\Nat$ is the portion
of the oracle tape that has not been read yet.
\end{itemize}
\end{defn}

%
\noindent
Thus, we give the definition of the family of reachability
relations for STM machines.

\begin{defn}[Stream Machine Reachability Functions]
  \label{def:smreachfuns}
Given an STM $\STM$ with transition function $\delta$,
we denote with  $\vdash_\delta$ its standard step function and
we call $\{\reaches n {\STM}\}_n$
the smallest family of relations
for which:
\footnotesize
\begin{align*}
\langle \sigma, q, \tau, \eta \rangle
&\reaches 0 M
\langle \sigma, q, \tau, \eta\rangle \\
\Big (
\langle \sigma, q, \tau, \eta\rangle
\reaches n M \langle \sigma',
q', \tau',\eta'\rangle
\Big )
\wedge
\Big (
\langle \sigma', q',
\tau', \eta'\rangle
&\vdash_{\delta}
\langle \sigma'', q', \tau'',
\eta'' \rangle
\Big )
\rightarrow
\Big (
\langle \sigma,q,\tau, \eta\rangle
\reaches {n+1} M\langle
\sigma'', q'\tau'',\eta''\rangle
\Big )
\end{align*}
\normalsize
\end{defn}

\begin{defn}[STM Computation]\label{df:STMcomputation}
Given an STM,
$\STM = \langle {\Qs}, {q_0}, {\Sigma}, {\delta}\rangle$,
$\eta:\Nat \longrightarrow \Bool$ and a function
$g : \Nat \longrightarrow \Bool$,
we say that
\emph{$\STM$ computes $g$},
written $f_\STM = g$ iff for every
string $\sigma\in \Ss$,
and oracle tape ${\eta} \in \Bool^\Nat$,
there are $n\in \Nat$, ${\tau}\in \Ss, {q'} \in \Qs$,
and a function
${\psi} : \Nat \longrightarrow \Bool$ such that:
$$
\langle {\eepsilon}, {q_0}, {\sigma}, {\eta}\rangle
\ \reaches n \STM \
\langle {\gamma}, {q'}, {\tau}, {\psi}\rangle,
$$
and $\langle {\gamma}, {q'}, {\tau}, {\psi}\rangle$
is a final configuration for $\STM$
with $f_\STM(\sigma, \eta)$ being
the longest suffix of ${\gamma}$ not including
${\circledast}$.
\end{defn}

Similar notations are employed for all the
families of Turing-like machines we define in this paper.
However, PTMs are an exception since they compute
distributions over $\Ss$ instead of functions $\Ss\times\Os \longrightarrow \Ss$. 

\begin{defn}[Probabilistic Turing Machines]
  \label{def:ptmX}
  A Probabilistic Turing Machine (PTM) is a TM with two transition functions $\delta_\zzero$ and $\delta_\oone$ at each step of the computation, $\delta_\zzero$ is applied with probability $\frac 1 2$, otherwise $\delta_\oone$ is applied. Given a PTM  $\PTM$, a configuration $\langle\sigma, q, \tau\rangle$, we define its semantics on configurations $\langle \sigma, q, \tau\rangle$ as the following
\emph{sequence of random variables}:
\footnotesize
\begin{align*}
\forall \eta \in \Bool^\Nat. X_{M, 0}^{\langle\sigma, q, \tau\rangle} &:= \eta \mapsto \langle\sigma, q, \tau\rangle\\
\forall \eta \in \Bool^\Nat. X_{M, n+1}^{\langle\sigma, q, \tau\rangle} & := \eta \mapsto \begin{cases}
\delta_\bbool(X_{M, n}^{\langle\sigma, q, \tau\rangle}(\eta)) & \text{ if } \eta(n)=\bbool \land \exists \langle \sigma', q' \tau'\rangle. \delta_\bbool(X_{M, n}^{\langle\sigma, q, \tau\rangle}(\eta))=\langle \sigma', q', \tau'\rangle\\
X_{M, n}^{\langle\sigma, q, \tau\rangle}(\eta) & \text{ if } \eta(n)=\bbool \land \lnot\exists \langle \sigma', q' \tau'\rangle. \delta_\bbool(X_{M, n}^{\langle\sigma, q, \tau\rangle}(\eta))=\langle \sigma', q', \tau'\rangle\\
\end{cases}
\end{align*}
\normalsize

\noindent
Intuitively, the variable $X_{\PTM, n}^{\langle\sigma, q, \tau\rangle}$ maps each possible cylinder $\eta: \Nat \longrightarrow \Bool$ to
the configuration reached by the machine after exactly $n$ transitions where the first transition step has employed $\delta_{\eta(0)}$, the second has employed $\delta_{\eta(1)}$ and so on.
We say that a PTM $\PTM$ computes $Y_{\PTM,\sigma}$ iff
$\exists t \in \Nat. \forall \sigma.X_{\PTM, t}^{\langle \sigma, q_0, \tau\rangle}$ is final.
In such case, for every $\eta$, $Y_{\PTM,\sigma}(\eta)$ is the longest suffix of the leftmost element in
$X_{\PTM, t}^{\langle \sigma, q_0, \epsilon\rangle}(\eta)$, which does not contain $\circledast$.
\end{defn}

We start with the proof of the equivalence between the class of the PTM's polytime functions and that of the STMs' polytime ones.

\begin{prop}
  \label{prop:RFPtoSFP}
  For any poly-time STM $\STM$ there is a polytime PTM $\PTM$ such that for every $\sigma \tau \in \Ss$:
  $$
  \mu(\{\omega \in \Os\mid N(\sigma \omega) = \tau\})=\text{Pr}[M(\sigma)=\tau]
  $$
  and viceversa.
\end{prop}

\begin{proof}
The claim can be restated as follows:
\begin{align*}
\forall \sigma, \tau.\mu(\{\eta \in \Bool^\Nat\ |\  N(\sigma, \eta)= \tau\})&=\mu(M(\sigma)^{-1}(\tau))\\
\forall \sigma, \tau.\mu(\{\eta \in \Bool^\Nat\ |\  N(\sigma, \eta)= \tau\})&=\mu(\{\eta \in \Bool^\Nat\ |\  Y_{M,\sigma} (\eta) = \tau\}).
\end{align*}
\noindent
Actually, we will show a stronger result: there is bijection $I: \text{STMs} \longrightarrow \text{PTM}$ such that:
\begin{equation}
\label{eq:measure}
\forall n \in \Nat.\{\eta \in \Bool^\Nat\ |\  \langle \sigma, q_0, \tau, \eta\rangle  \reaches n \delta \langle \tau, q, \psi, n\rangle \} = \{\eta \in \Bool^\Nat\ |\  X_{I(N), n}^{\langle \epsilon, q_0, \sigma\rangle} (\eta)= {\langle \tau, q, \psi\rangle}\}
\end{equation}
\noindent
which entails:
\begin{equation}
\label{eq:measurecons}
\{\eta \in \Bool^\Nat\ |\  N(\sigma, \eta)= \tau\} = \{\eta \in \Bool^\Nat\ |\  Y_{I(N),\sigma} (\eta) = \tau\}.
\end{equation}
\noindent
For this reason, it suffices to construct $I$ and prove that  \eqref{eq:measure} holds.
$I$ splits the function $\delta$ of $N$ in
such a way that transition is assigned to $\delta_0$
if it matches the character $\zzero$ on the oracle-tape, otherwise it
is assigned to $\delta_1$. Observe that $I$ is bijective, indeed, its inverse is a function as well, because it consists in a disjoint union.
Claim \eqref{eq:measure} can be shown by induction on the number of steps required by $N$ to compute its output value, the proof is standard, so we omit it.

\end{proof}

\begin{prop}
  \label{prop:SFPtoPOR}
  For every $f: \Ss \times \Bool^\Nat \longrightarrow Ss$ in $\SFP$, there is a function $g: \Ss \times \Os \longrightarrow \Ss$ in $\POR$ such that:
  $$
  \forall x, y \in \Ss. \mu(\{\omega \in \Os | g(x, \omega)= y\}) = \mu(\{\eta \in \Bool^\Nat | f(x, \eta)= y\}).
  $$
\end{prop}

To prove the correspondence
between the class of polytime STM computable function
and $\POR$, we pass through the class of \emph{finite-stream} TMs.
These machines are defined analogously to STMs,
but instead of an infinite stream of bits $\eta$,
they employ a finite sequence of random bits as additional argument. 

\begin{lemma}\label{lemma:SFPtopolyF}
For each $f\in \SFP$ with time-bound $p\in \POLY$,
there is a polytime \emph{finite-stream} TM
computable function $h$ such that for any
$\eta\in \Bool^\Nat$ and $x, y\in \Ss$,
$$
f(x,\eta) = h(x, \eta_{p(|x|)}).
$$
\end{lemma}

\begin{proof}
Assume that $f \in \SFP$.  For this reason
there is a polytime STM, $\STM =\langle {\Qs},
{q_0}, {\Sigma}, {\delta}\rangle$,
such that $f=f_\STM$.
Take the \emph{Finite Stream Turing Machine} (FSTM) $\STM'$
which is defined identically to $\STM$.
It holds that
for any $k\in \Nat$ and some $\sigma,\tau , y' \in \Ss$,
$$
\langle {\eepsilon}, {q_0'}, {x}, {y} \rangle
\triangleright^k_{\STM} \langle {\sigma},
{q}, {\tau}, {y'}\rangle
\ \ \ \Leftrightarrow \ \ \
\langle \eepsilon, q_0', x, y\eta\rangle
\triangleright^k_{\STM'}
\langle {\sigma}, {q}, {\tau}, {y'\eta}\rangle.
$$
Moreover,
$\STM'$ requires a number of steps which is exactly equal to
the number of steps required by $\STM$, and thus
the complexity is preserved.
We conclude the proof defining $h=f_{\STM'}$.
\end{proof}

\begin{figure}
  \centering
  \begin{tikzpicture}[node distance = 8 cm]
    \node (c) {$c := \langle \sigma, q, \tau, y\rangle$};
    \node[below = 1 cm of c] (sc) {$s_c\in \Ss$};
    \node[right of = c] (d) {$d := \langle \sigma, q, \tau, y\rangle$};
    \node[below = 1 cm of d] (sd) {$s_d\in \Ss$};

    \draw[->] (c) edge node[fill = white] {$\vdash_{\delta}$} (d);
    \draw[->] (sc) edge node[fill = white] {$\forall \omega. \apply(s_c, e_t(\delta), \omega)=s_d$} (sd);
    \draw[->] (c) edge node[fill = white] {$e_c$} (sc);
    \draw[->] (d) edge node[fill = white] {$e_c$} (sd);
  \end{tikzpicture}
    \caption{Behavior of the function $\apply$.}
    \label{fig:apply}
  \end{figure}

The next step is to show that each polytime FSTM computable function $f$
corresponds to a function $g: \Ss \times \Ss \times \Os \longrightarrow \Ss$ of $\POR$ which can be
defined without recurring to $Q$.

\begin{restatable}{lemma}{polyFtoCob}
\label{lemma:polyFtoCob}
  For any polytime FSTM computable function $f$ and $x\in \Ss$,
  there is $g\in \POR$
  such that
  $
  \forall x, y, \omega. f(x, y)=g(x, y,\omega).
  $
  Moreover, if $f$ is defined without recurring to $Q$, $g$ can be defined without $Q$ as well.

\end{restatable}


A formal proof of Lemma \ref{lemma:polyFtoCob} requires too much
effort to be done extensively. In this work, we will simply mention
the high-level structure of the proof. It relies on the following observations:
\begin{enumerate}
  \item It is possible to encode FSTMs, together with configurations
  and their transition functions using strings, call these encodings
  $e_c\in \POR$ and $e_t$. Moreover,
  there is a function $\apply \in \POR$
  which satisfies the simulation schema of Figure \ref{fig:apply}. The
  proof of this result is done by explicit definition of the functions
  $e_c$, $e_t$ and $\apply$, proving
  the correctness of these entities with respect to the simulation schema above.
  \item For each $f \in \POR$ and $x, y \in \Ss$, if there is a
  term $t(x)$ in $\Lpw$ which bounds the size of $f(x, \omega)$ for each possible input, then the function $m(z, x, \omega) = f^{|z|}(x, \omega)$ is in $\POR$ as well, moreover, if $f$ is defined without recurring to $Q$, also $m$ can be defined without recurring to $Q$. This is shown in Lemma \ref{lemma:saPOR}.
  \item Fixed a machine $M$, if $\sigma \in \Ss$ is a correct encoding of a configuration of $M$, for every $\omega$, it holds that $|\apply(\sigma, \omega)|\le |\sigma| + c$, for $c \in \Nat$ fixed once and forall.
  \item If $c := e_c(\sigma, q, \tau, y, \omega)$ for some omega, then there is a
  function $\dectape$  such that $\forall \omega \in \Os.
  \dectape(x, \omega)$ is the longest suffix without
  occurrences of $\circledast$ of $\sigma$.
\end{enumerate}

\begin{lemma}
  \label{lemma:saPOR}
  For each $f : \Ss^{k+1} \times \Os \longrightarrow \Ss \in \POR$,
  if there is a term $t \in \Lpw$ such that
  $\forall x, \vec z, \omega. f(x, \vec z, \omega)|_t = f(x, \vec z, \omega)$
  then there is also a function $\sa_{f, t} : \Ss^{k+2} \times \Os \longrightarrow \Ss$
  such that:
  $$
  \forall x, n \in \Ss,  \omega \in \Os.
  sa_{f, t}(x, n , \vec z, \omega) =
  \underbrace{f(f(f(x, \vec z,  \omega), \vec z, \omega), \ldots)}_{|n|\text{ times}}.
  $$
Moreover, if $f$ is defined without recurring to $Q$, $sa_{f, t}$ can be defined without $Q$ as well.
\end{lemma}

\begin{proof}
  \label{proof:saPOR}
Given $f\in \POR$ and $t\in \mathcal{L}_{\mathbb{PW}}$,
let $sa_{f,t}$ be defined as follows:
\begin{align*}
sa_{f,t} (x,\eepsilon, \vec{z}, \omega) &:= x \\
sa_{f,t} (x,y\bbool, \vec{z}, \omega) &:= f\big(sa_{f,t}(x,y,\omega),
\vec{z},\omega\big)|_{t}
\end{align*}
The correctness of $\sa$ comes as a direct consequence of its definition by induction on $n$.
\end{proof}

Combining these results, we are able to prove Lemma \ref{lemma:polyFtoCob}

\begin{proof}[Proof of Lemma \ref{lemma:polyFtoCob} (Sketch)]
  As a consequence of points (2) and (3), we obtain that: $k(x, n, y, \omega) = \apply^{|n|}(x, y, \omega)$ belongs to $\POR$ and can be defined without recurring to $Q$.
  As a consequence of (1) we have that:
  $$
  k'(x, y, \omega):= k(e_c( x, q_0, \eepsilon, y,
  \omega), y, \omega)
  $$
  belongs to $\POR$ as well and can be
  defined without recurring to $Q$.
  Finally, as a consequence of (4) and $\POR$'s closure under
  composition, there is a function $g$ which returns the
  longest prefix of the leftmost projection
  of the output of $k'$. This function is exactly:
  $$
  g(x, y, \omega):= dectape(k'(x, y, \omega),\omega).
  $$
\end{proof}
As another consequence of Lemma \ref{lemma:polyFtoCob},
we show the result we were aiming to:
each function $f \in \SFP$ can be simulated by a function in $g\in \POR$,
using as an additional input a polynomial prefix of $f$'s oracle.

\begin{cor}
  \label{cor:sfptopor-}
For each $f\in \SFP$ and polynomial
time-bound $p\in \POLY$,
there is a function $g\in \POR$ such that
for any $\eta : \Nat \longrightarrow \Bool$, $\omega : \Nat \longrightarrow \Bool$ and $x\in\Ss$,
$$
f(x,\eta) = g\big(x,\eta_{p(|x|)}, \omega\big).
$$
\end{cor}
%
%
%
%
%

Now, we need to establish that there is a
function $e \in \POR$ which produces strings with
the same distribution of the prefixes of the functions in $\Ss^\Nat$.
Intuitively, this function is very simple:
it extracts $|x|+1$ bits from $\omega$
and concatenates them in its output.
The definition of the function $e$ passes through a bijection $\mathit{dyad}: \Nat \longrightarrow \Ss$, called \emph{dyadic
representation} of a natural number.
Thus, the function $e$ can
simply create the strings $\oone^0, \oone^1 ,\ldots, \oone^k$, and
sample the function $\omega$ on the coordinates $\mathit{dy}(\oone^0), \mathit{dy}(\oone^1) ,\ldots, \mathit{dy}(\oone^k)$, concatenating the result in a string.

%

\begin{defn}
  The function $\mathit{dyad}: \Nat \longrightarrow \Ss$
  associates each $n \in \Nat$ to the string obtained
  stripping the left-most bit from the binary representation of
  $n+1$.
\end{defn}

\begin{remark}
  \label{rem:dycorr}
  There is a $\POR$ function $\mathit{dy}: \Ss \times \Os \longrightarrow \Ss$ such that $\forall \sigma \in \Ss.\forall \omega\in\Os.  \mathit{dy}(\sigma, \omega)= \mathit{dyad}(\oone^{|\sigma|})$.
\end{remark}

The construction of this function is not much interesting to the aim of our proof, so we omit it.

\begin{defn}
Let $e :\Ss\times \Os \longrightarrow \Ss$ be defined as follows:
\begin{align*}
e(\eepsilon,\omega) &= \eepsilon; \\
e(x\bbool,\omega) &= e(x, \omega)Q\big(\mathit{dy}(x,\omega),\omega\big)|_{x\bbool}.
\end{align*}
\end{defn}
\noindent

\begin{lemma}[Correctness of $e$]
  \label{lemma:corrofe}
  For any $\sigma\in \Ss$ and $i\in\Nat$,
  if $|\sigma|=i+1$, for any $j\leq i\in \Nat$ and
  $\omega \in \Os$, $(i)$
  $e(\sigma, \omega)(j) = \omega(\mathit{dy}(\oone^{j}, \omega))$
  and $(ii)$ the length of $e(\sigma, \omega)$ is exactly $i+1$.
\end{lemma}
\begin{proof}
  Both claims are proved by induction on $\sigma$. The latter is trivial,
  whereas the former requires some more effort:
  \begin{itemize}
    \item[$\eepsilon$] The claim comes from vacuity of the premise $|\sigma| = i+1$.
    \item[$\tau \bbool$] It holds that:
    $
    e(\tau\bbool, \omega)(j)=
    e(\tau, \omega)Q\big(\mathit{dy}(\tau,\omega),\omega\big)=
    e(\tau, \omega)\omega\big(\mathit{dy}(\tau,\omega)\big).
    $
    By $(ii)$, for $j=i+1$, the
    $j$-th element of $e(\tau\bbool, \omega)$ is exactly
    $Q\big(\mathit{dy}(\tau,\omega), \omega \big)$, which is equal to
    $\omega(\mathit{dy}(\tau,\omega))$, in turn
    equal to $\omega(\mathit{dy}(\oone^{j}, \omega))$ (by Remark \ref{rem:dycorr}). For smaller values of $j$,
    the first claim is a consequence of the definition of $e$ together with IH.
  \end{itemize}
\end{proof}

\begin{defn}
We define $\sim_{\mathit{dy}}$ as the smallest relation in
 $\Os \times \Bool^\Nat$ such that:
 $$
 \eta \sim_\mathit{dy} \omega \leftrightarrow \forall n \in \Nat.
  \eta(n)= \omega (\mathit{dy}(\oone^{n+1}, \omega)).
 $$
\end{defn}

\begin{lemma}
  \label{lemma:funbij}
  It holds that:
  \begin{itemize}
    \item $\forall \eta \in \Bool^\Nat. \exists ! \omega \in \Os. \eta \sim_{\mathit{dy}} \omega$;
    \item $\forall \omega \in \Os. \exists ! \eta \in \Bool^\Nat. \eta \sim_{\mathit{dy}} \omega$.
  \end{itemize}
\end{lemma}
\begin{proof}
  The proofs of the two claims are very similar.
  Since \emph{dyad} is a bijection, applying Remark \ref{rem:dycorr},
  we obtain the existence of an $\omega$ which is in relation with $\eta$.
  Now suppose that there are $\omega_1, \omega_2$ both in relation with $\eta$
  but they are different.
  Then, there is a $\sigma \in\Ss$, such that
  $\omega_1(\sigma)\neq\omega_2(\sigma)$
  and, by Remark~\ref{rem:dycorr},
  $\mathit{dy}(\sigma, \omega_1) = \mathit{dy}(\sigma, \omega_2)$
  which entails that $\eta(k) \neq \eta(k)$ for some $k$, that is a contradiction.
\end{proof}
\begin{cor}
  The relation $\sim_{\mathit{dy}}$ is a bijection.
\end{cor}

\begin{proof}
  Consequence of Lemma \ref{lemma:funbij}.
\end{proof}

\begin{lemma}
  \label{lemma:auxsimdy}
  $$
  \eta \sim_{\mathit{dy}} \omega \to \forall n \in \Nat. \eta_n = e(1^{n+1}, \omega).
  $$
\end{lemma}
\begin{proof}
  By contraposition: suppose
  $
  \eta_n \neq e(\ovverline n \Nat, \omega).
  $
  As a consequence of the correctness of $e$ (Lemma \ref{lemma:corrofe}),
  there is an $i \in \Nat$ such that
  $\eta(i)\neq\omega(\mathit{dy}(\ovverline i\Nat, \omega))$,
  which is a contradiction.
\end{proof}

\noindent
We can finally conclude the proof of Proposition~\ref{prop:SFPtoPOR}.

\begin{proof}[Proof of Proposition~\ref{prop:SFPtoPOR}]
  From Corollary \ref{cor:sfptopor-}, we know that there is a function $f'\in \POR$,
  and a $p \in \POLY$ such that:
  \begin{equation}
  \forall x, y \in \Ss.\forall \eta.\forall \omega. y = \eta_{p(x)} \to  f(x, \eta) = f'(x, y, \omega).\tag{$*$}
\end{equation}
\noindent
  Fixed an $\overline \eta \in \{\eta \in \Bool^\Nat \ | \ f(x, \eta)=y \}$,
  its image with respect to $\sim_{\mathit{dy}}$
  is in
  $
  \{\omega \in \Os \ \mid \ f'(x, e(p'(\mathit{s}(x, \omega), \omega), \omega), \omega) =y\},
  $
  where $\mathit{s}$ is the $\POR$-function computing $\oone^{|x|}$.
  Indeed, by Lemma~\ref{lemma:auxsimdy}, it holds that
  $\overline \eta_{p(x)} = e(p(\mathit{size}(x, \omega), \omega)$, where $p'$ is the $\POR$-function computing the polynomial $p$, defined without recurring
  to $Q$.
  By $(*)$, we prove the claim.
  It also holds that, given a fixed $\overline \omega \in \{\omega \in \Os \ | \ f'(x, e(p'(\mathit{size}(x, \omega), \omega), \omega), \omega) =y\}$,
  its pre-image with respect to $\sim_{\mathit{dy}}$ is in $\{\eta \in \Bool^\Nat \ | \ f(x, \eta)=y \}$.
  The proof is analogous to the one we showed above.
  Now, since $\sim_{\mathit{dy}}$ is a bijection between the two sets:
  $
  \mu\big(\{\eta \in \Bool^\Nat \ | \ f(x, \eta)=y \}\big)
  = \mu\big(\{\omega \in \Os \ | \ f'(x, e(p(\mathit{size}(x, \omega), \omega), \omega), \omega) =y\}\big),
  $
  which concludes the proof.
\end{proof}

\subsection{From $\POR$ to $\SFP$}

We start defining the imperative language $\SIFPRA$ and proving its
polytime programs equivalent to $\POR$. To do so, we first introduce
the definition of $\SIFPRA$ and its \emph{big-step} semantics.

\begin{defn}[Correct programs of $\SIFPRA$]
 \label{def:sifp}
The language of $\SIFPRA$ programs is $\lang{\stm_\RA}$, i.e.
the set of strings produced by the non-terminal symbol $\stm_\RA$ defined by:
\begin{align*}
\id &::= X_i\ |\ Y_i\ |\ S_i\ |\ R\ |\ Q\ |\ Z\ |\ T\qquad i \in \Nat\\
\xp &::= \epsilon\ |\ \xp.\zero\ |\ \xp.\one\ |\ \id\ |\ \xp \sqsubseteq \id\ |\ \xp \land \id\ |\ \lnot \xp\\
\stm_\RA & ::= \id \takes \xp\ |\ \stm_\RA;\stm_\RA\ |\ \while \xp \stm_\RA\ |\ \fl \xp
\end{align*}
\end{defn}

The \emph{big-step} semantics associated to the language of the $\SIFPRA$ programs relies on the notion of \emph{Store}, which for us is a function $\store: \id \rightharpoonup \{\zero, \one\}^*$.

We define the updating of a store $\store$ with a mapping from $y \in \id$ to $\tau \in \{\zero, \one\}^*$ as:
\[
\store[y\leftarrow \tau](x) := \begin{cases} \tau & \text{if } x = y\\ \store(x) & \text{otherwise.}\end{cases}
\]



\begin{defn}[Semantics of $\SIFP$ expressions]
  \label{def:expsemantics}
The semantics of an expression $E \in \lang{\xp}$ is the smallest relation
$\sred: \lang{\xp} \times (\id \longrightarrow \{\zero, \one\}^*)\times \Os \times \{\zero, \one\}^*$ closed under the following rules:
\begin{center}
\AxiomC{\phantom{$\langle \epsilon, \store\rangle \sred \epsilon$}}
\UnaryInfC{$\langle \epsilon, \store\rangle \sred \epsilon$}
\DisplayProof
\hspace{18pt}
\AxiomC{$\langle e, \store \rangle \sred \sigma$}
\UnaryInfC{$\langle e.\zero, \store\rangle \sred \sigma \conc \zero$}
\DisplayProof
\hspace{18pt}
\AxiomC{$\langle e, \store \rangle \sred \sigma$}
\UnaryInfC{$\langle e.\one, \store\rangle \sred \sigma \conc \one$}
\DisplayProof

\vspace{12pt}
\AxiomC{$\langle e, \store \rangle \sred \sigma$}
\AxiomC{$\store(\id) = \tau$}
\AxiomC{$\sigma \subseteq \tau$}
\TrinaryInfC{$\langle e \sqsubseteq \id, \store\rangle \sred \one$}
\DisplayProof
\hspace{18pt}
\AxiomC{$\langle e, \store \rangle \sred \sigma$}
\AxiomC{$\store(\id) = \tau$}
\AxiomC{$\sigma \not\subseteq \tau$}
\TrinaryInfC{$\langle e \sqsubseteq \id, \store\rangle \sred \zero$}
\DisplayProof

\vspace{12pt}
\AxiomC{$\store(\id)=\sigma$}
\UnaryInfC{$\langle \id, \store\rangle \sred \sigma$}
\DisplayProof
\hspace{18pt}
\vspace{12pt}
\AxiomC{$\id \not \in \mathit{dom}(\store)$}
\UnaryInfC{$\langle \id, \store\rangle \sred \epsilon$}
\DisplayProof
%
\hspace{18pt}
\AxiomC{$\langle e, \store \rangle \sred \zero$}
\UnaryInfC{$\langle \lnot e, \store\rangle \sred \one$}
\DisplayProof
\hspace{18pt}
\AxiomC{$\langle e, \store \rangle \sred \sigma$}
\AxiomC{$\sigma \neq \zero$}
\BinaryInfC{$\langle \lnot e, \store\rangle \sred \zero$}
\DisplayProof

\vspace{12pt}
\AxiomC{$\langle e, \store \rangle \sred \one$}
\AxiomC{$\store(\id) = \one$}
\BinaryInfC{$\langle e \land \id, \store\rangle \sred \one$}
\DisplayProof
\hspace{18pt}
\AxiomC{$\langle e, \store \rangle \sred \sigma$}
\AxiomC{$\store(\id) = \tau$}
\AxiomC{$\sigma \neq \one \land \tau \neq \one$}
\TrinaryInfC{$\langle e \land \id, \store \rangle \sred \zero$}
\DisplayProof
\hspace{18pt}

\end{center}
\end{defn}

\begin{defn}[big-step Operational Semantics of $\SIFPRA$]
  \label{def:sifpraos}
The semantics of a program $P \in \lang{\stm_\RA}$ is the smallest relation
$\ssos\subseteq \lang{\stm_\RA} \times (\id \longrightarrow \{\zero, \one\}^*)\times \Os\times (\id \longrightarrow \{\zero, \one\}^*)$ closed under the following rules:
\begin{center}
\AxiomC{$\langle e, \store\rangle\sred \sigma$}
\UnaryInfC{$\langle \id\takes e, \store, \omega\rangle \ssos \store\as {\id}{\sigma}$}
\DisplayProof
\hspace{18pt}
\AxiomC{$\langle s, \store, \omega\rangle\ssos \store'$}
\AxiomC{$\langle t, \store', \omega\rangle\ssos \store''$}
\BinaryInfC{$\langle s;t, \store, \omega\rangle \ssos \store''$}
\DisplayProof

\vspace{12pt}
\AxiomC{$\langle e, \store\rangle\sred \one$}
\AxiomC{$\langle s, \store, \omega\rangle\ssos \store'$}
\AxiomC{$\langle \while e s, \store', \omega\rangle\ssos \store''$}
\TrinaryInfC{$\langle \while e s, \store, \omega\rangle \ssos \store''$}
\DisplayProof
\hspace{18pt}
\AxiomC{$\langle e, \store\rangle\sred \sigma$}
\AxiomC{$\sigma \neq \one$}
\BinaryInfC{$\langle \while e s, \store, \omega\rangle \ssos \store$}
\DisplayProof

\vspace{12pt}
\AxiomC{$\langle e, \store\rangle \sred \sigma$}
\AxiomC{$\omega(\sigma)=b$}
\BinaryInfC{$\langle \fl e, \store, \omega\rangle \ssos \store[R \leftarrow b]$}
\DisplayProof

\end{center}
\end{defn}

This semantics allows us to associate each $\SIFPRA$ program to the
function it evaluates:

\begin{defn}[Function evaluated by a $\SIFPRA$ program]
  \label{def:simprafuneval}
We say that the function evaluated by a correct $\SIFPRA$ program $P$ is $\mathcal \llbracket \cdot\rrbracket: \lang{\Stm_{\RA}} \longrightarrow (\Ss^n \times \Os \longrightarrow \Ss)$, defined as below\footnote{Instead of the infix notation for $\ssos$, we will use its prefixed notation. So, the notation express the store associated to the $P$, $\Sigma$ and $\omega$ by $\ssos$. Moreover, notice that we employed the same function symbol $\ssos$ to denote two distinct functions: the \emph{big-step} operational semantics of $\SIFPRA$ programs and the \emph{big-step} operational semantics of $\SIFPLA$ programs}:
\[
\llbracket P\rrbracket:= \lambda x_1, \ldots, x_n, \omega.\ssos(\langle P, []\as {X_1} {x_1}, \ldots, \as {X_n} {x_n}, \omega\rangle)(R).
\]
\end{defn}

Observe that, among all the different registers, the register
$R$ is meant to contain the value computed by the program at the end
of its execution, similarly the $\{X_i\}_{i \in \Nat}$ registers are
used to store the function's inputs. The correspondence between
$\POR$ and $\SIFPRA$
can be stated as follows:

\begin{lemma}[Implementation of $\POR$ in $\SIFPRA$]
\label{lemma:portosifp}
For every function $f \in \POR$, there is a polytime $\SIFPRA$ program $P$
such that:
$\forall x_1, \ldots x_n.
\llbracket P\rrbracket(x_1, \ldots, x_n, \omega)=f(x_1, \ldots, x_n, \omega)$.
Moreover, if $f$ can be defined without recurring to $Q$, then $P$ does not
contain any $\fl e$ statement.
\end{lemma}

The proof of this result is quite simple: it
relies on the fact that it is possible to associate to each $\POR$
function an equivalent polytime program, and on the observation
that it is possible to compose them and to
implement bounded recursion on notation in $\SIFPRA$ with a polytime overhead.

\begin{proof}[Proof Sketch of Lemma \ref{lemma:portosifp}]
For each function $f \in \POR$ we define a program $\LL{f}$ such that
$\llbracket \LL{f}\rrbracket(x_1, \ldots, x_n)=f(x_1, \ldots, x_n)$
The correctness of $\LL{f}$
is given by the following invariant properties:
\begin{itemize}
\item The result of the computation is stored in $R$.
\item The inputs are stored in the registers of the group $X$.
\item The function $\LL\cdot$ does not change the values it accesses as input.
\end{itemize}
\noindent
We define the function $\LL f$ as follows:
$\LL{E}:= R \takes \epsilon$;  $\LL{S_0}:= R \takes X_0.\zero$;
$\LL{S_1}:= R \takes X_0.\one$;
$\LL{{P}^n_i}:= R \takes X_i$;
$\LL{C}:= R \takes X_1 \sqsubseteq X_2$;
$\LL{Q}:= \fl {X_1}$.
\noindent
The correctness of these base cases is trivial.
Moreover, it is simple to see that
the only translation containing $\fl e$ for some $e \in \lang\xp$ is
the translation of $Q$.
The encoding of the composition and of the bounded recursion are
a little more convoluted: the proof of their correctness requires
a conspicuous amount of low-level definitions and technical results,
whose presentation is not the aim of this work.
\end{proof}
The next step is to sow that every $\SIFPRA$ program is equivalent to a $\SIFPLA$ program in the sense of Lemma \ref{lemma:sifpratosifpla}.

\begin{lemma}
  \label{lemma:sifpratosifpla}
  For each total program $P \in \SIFPRA$ there is a $Q \in \SIFPLA$ such that:
  $$
  \forall x, y. \mu\left(\{\omega \in \Bool^\Ss| \llbracket P\rrbracket (x, \omega)= y\}\right)=
                \mu\left(\{\eta \in \Bool^\Nat| \llbracket Q \rrbracket(x, \eta)= y\}\right).
  $$
  Moreover, if $P$ is polytime $Q$ is polytime, too.
\end{lemma}

Before delving into the details of the proof of this Lemma, we must define
the language $\SIFPLA$ together with its standard semantics:

\begin{defn}[$\SIFPLA$]
  The language of the $\SIFPLA$ programs is $\lang{\stm_\LA}$,
  i.e. the set of strings produced by the non-terminal symbol $\stm_\LA$
  described as follows:
  $$
  \stm_\LA  ::= \id \takes \xp\ |\ \stm_\LA;\stm_\LA\ |\ \while \xp \stm_\LA\ |\ \rb\\
  $$
\end{defn}

\begin{defn}[Big Step Operational Semantics of $\SIFPLA$]
  \label{def:sifplaos}
The semantics of a program $P \in \lang{\stm_\LA}$ is the smallest relation $
\ssos\subseteq \left(\lang{\stm_\LA} \times (\id \longrightarrow \{\zero, \one\}^*)\times \Bool^\Nat\right)
\times
\left((\id \longrightarrow \{\zero, \one\}^*)\times \Bool^\Nat\right)$
closed under the following rules:
\begin{center}
\AxiomC{$\langle e, \store\rangle\sred \sigma$}
\UnaryInfC{$\langle \id\takes e, \store, \eta\rangle \ssos \langle \store\as {\id}{\sigma}, \eta\rangle$}
\DisplayProof
\hspace{18pt}
\AxiomC{$\langle s, \store, \eta\rangle\ssos \langle \store', \eta'\rangle$}
\AxiomC{$\langle t, \store', \eta\rangle\ssos \langle\store'', \eta''\rangle$}
\BinaryInfC{$\langle s;t, \store, \eta\rangle \ssos \langle\store'', \eta''\rangle$}
\DisplayProof

\vspace{12pt}
\AxiomC{$\langle e, \store\rangle\sred \one$}
\AxiomC{$\langle s, \store, \eta\rangle\ssos \langle\store', \eta'\rangle$}
\AxiomC{$\langle \while e s, \store', \eta\rangle\ssos \langle\store'', \eta''\rangle$}
\TrinaryInfC{$\langle \while e s, \store, \eta\rangle \ssos \langle\store'', \eta''\rangle$}
\DisplayProof

\vspace{12pt}
\hspace{18pt}
\AxiomC{$\langle e, \store\rangle\sred \sigma$}
\AxiomC{$\sigma \neq \one$}
\BinaryInfC{$\langle \while e s, \store, \eta\rangle \ssos \langle\store, \eta\rangle$}
\DisplayProof
\hspace{18pt}
\AxiomC{\phantom{$\langle \rb \store, \bool\eta\rangle \ssos \langle\store[R \leftarrow \bool], \eta\rangle$}}
\UnaryInfC{$\langle \rb, \store, \bool\eta\rangle \ssos \langle\store \as R \bool, \eta\rangle$}
\DisplayProof

\end{center}
\end{defn}

In particular, we prove Lemma \ref{lemma:sifpratosifpla} showing that
$\SIFPRA$ can be simulated in $\SIFPLA$ with respect to two novel
\emph{small-step} semantic relations ($\leadstola, \leadstora$)
derived splitting the \emph{big-step} semantics into small
transitions, one per each $ \cdot ; \cdot$ instruction.
Intuitively, the idea behind this novel semantics is to enrich the \emph{big-step} operational
semantic with some
pieces of information necessary to build a proof that
it is possible to each $\SIFPLA$ instruction by meas of a sequence of $\SIFPRA$ instructions
preserving the distribution of the values in the register $R$, i.e. that storing the result.
In particular we enrich the configurations of $\SIFPLA$'s and $\SIFPRA$'s \emph{small-step}
operational semantics with a list $\Psi$ containing pairs $(x, b)$.
In the case of $\SIFPLA$, this list is
meant to keep track of the calls to the primitive $\fl x$ and their result $b$.
While, on the side of $\SIFPRA$, the $x$-th call of
the primitive $\rb$ causes the pair $(x, b)$ to be added on top of the list.

This is done to keep track of the accesses to the random tapes
done by the simulated and the simulator programs.
On the side of $\SIFPRA$ it is possible to show that
this table is stored in a specific register.
This register plays an
important role in the simulation of the $\fl e $ instructions.
In particular, this is done as follows:
\begin{itemize}
  \item At each simulated query $\fl e$,
  the destination program looks up the associative table;
  \item If it finds the queried coordinate $e$ within a pair $(e, b)$, it returns $b$,
  otherwise:
  \begin{itemize}
    \item It reduces $\fl{e}$ to a call of $\rb$ which outputs either $\bool=\zzero$ or
    $\bool=\oone$.
    \item It records the pair $\langle e,\bool\rangle$ in
   the associative table and stores $\bool$ in $R$.
  \end{itemize}
\end{itemize}
Even in this case the construction of the proof is convoluted, but we believe
that it is not too much of a problem to see, at least intuitively,
that this kind of simulation preserves the distributions
of strings computed by the original program.
This concludes Lemma \ref{lemma:sifpratosifpla}.

\noindent
The next step is to show that $\SIFPLA$ can be reduced to
$\SFPOD$: the class corresponding to $\SFP$ defined on
a variation of the Stream Machines
which are capable to read characters from the oracle tape
\emph{on-demand}, i.e. only on those states $q \in \Qs_{\natural}\subseteq \Qs$.
The transition function of the $\SFPOD$ is a function $\delta: {\hat{\Bool}}
\times {\Qs} \times {\left({\hat{\Bool}\cup\{\natural\}}\right)} \times
\Bool\longrightarrow {\hat{\Bool}} \times {\Qs} \times {\hat{\Bool}} \times \{{L},{R}\}$
which for all the states in $\Qs_{\natural}$ is labeled with the symbol $\natural$ instead of
a Boolean value. This peculiarity will be employed in Definition \ref{def:SFPODtoSFPmap} to distinguish those configurations
causing the oracle tape to shift to the others.
%

We do not show the reduction from $\SIFPLA$ to $\SFPOD$
extensively because this kind of reductions are cumbersome
and, in literature, it is common to avoid their formal definition
on behalf of a more readable presentation.
For this reason, we only describe informally but exhaustively
how to build the \emph{on-demand} stream machine which corresponds to
a generic program $P \in \lang{\stm_\LA}$. The correspondence between
$\SFPOD$ is expressed by the following Proposition:

\begin{prop}
  \label{prop:SFPODimplSIFPLA}
  For every $P \in \lang{\stm_\LA}$ there is a $M_P \in \SFP$ such that
  for every $x \in \Ss$ and $\eta \in \Bool^\Ss$, $P(x, \eta)=P(x, \eta)$.
  Moreover, if $P$ is polytime, then $M_P$ is polytime.
\end{prop}
\begin{proof}
The construction relies on the fact that it is possible to implement a $\SIFPLA$ program
by means of a multi-tape \emph{on-demand stream machine}
which uses a tape to store the values
of each register, plus an additional tape containing the partial results
obtained during the evaluation of the expressions
and another tape containing $\eta$.
We denote the tape used for storing the result coming from
the evaluation of the expressions with $e$.

The machine works thanks to some invariant properties:
\begin{itemize}
\item On each tape, the values are stored to the immediate right of the head.
\item The result of the last expression evaluated is stored on the $e$ tape to the immediate right of the head.
\end{itemize}

The value of a $\SIFP$ expression can be easily computed using the $e$ tape.
We show it by induction on the syntax of the expression:
\begin{itemize}
  \item Each access to the value stored in a register basically consist in a copy of the
  content of the corresponding tape to the $e$ tape, which is a simple operation,
  due to the invariant properties properties mentioned above.
  \item Concatenations ($f.\zero$ and $f.\one$)
  are easily implemented by the addition of a character at the
  end of the $e$ tape which contains the value of $f$,
  as stated by the induction hypothesis on the invariant properties.
  \item The binary expressions are non-trivial, but since one of the
  two operands is a register identifier, the machine can directly compare $e$
  with the tape which corresponding to the identifier, and to replace the content of $e$
  with the result of the comparison, which in all cases $\zzero$ or $\oone$.
\end{itemize}
All these operations can be implemented without consuming any character
on the oracle tape and with linear time with respect to the size of the
expression's value.
%
To each statement $s_i$, we assign a sequence of machine states,
$q_{s_i}^I, q_{s_i}^1, q_{s_i}^2, \ldots, q_{s_i}^F$.
\begin{itemize}
  \item Assignments consist in a copy of the value in $e$ to the tape corresponding to
  the destination register and a deletion of the value on $e$ by replacing its symbols
  with $\circledast$ characters. This can be implemented without consuming any character
  on the oracle tape.
  \item The sequencing operation $s;t$ can be implemented inserting in $\delta$
  a composed transition from $q_s^F$ to $q_t^I$, which does not consume the oracle tape.
  \item A $\while$ statement $s:= \while f t$ requires the evaluation of $f$ and
  then passing to the evaluation of $t$, if $f\sred \one$, or stepping
  to the next transition if it exists and $f\not\sred \one$.
  After the evaluation of the body, the machine returns to the initial state of
  this statement, namely: $q_s^I$.
  \item A $\rb$ statement is implemented consuming a character on the tape and copying
  its value on the tape which corresponds to the register $R$.
\end{itemize}

The following invariant properties hold at the beginning of the execution and are kept true
throughout $M_P$'s execution. In particular, if we assume $P$ to be polytime,
after the simulation of each statement, it holds that:

\begin{itemize}
  \item The length of the non blank portion of the first tapes corresponding to
  the register is polynomially bounded because their contents are precisely
  the contents of $P$'s registers, which are polynomially bounded
  as a consequence of the hypotheses on their polynomial time complexity.
  \item The head of all the tapes corresponding to the registers point to the
  leftmost symbol of the string thereby contained.
\end{itemize}

It is well-known that the reduction of the number of tapes on a polytime Turing Machine
comes with a polynomial overhead in time; for this reason, we can conclude that
the polytime \emph{multi-tape} on-demand stream machine we introduced above can be
\emph{shrinked}
to a polytime \emph{canonical} on-demand stream machine. This concludes the proof.
\end{proof}

It remains to show that each
on-demand stream machine can be reduced to an equivalent STM.

\begin{lemma}
  \label{lemma:SFPODtoSFP}
  For every $\STM = \langle \Qs,\Qs_{\natural}, \Sigma, \delta, q_0\rangle
  \in \SFPOD$, the machine $\STM' = \langle \Qs, \Sigma, H(\delta), q_0\rangle
  \in \SFP$ is such that for every $n \in \Nat$,
  for every configuration of $\STM$ $\langle \sigma, q, \tau, \eta\rangle$ and for
  every $\sigma', \tau' \in \Ss, q \in \Qs$:
  $$
  \mu \left(\{\eta \in \Bool^\Nat| \exists \eta'. \langle \sigma, q, \tau, \eta\rangle\reaches n \delta \langle \sigma', q', \tau', \eta'\rangle\}\right)
  =
  \mu \left(\{\chi \in \Bool^\Nat|  \exists \chi'. \langle \sigma, q, \tau, \xi\rangle\reaches n {H(\delta)} \langle \sigma', q', \tau', \chi'\rangle\}\right).
  $$
\end{lemma}

\noindent
Even in this case, the proof relies on a reduction. In particular,
we show that
given an on-demand stream machine $\STM$ it is possible to build a
stream machine $\STM'$ which is equivalent to $\STM$.
Intuitively, the encoding from an \emph{on-demand} stream machine $\STM$
to an ordinary stream
machine takes the transition function $\delta$ of $\STM$ and substitutes
each transition not causing the oracle tape to shift --- i.e. tagged with $\natural$ ---
with two distinct
transitions, so to match both $\zero$ and $\one$ on the tape storing $\omega$.
This causes the resulting machine to reach the same target state with the same behavior n the work tape, and to shift to the right the head on the oracle tape.

\begin{defn}[Encoding from On-Demand to Canonical Stream Machines]
  \label{def:SFPODtoSFPmap}
  We define the encoding from an On-Demand Stream Machine to a Canonical Stream Machine
  as below:
  $$
    H := \langle \mathbb Q, \Sigma, \delta, q_0\rangle \mapsto \left\langle \mathbb Q, \Sigma, \bigcup\Delta_H(\delta), q_0\right\rangle.
  $$
  where $\Delta_H$ is defined as follows:
  \begin{align*}
    \Delta_H(\langle p, c_r, \zzero, q, c_w, d\rangle) &:= \{\langle p, c_r, \zzero, q, c_w, d\rangle\}\\
    \Delta_H(\langle p, c_r, \oone, q, c_w, d\rangle) &:= \{\langle p, c_r, \oone, q, c_w, d\rangle\}\\
    \Delta_H(\langle p, c_r, \natural, q, c_w, d\rangle) &:= \{\langle p, c_r, \zzero, q, c_w, d\rangle, \langle p, c_r, \oone, q, c_w, d\rangle\}.\\
  \end{align*}
\end{defn}

\begin{proof}[Proof of Lemma \ref{lemma:SFPODtoSFP}]
  The definition of $\reaches n \delta$ allows us to rewrite the statement
  $
  \exists \eta'. \langle \sigma, q, \tau, \eta\rangle\reaches n \delta \langle \sigma', q', \tau', \eta'\rangle
  $
  as:
  \begin{align*}
    &\exists \eta', \eta''\in \Bool^\Nat.\exists c_1, \ldots, c_k.\\
    &\langle \sigma, q, \tau, c_1c_2\ldots c_k\eta'\rangle\reaches {n_1} \delta \langle \sigma_1, q_{i_1}, \tau_1, c_1c_2\ldots c_k\eta'\rangle\reaches 1 \delta \langle \sigma_1', q_{i_1}', \tau_1, c_2\ldots c_k\eta'\rangle \land \\
     &\langle \sigma_1', q_{i_1}', \tau_1, c_2\ldots c_k\eta'\rangle  \reaches {n_2} \delta \langle \sigma_2, q_{i_2}, \tau_2, c_2\ldots c_k\eta'\rangle\reaches 1 \delta \langle \sigma_2', q_{i_2}', \tau_2', c_3\ldots c_k\eta'\rangle \land \\
     &\langle \sigma_2', q_{i_2}', \tau_2', c_3\ldots c_k\eta'\rangle \reaches {n_3} \delta \ldots \reaches {n_{k+1}} \delta
    \langle \sigma', q', \tau', \eta'\rangle.
  \end{align*}
  The claim can be rewritten as follows:\\[2ex]
  \resizebox{\textwidth}{!}{%
  $\begin{aligned}            %
    &\exists \eta'', c_1, \ldots, c_k \in\Bool. \exists n_1, \ldots, n_{k+1}\in \Nat.  \forall \xi_1, \ldots, \xi_{k+1}\in \Ss. |\xi_1|=n_1\land \ldots |\xi_{k+1}|=n_{k+1}\land\\ %
    &\langle \sigma, q, \tau, c_1c_2\ldots c_k\eta' %
    \rangle\reaches {n_1} \delta \langle \sigma_1, q_{i_1}, \tau_1, c_1c_2\ldots c_k\eta'\rangle\reaches 1 \delta \langle \sigma_1', q_{i_1}', \tau_1, c_2\ldots c_k\eta'\rangle \land \\ %
     &\langle \sigma_1', q_{i_1}', \tau_1, c_2\ldots c_k\eta'\rangle  \reaches {n_2} \delta \langle \sigma_2, q_{i_2}, \tau_2, c_2\ldots c_k\eta'\rangle\reaches 1 \delta \langle \sigma_2', q_{i_2}', \tau_2', c_3\ldots c_k\eta'\rangle \land \\ %
     &\langle \sigma_2', q_{i_2}', \tau_2', c_3\ldots c_k\eta'\rangle \reaches {n_3} \delta \ldots \reaches {n_{k+1}} \delta %
    \langle \sigma', q', \tau', \eta'\rangle\\ %
    &\Longleftrightarrow\\                     %
    &\langle \sigma, q, \tau, \xi_1c_1\xi_2c_2\ldots c_k\xi_{k+1}\eta''\rangle\reaches {n_1} {H(\delta)} \langle \sigma_1, q_{i_1}, \tau_1, c_1\xi_2c_2\ldots c_k\xi_{k+1}\eta''\rangle\reaches 1 {H(\delta)} \langle \sigma_1', q_{i_1}', \tau_1, \xi_2c_2\ldots c_k\xi_{k+1}\eta''\rangle \land \\ %
     &\langle \sigma_1', q_{i_1}', \tau_1, \xi_2c_2\ldots c_k\xi_{k+1}\eta''\rangle  \reaches {n_2} {H(\delta)} \langle \sigma_2, q_{i_2}, \tau_2, c_2\ldots c_k\xi_{k+1}\eta''\rangle\reaches 1 {H(\delta)} \langle \sigma_2', q_{i_2}', \tau_2', \xi_3c_3\ldots c_k\xi_{k+1}\eta''\rangle \land \\ %
     &\langle \sigma_2', q_{i_2}', \tau_2', \xi_3c_3\ldots c_k\eta''\rangle \reaches {n_3} {H(\delta)} \ldots \reaches {n_{k+1}} {H(\delta)} %
    \langle \sigma', q', \tau', \eta''\rangle. %
  \end{aligned}$%
  }~\\[2ex]
  \normalsize
  Intuitively, this holds because it suffices to take the
  $n_i$s as the length of longest sequence of non-shifting
  transitions of the on-demand stream machine and the
  correspondence can be proven by induction on the number of steps of each formula in the conjunction.
  Thus, we can express the sets of the claim as follows:
  \small
  \begin{align*}
    \{\eta \in \Bool^\Nat| \exists \eta'. \langle \sigma, q, \tau, \eta\rangle\reaches n \delta \langle \sigma', q', \tau', \eta'\rangle\} &= \{\eta \in \Bool^\Nat| \forall 0 \le i < k. \eta(i) =c_i\rangle\}\\
    \{\chi \in \Bool^\Nat|  \exists \chi'. \langle \sigma, q, \tau, \xi\rangle\reaches n {H(\delta)} \langle \sigma', q', \tau', \chi'\rangle\} &= \{\chi \in \Bool^\Nat| \forall 1 \le i \le k. \chi(n_i+i)=c_i\land \chi(0)=c_1 \rangle\}.
  \end{align*}
  \normalsize
  The conclusion comes because both these sets are cylinders with the same measure.
\end{proof}

\begin{proposition}[From $\POR$ to $\SFP$] \label{prop:PORtoSFP}
For any $f:\Ss^{k}\times \Bool^{\Ss}\to \Ss $ in $\POR$ there exists a function $f^{\sharp}:\Ss^{k}\times \Bool^{\Nat} \to \Ss$ such that for all $n_{1},\dots, n_{k},m\in \Ss$,
\[
\mu \big(\{ \eta\in \Bool^{\Nat}\mid f(n_{1},\dots, n_{k},\eta)=m\}\big)=
\mu \big(\{ \omega\in \Os\mid f^{\sharp}(n_{1},\dots, n_{k},\omega)=m\}\big).
\]
\end{proposition}

\begin{proof}
  This result is a consequence of Lemma \ref{lemma:portosifp}, Lemma \ref{lemma:sifpratosifpla}, Proposition \ref{prop:SFPODimplSIFPLA} and Lemma \ref{lemma:SFPODtoSFP}.
\end{proof}


\section{Proofs from Section \ref{sec:SSBPP}}
\label{sec:appPIT}

\subsection{The Randomized Algorithm}
\label{sec:appABalgo}

Let us first describe the randomized algorithm $\mathtt{PZT}$ in more detail:
\begin{enumerate}
\item If the input $x$ is not the output of a circuit, reject it. Otherwise, 
  let $n$ be its arity, $d$ its degree and $m$ its size.
  Set $i$ to $1$.
\item Check whether $m$ is 
smaller than some constant value $\varrho \in \Nat$.
 \begin{itemize}
  \item If so, walk the table $T$ 
  looking for a pair $(x_j, y_j)$ where $x_j=x$; set $o_i=0$ if $y_j=1$, set 
  $o_i=1$, otherwise.
  \item Otherwise, proceed as follows. Choose $r_1, \ldots, r_n$ uniformly and 
  independently in $\{0, \ldots, 2^{m+3}-1\}\subseteq\Nat$. Let $k$ be a random 
  value in  $\{1, \ldots, 2^{2m}\}\subseteq \Nat$. Finally, evaluate the result 
  of $x$, seen as a circuit, on $r_1, \ldots, r_n\mod k$, with result $o_i$.
  \end{itemize}
\item If $i<s$, then increase $i$ by $1$ and go back to $2$.
\item If for all $i$, $1\le i \le s$, 
$o_i=0$ output $\epsilon$; otherwise, output $\zzero$.
\end{enumerate}

%
Let us now discuss the formula $G$ that represents $\mathtt{PZT}$.
To do so, we introduce a set of basic predicates, which will be employed for the definition of $G$:
\begin{align*}
  \mathit{Circ}(x) &:= \text{$x$ is the encoding of a circuit with a single output}\\
  \mathit{Eval}(x, k, y, t) &:= \text{When fed with inputs encoded by $y$, $x$ produces in output $t$ modulo $k$}\\
  \mathit{NumVar}(x, n) &:= \text{$x$ is the encoding of an arithmetic circuit with $n$ variables}\\
  \mathit{Degree}(x, d) &:= \text{The degree of the arithmetic circuit encoded by $x$ is $d$}\\
  T(x, y) &:= \bigvee_{\overline x \in \bigcup_{i=0}^\varrho \{0, 1\}^i} \left(x=\overline x \land y = y_{\overline x}\right)\\
  K(r, z) &:= \text{$z$ is a uniformly chosen random string in $\{0, 1\}^{|r|}$}.
\end{align*}
All these predicates characterize polytime random functions; for this reason, we can assume without lack of generality that they are $\Sigma^b_1$-formul\ae{} of $\Lpw$. Using some of these predicates, it is possible to define a formula $G_1$ which executes one evaluation of the polynomial $x$:\\[2ex]
\resizebox{\textwidth}{!}{$
\begin{aligned}
  &G_1(x, m, n, d, z, y) :=  \left[ m \le \varrho \land T(x, 1)\to y=0 \land T(x, 0)\to y=1 \right] \lor \Big[(m > \varrho)\ \land\\
                           &\big(\left( y=0 \land \exists z_0, z_1. |z_0| = 2m \land |z_1| = n \cdot (m+3) \land z_0\cdot z_1 = z \land \exists t \preceq z. \left(\mathit{Eval}(x, z_0, z_1, t) \land t \neq 0\right)\right)\lor\\
                           &\left( y=1 \land \exists z_0, z_1. |z_0| = 2m \land |z_1| = n \cdot (m+3) \land z_0\cdot z_1 = z \land \mathit{Eval}(x, z_0, z_1, 0)\right)\big)\Big].
\end{aligned}
$}\\[2ex]
\normalsize

\begin{remark}
  The formula $G_1$ is a $\Sigma^b_1$ predicate of $\Lpw$ which characterizes one iteration of the algorithm \texttt{PZT}.  
\end{remark}

Differently from the original algorithm, the formula $G_1$ employs an additional parameter $z$ as a \emph{source of randomness} to determine the values of $r_i, \ldots, r_n$ and $k$. This way, we are able to isolate the randomization part in a small portion of the formula we are building, i.e. that one where we determine the value of $z$ by means of the predicate $\Flip$. 

Thanks to this construction, $G_1$ is a $\Sigma^b_1$ formula realizing some $\Flip$-free $\POR$ function $g_1$. We can leverage this fact to define another $\Flip$-free function $\iota_{g_1}(x, n, d, y, z, i)$ which iterates the function $g_1$ $i$ times with the $i$-th $(n(m+3)+2m)$-long sub-string of $z$ as source of randomness --- i.e. as argument for $g$'s $z$ --- and returns $\epsilon$ if and only if all these executions of $g$ returned $1$, otherwise it returns $0$.

This proves that there is a $\Flip$-free $\Sigma^b_1$ formula $G^v$ realizing that function provable under $\RS$ --- and even $S^0_1$. A picture of the proof of the existence of $G^v$ is given in Figure \ref{fig:igv}. Intuitively, the formula $G_v$ characterizes steps 2-7 of the \texttt{PZT} algorithm.

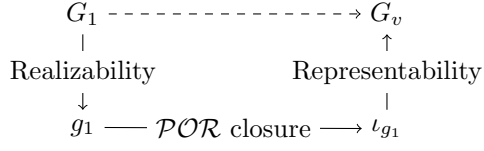
\begin{figure}
  \begin{center}
    \begin{tikzpicture}[node distance = 4cm]
      \node (gv) {$G_1$};
      \node[right of = gv] (Igv) {${G_v}$};
      \node[below = 1 cm of gv] (g) {$g_1$};
      \node[below = 1 cm of Igv] (iota) {$\iota_{g_1}$};
      
      \draw[->] (gv) edge node[fill=white] {Realizability} (g);
      \draw[->] (g) edge node[fill=white] {$\POR$ closure} (iota);
      \draw[->] (iota) edge node[fill=white] {Representability} (Igv);
      \draw[->] (gv) edge[dashed] (Igv);
    \end{tikzpicture}
  \end{center}
  \caption{Summary of the proof of the existence of ${G^v}$}
  \label{fig:igv}
\end{figure}

In order to define $G$, we will only need to compose $G_v$ with another sub-formula which characterizes the first step of the \texttt{PZT} algorithm. This is quite simple because we only need to encode the generation of the values of $d, n, m, z$ and to fix a number of iterations, as we will show in section \ref{sec:apperrorbound}, $\nruns$ is an appropriate choice. Thus, the $\Sigma^b_1$ formula $G$ can be defined as follows:\\[2ex]
\resizebox{\textwidth}{!}{$
\begin{aligned}
  &G(x, y):= \exists m \preceq x. |x| = m\land
              \Big[\big((\mathit{Circ}(x)\land \exists n \preceq 2^m.\exists d \preceq 2^m. \exists z. |z| = \nruns\cdot(n(m+3)+2m) \land\\
            &\mathit{NumVar}(x, n) \land \mathit{Degree}(x, d) \land K(1^{\nruns \cdot (n(m+3)+2m)}, z) \land {G^v} (x, m, n, d, y, z, \nruns)\big)
              \lor \big(\lnot \mathit{Circ}(x)\land y=0\big)\Big].
\end{aligned}$%
}\\[2ex]
\normalsize

\subsection{Proving the Error Bound}
\label{sec:apperrorbound}

Within this section, we argument that:
$$
\IDE\vdash \forall x.\forall y.\MEAS[ G(x,y)\leftrightarrow H(x,y)],
$$
\noindent
which is equivalent to showing that:
$$
\IDE\vdash \forall x.\forall y.\Thresh[\Flipless[G](x,y,z)\leftrightarrow H(x,y)].
$$
In turn, it is possible to obtain a formula equivalent to $\Flipless[G](x,y,z)$ by removing the quantification over $z$ and the randomization its value by means of $K$. Doing so, we are able to obtain an equivalent claim:   
\small
\begin{align*}
\IDE\vdash &\forall x.\forall y.\exists m, d, n\preceq 1^{|x|}. |x|=m\land \mathit{NumVar}(x, n) \land \mathit{Degree}(x, d)\\
&\exists^{\frac 2 3 2^{\nruns \cdot (n(m+3)+2m)}}z.|z| = \nruns \cdot (n(m+3)+2m) \land (G^v(x, n, d, z, y)\leftrightarrow H(x,y)).
\end{align*}
\normalsize
It also is easy to observe $\exists^{\ge h}z. |z| = k\land P(z) \Leftrightarrow |\{z \in \Bool^k\mid P(z)\}|  \ge h$. 
This allows to replace the threshold quantification with a formula of $\IDE$ measuring the cardinality of some finite set. Therefore, we can reduce out goal to showing Lemma \ref{lemma:errbound} using only the instruments provided by $\IDE$:
\begin{lemma}
  \label{lemma:errbound}
  For every encoding of a polynomial circuit $x$, for every $y\in \Bool$, whereas $x$ has $n$ variables, size $m$ and degree $d$, it holds that:
  $$
  |\{z \in \Bool^{\nruns \cdot (n(m+3)+2m)} \mid (G^v(x, n, d, z, y)\leftrightarrow H(x,y))\}|\ge \frac 2 3\cdot 2^{\nruns \cdot (n(m+3)+2m)}.
  $$  
\end{lemma}
Since $G$ and $H$ characterize two decisional functions, the claim of Lemma \ref{lemma:errbound} can be restated in the following way:\\[2ex]
\resizebox{\textwidth}{!}{$\begin{aligned}
|\{z \in \Bool^{\nruns \cdot (n(m+3)+2m)} \mid (G^v(x, n, d, z, 0)\land H(x,0))\lor(G^v(x, n, d, z, \epsilon)\land H(x,\epsilon))\}|\ge \frac 2 3\cdot 2^{\nruns \cdot (n(m+3)+2m)}.
\end{aligned}$}\\[2ex]
Thus, it is possible Lemma \ref{lemma:errbound} as a consequence of \eqref{eq:claim5.1} and \eqref{eq:claim5.2} which, respectively, can be stated as Lemmas \ref{lemma:claim5.1} and \ref{lemma:claim5.2}.

\begin{lemma}[Claim \eqref{eq:claim5.1}]
  \label{lemma:claim5.1}
  For every $z \in \Bool^{\nruns \cdot (n(m+3)+2m)}$, every encoding of a polynomial circuit $x$, every $y\in \Bool$, whereas $x$ has $n$ variables, size $m$ and degree $d$, it holds that $G^v(x, n, d, z, 0)\to H(x,0)$.
\end{lemma}
\begin{proof}
  This result is a consequence of the compatibility of $\bmod k$ with respect to addition, multiplication ad inverse in $\mathbb Z$. Precisely: for every $n, m, k \in \mathbb Z$ with $k \ge 2$, it holds that:
  \begin{enumerate}
  \item $(n\odot m)\mod k=\left((n\mod k) \odot (m \mod k)\right) \mod k$ for $\odot \in \{+, \cdot\}$.
  \item $(-n) \mod k= - (n \mod k)$.
  \item $n \mod k \neq 0 \to n \neq 0$.
  \end{enumerate}
  These claims are shown as follows:
  \begin{enumerate}
    \item In this case, we will only show the case for $+$, that of $\cdot$ is analogous. The proof goes by induction on the recursion parameter, e.g. $x$. The case $x=0$ is trivial. For the inductive case, let $x=ak+b$ and $y=ck+d$ for $b, d <k$, the existence and uniqueness of this decomposition can be shown by induction on $x$. The IH tells that $(x\mod k +y \mod k)\mod k= b+d\mod k$.
 \begin{align*}
    (x+1\mod k +y \mod k)\mod k&= (x+1\mod k +d)\mod k\\
                               &= (b+1\mod k +d)\mod k    
  \end{align*}
  Now, if $b<k-1$, then $(b+1\mod k +d)\mod k= b+1+d \mod k=ak+b+1+ck+d \mod k= x+1+y\mod k$. If $b=k-1$ 
  \begin{align*}
    (x+1\mod k +y \mod k)\mod k&= (d)\mod k\\
                               &= (b+1 + ak+ck +d)\mod k\\
                               &= x+1+y\mod k
  \end{align*}
\item This proof goes by cases on $n$. The case where $n=0$ is trivial. The case $n+1$ relies on the uniqueness of the decomposition of $n=ak+b$, which allows us to show that $-((ak+b+1)\mod k) = -((b+1) \mod k)$. If $0\le b < k-1$, then it is equal to $-b-1$, otherwise it is equal to $0$. On the other hand,   $(-ak-b-1)\mod k = (-b-1)\mod k$ and still, if $0\le b < k-1$, then this is equal to $-b-1$, otherwise it is equal to $0$.
  \item The counter-nominal is trivial: $n=0 \rightarrow n \mod k =0$.
  \end{enumerate}
  Leveraging points 1 and 2, we show that the evaluation of a polynomial $x$ as performed by the predicate $\mathit{Eval}$ on input $\vec r$ is equal to $x(\vec r)\mod k$, the assumption $G^v(x, n, d, \vec rk, 0)$ allows us to conclude the premise of point 3, thus an application of point 3, allows us to conclude that $\mathbb Z \models p(\vec x) \neq 0$, by the definition of $H$, we conclude that $H(x, 0)$ holds.
\end{proof}
\begin{lemma}[Claim \eqref{eq:claim5.2}]
  \label{lemma:claim5.2}
  For every encoding of a polynomial circuit $x$, every $y\in \Bool$, whereas $x$ has $n$ variables, size $m$ and degree $d$, it holds that:
  $$
  |\{z \in \Bool^{\nruns \cdot (n(m+3)+2m)} \mid G^v(x, n, d, z, \epsilon)\to H(x,\epsilon)\}|\ge \frac 2 3\cdot 2^{\nruns \cdot (n(m+3)+2m)}.
  $$
\end{lemma}

\subsection{Proof of Lemma \ref{lemma:claim5.2}}
\label{sec:proofofclaim5.2}

Lemma \ref{lemma:claim5.2} is shown finding an upper bound to the probability of error of the algorithm, i.e. to the number of values for $z$ causing $G^v(x, n, d, z, \epsilon)\to H(x,\epsilon)$. Intuitively, there are two possible causes of error within the algorithm underlying the formula $G$.
\begin{enumerate}
\item if the outcome of the evaluation of the polynomial on $\vec r$ is some value $y\neq 0$ and $k$ divides $y$, then the algorithm will reject its input even though $x$ belongs to the language. 
\item if the values $\vec r$ are a solution of the \emph{non-identically zero} polynomial, then the algorithm will reject $x$, even though it belongs to the language.
\end{enumerate}
The bound to the first error is found in section \ref{subsub:arorabarak}, while a bound to the probability of the second error is found in section \ref{subsub:schwartzzippel}. Lemma \ref{lemma:claim5.2} can be shown combining these results. 

\subsubsection{Estimation of the Error, Case 1}
\label{subsub:arorabarak}

\begin{prop}[Argument in \cite{AroraBarak}]
  \label{prop:ARargument}
  There is some $\varrho \in \Nat$ such that for every $m \in \Nat$ greater than $\varrho$, every $0\le y < 2^{({m+3})\cdot2^m}$:
  \small
  $$
  |\{k \in \{1, \ldots, 2^{2m}\}\mid \}\mid k \text{ is not \emph{a prime not dividing} }y\}| \le 2^{2m} - \frac {2^{2m}} {16m}
  $$
\end{prop}

%
This probability depends on the range bounding $k$ and on the number of evaluations taken in exam. The proof of this result relies on the following observations:
%

\begin{enumerate}
\item[(\emph{a})] Thanks to the Prime Number Theorem, there is some $m'$ such that, for every $m \ge m'$, the number of primes in $\{1, \ldots, 2^{2m}\}$ is at least $K=\frac {2^{2m}}{4 \cdot 2m}$.
\item[(\emph{b})] For sufficiently big $m$, the number of primes in $\mathbb Z$ dividing $0\le y < 2^{({m+3})\cdot2^m}$ is smaller than $L= 3m\cdot 2^m$.
\item[(\emph{c})] For every $n \in \Nat$, it holds that $3(2n+201)\leq 2^{n+96}$.
\item[(\emph{d})] For every $m\ge 100$, said $L$ the number of primes dividing $y$, it holds that $L\leq \frac K 2= \frac {2^{2m}}{16 m}$. This is shown by induction using (\emph{c}).
\end{enumerate}

As we anticipated, points $(a)$, $(b)$, $(d)$ hold only for sufficiently big values of $m$. For this reason, we define $\varrho$ as the greatest of these values.


\paragraph{Point (\emph{a})}

We exploit here the fact, proved in \cite{CornarosDimitracopoulos}, that the PNT is provable in $\IDE$; the formulation of the PNT can be paraphrased as follows:
$$
\forall x\in \Nat. \forall q \in \mathbb Q^+. \exists a_x, z_x. z_q \le x \Rightarrow \left | \frac {\pi(x)\log_{a_x}^*(x)} x - 1 \right| \le q
$$
Where $\pi(x)$ is the number of primes in $\{1, \ldots, x\}$ and $\log_{a}^*(x)$ is a \emph{good approximation} of $\log_2(x)$, i.e. $\forall x.\log_2(x) \le 2 \log_{a_x}^*(x)\le 4 \log_2(x)$. This can be used to deduce that:
$$
z_q \le x \Rightarrow \frac x {\log_{a_x}^*(x)}\left(1-q\right) \le \pi(x) \le \frac x {\log_{a_x}^*(x)}\left(1+q\right),
$$
Thus, fixing $q=\frac 1 2$ we obtain the following claim:
$$
z_q \le x \Rightarrow \frac x {4\cdot \log_2(x)}\le \frac {x} {2\log_{a_x}^*(x)}\le  \pi(x) 
$$

We have shown that:

$$
\forall x. \pi(x) = |\{ p \in \{1,\ldots, x\} \mid p \text{ is prime}\}| \ge \frac x {4 \cdot \log_2(x)}.
$$

For sake of readability, we name $\Pi(x)$ the set $\{ p \in \{1,\ldots, x\} \mid p \text{ is prime}\}$.

\paragraph{Point(\emph{b})}
To this aim, it suffices to observe that $y$ is the product of $q_1, \ldots, q_L$ primes where $L \le |y| = \log(2^{(m+3)2^m})$. This is a consequence of the following version of the Fundamental Theorem of Arithmetic.
\begin{theorem}
  $\IDE \vdash \forall n\ge 2. \mathit{prime}(n) \lor \exists S, s. \mathit{card}(S, s)\land s\le |n|\land \forall q \in S. \mathit{prime}(q) \land q|n$. 
\end{theorem}
\begin{proof}
  We observe that the predicates $\in, |, \mathit{prime}, \mathit{card}$ can be easily modeled in the language of arithmetic.  Then we go by induction on $n$ to show the claim, observing that the formula $A(n)=\mathit{prime}(n) \lor \exists S, s. \mathit{card}(S, s)\land s\le |n|\land \forall q \in S. \mathit{prime}(q) \land q|n$ is $\Delta^{0}_{1}$.
 The base case is trivial since $2$ is prime and can be divided only by $2$, so the cardinality of $S$ is smaller than $|2|$. For the inductive case suppose that $n+1$ is prime, in this case the claim is trivial, otherwise, if $n+1$ is not prime, there are $a, b \in \Nat$, $ab=n+1$, thus we can apply the IH on $a, b$, building $S$ as the union of $S_a, S_b$.
\end{proof}
Thus, $L \le |y| = {(3+m)\cdot 2^m}={3\cdot 2^m+m\cdot 2^m}\le {3m\cdot 2^m}$ --- the last step is for $m\ge 1$, and is shown by induction. To sum up:
$$
\forall 0 \le y \le 2^{(m+3)2^m}. |\{ p \in \mathbb Z \mid q \text{ divides } y \}| \le {3m\cdot 2^m}.
$$

For simplicity's sake, we omit the proof of $(c)$ and $(d)$, which are standard by induction.

\begin{proof}[Proof of Proposition \ref{prop:ARargument}]
  From $(d)$ we can deduce that $|\{p \in \Pi(2^{2m})\mid p \text{ does not divide }y\}| \\>\frac {2^{2m}}{4 \cdot 2m}$. Thus, the claim is a trivial consequence of the following observation, which concludes the proof:
  \begin{multline*}
    {|\{1, \ldots, 2^{2m}-1\} \setminus\{ p \in \Pi(2^{2m}) \mid p \text{ does not divide } y\}|}\le\\
    2^{2m}-{|\{ p \in \Pi(2^{2m}) \mid p \text{ does not divide } y\}|}
  \end{multline*}
\end{proof}

\subsubsection{Estimation of Error, case 2}
\label{subsub:schwartzzippel}

The goal of this section is to show that for every non-zero $n$-variate polynomial $p$ of degree $d$ and every set $S\subseteq \mathbb Z$, $S$ contains sufficiently many witnesses of the fact that $p$ is a non-zero polynomial.

\begin{remark}\label{rem:polynomials}
We here implicitly assume an encoding of polynomials as finite strings of coefficients, so that a unary polynomial $p(x)=\sum_{i=0}^{i<k}a_{i}x^{i}$ translates into (a suitable encoding of) the string $(a_{0},a_{1},\dots, a_{k-1})$. Observe that the property ``$\forall x.p(x)=q(x)$'' is decidable: once $p$ and $q$ are translated into strings of coefficients it is enough to check equality of these strings. For this reason we can suppose that all statements of the form  ``$\mathbb Z\models\forall x.p(x)=q(x)$'' or ``$\mathbb Z\models\exists x. p(x) \neq q(x)$'' are encoded via $\Delta^{0}_{1}$-formulas. 
\end{remark}

\begin{lemma}[Schwartz-Zippel]
  \label{lemma:sz}
  For every \emph{$n$-variate} polynomial $p$, for every $S \subseteq \mathbb Z$, said $d$ the degree of $p$, it holds that:
  $$
  {|\{(x_1,\ldots, x_n) \in S^n\ |\ \mathbb Z \models p(x_1, \ldots, x_n) = 0\}|}\le d |S|^{n-1}
  $$
\end{lemma}

The proof of this result is by induction on the number of variables and the degree of the polynomial $p$. It relies on a weak statement of the Fundamental Theorem of Algebra, proved below (Corollary \ref{cor:szbasecase}), stating that each univariate polynomial has at most $d$ roots in $\mathbb Z$, where $d$ is the degree of the polynomial.  We start by defining the notion of non-zero univariate polynomial:

\begin{defn}
  We say that $p \in \POLY$ is a univariate \emph{irreducible} polynomial if and only if it is univariate, and $\mathbb Z \models \forall x. p(x) \neq 0$.
  Moreover, we say that $p \in \POLY$ is a \emph{non-zero} polynomial if and only if $\mathbb Z \models \exists x. p(x) \neq 0$.
\end{defn}

This notion extends naturally to multivariate polynomials. The first step of the proof is to show that each univariate polynomial in $\mathbb Z$ can be expressed as the product of $\prod_{i=0}^{i<k}(x- \overline x_i) q(x)$ for $k$ smaller than the degree of $d$. This result relies itself on the proof of the correctness of the polynomial division algorithm, in particular on the following properties:

\begin{remark}~
  \begin{enumerate}
    \item
      Let $p$ be an univariate polynomial in normal form with coefficients in $\mathbb Z$ and let $\overline x$ be a solution of $p$, The division algorithm applied on $p, (x-\overline x)$ outputs a polynomial $q$ with no remainder.
    \item
      Let $p$ be an univariate polynomial in normal form with coefficients in $\mathbb Z$ and let $r$ be a non-zero polynomial in normal form with degree smaller than that of $p$. If $q$ is obtained with remainder $0$ applying the division algorithm to $p$ and $q$, it holds that $rq=p$.
    \item
        The degree of $r=p/q$ plus the degree of $q$ is equal to the degree of $p$.
    \end{enumerate}
\end{remark}

Although these results may seem very simple, for our purpose it is important to ensure that they can be established within $\IDE$. This is true, because the algorithm performing polynomial division is polytime and thus its totality, as well as its correctness, can be proved in $\IDE$.

%
%
Now, we want to prove that each uni-variate non-zero polynomial of degree $d$ has at most $d$ roots in $\mathbb Z$. We start by showing that it is always possible to express it as a product of simpler polynomials.
\begin{lemma}
  For every univariate polynomial $p$, if $p$ has $k$ distinct roots $\overline x_{0},\dots, \overline x_{k-1}$, then there is a polynomial $q$ such that $\mathit{deg}(q)=\mathit{deg}(p)-k$ and 
  $$
  \mathbb Z \models \forall x. p(x) = \prod_{i=0}^{i<k}(x- \overline x_i) q(x).
  $$
%
\end{lemma}
\begin{proof}
  We want to show that for all polynomial $p$, degree $d$, natural number $k<d$ and integers $x_{1},\dots, x_{d}$, if $p$ is unary, has degree $d$ and $\overline x_{1},\dots,\overline  x_{k}$ are distinct roots of $p$, then there is a polynomial $q$ of degree at most $\mathit{deg}(p)-k$ such that $p(x) = \prod_{i=0}^{i<k}(x-\overline x_{i})q(x)$. 
Using Remark \ref{rem:polynomials}, this can be encoded as a formula of the form $\forall p.\forall d .\forall s.C(p,d)$, where $C(p,d)$ only contains bounded quantifications, and $s$ is a bound on the values of the roots $\overline x_{i}$. 
 We will prove $\forall p.\forall d.\forall s.C(p,d)$ by induction on $d$: 
%
%
%
  \begin{itemize}
  \item If $d=0$, then the claim also holds trivially.
  \item If $d>0$ and $k=0$, then again the claim holds trivially. If $k>0$, then 
  polynomial division yields $p(x)\simeq (x-\overline x_{k})p'(x)$, with $\mathit{deg}(p'(x))=d-1$ and $\overline x_{1},\dots, \overline x_{k-1}$ roots of $p'(x)$. Then, by IH, we deduce $C(p',d-1)$, that is, $p'(x)=\prod_{i=0}^{i<k-1}(x-\overline x_{i})q(x)$, and thus
  $p(x)=  \prod_{i=0}^{i<k}(x-\overline x_{i})q(x)$ as desired.
%
%
%
  \end{itemize}
\end{proof}


\begin{cor}
  \label{cor:szbasecase}
  For every $S \subseteq \mathbb Z$ and for every non-zero univariate polynomial, it holds that:
  $$
  |\{ x \in S\ |\ \mathbb Z \models p(x) = 0\}| \le \mathit{deg} (p).
  $$
\end{cor}

\begin{proof}[Proof of Lemma \ref{lemma:sz}]
The proof is by induction on $n$:
  \begin{itemize}
  \item The claim for $n=1$ coincides with Corollary \ref{cor:szbasecase}.
  \item In the general case, we take the normal form of $p$, call that polynomial $p'$. It holds that
    $$
    \mathbb Z \models \forall \vec x. p(\vec x) = p'(\vec x),
    $$
    then we go by induction on $d$, we can factorize $x_1$ from each monomial of $p$ and show that
    $$
    \exists k. \mathbb Z \models \forall x_1, \ldots, x_n. p'(x_1, \ldots, x_n) = \sum_{i=0}^{k} x_1^i\cdot p_i(x_2, \ldots, x_{n})
    $$
    Where $p_k$ is non-zero. Applying the IH on $d$ to $p_k$, we obtain that:
    $$
    \{(x_2,\ldots, x_n) \in S^{n-1}\ |\ \mathbb Z \models p_k(x_2, \ldots, x_n) = 0\}\le (d-k)|S|^{n-2}
    $$
    Fix some  $(y_2, \ldots, y_n)\in S^{n-1}$ and assume that $p_k(y_2, \ldots, y_n)\neq 0$. In this case, the polynomial $\overline p_{y_2, \ldots, y_n}(x):=  \sum_{i=0}^{k} x^i\cdot p_i(y_2, \ldots, y_{n})$
    is not identically zero because it is a normal form and a normal form is identically zero if and only if all its coefficients are different from zero, but the coefficient of $x^k$ is different from $0$ for construction. Thus, we can apply the IH on $n$, showing that:
    \begin{equation}
      \forall (y_2, \ldots, y_n)\in S^{n-1}. |\{x \in S \mid \mathbb Z \models\overline p_{y_2, \ldots, y_n}(x) = 0\}|\le k,
      \tag{$*$}
    \end{equation}
    We continue by observing that 
    $$
    \forall (y_2, \ldots, y_n)\in S^{n-1}. \{x \in S \mid \mathbb Z \models \overline p_{y_2, \ldots, y_n}(x) = 0\} =
    \{x \in S \mid \mathbb Z \models p(x, y_2, \ldots, y_n) = 0\},
    $$
    for the definition of $\overline p$. We can conclude that\\[2ex]
  \resizebox{0.96\textwidth}{!}{%
  $\begin{gathered}
      \{(x_1,\ldots, x_n) \in S^n\ |\ \mathbb Z \models p(x_1, \ldots, x_n) = 0\} = \\[1.3ex]
      \{(x_1,\ldots, x_n) \in S^n\ |\ \mathbb Z \models p(x_1, \ldots, x_n) = 0 \land p_k(x_2, \ldots, x_n)=0\} \cup\\
      \{(x_1,\ldots, x_n) \in S^n\ |\ \mathbb Z \models p(x_1, \ldots, x_n) = 0 \land p_k(x_2, \ldots, x_n)\neq0\}=\\[1.3ex]
      \{(x_1,\ldots, x_n) \in S^n\ |\ \mathbb Z \models p(x_1, \ldots, x_n) = 0 \land p_k(x_2, \ldots, x_n)=0\} \cup\\
      \{(x_1,\ldots, x_n) \in S^n\ |\ \mathbb Z \models p(x_1, \ldots, x_n) = 0 \land p_k(x_2, \ldots, x_n)\neq0 \land \overline p_{x_2, \ldots, x_n}(x_1) = 0\}\subseteq\\[1.3ex]
      \{(x_1,\ldots, x_n) \in S^n\ |\ \mathbb Z \models p_k(x_2, \ldots, x_n)=0\} \cup
      \{(x_1,\ldots, x_n) \in S^n\ |\ \mathbb Z \models \overline p_{x_2, \ldots, x_n}(x_1) = 0 \} = \\[1.3ex]
      S \times\{(x_2,\ldots, x_n) \in S^{n-1}\ |\ \mathbb Z \models \overline p_{x_2, \ldots, x_n}(x_1) = 0 \} \cup
      \{(x_1,\ldots, x_n) \in S^n\ |\ \mathbb Z \models \overline p_{x_2, \ldots, x_n}(x_1) = 0 \}
    \end{gathered}$ %
    }\\[2ex]
    \normalsize
    This result can be lifted to sizes, obtaining the claim.
  \end{itemize}
\end{proof}

\begin{proof}[Proof of Lemma \ref{lemma:claim5.2}]
  We omit the case for $|x|< \varrho$, since the conclusion is trivial.
  The claim is equivalent to the following one:\\[2ex]
  \resizebox{\textwidth}{!}{%
  $\left|\left\{ z= \left(
      \begin{matrix}
        x_{1, 1} & \ldots & x_{n, 1 } & k_1\\
        x_{1, 2} & \ldots & x_{n, 2 } & k_2\\
        \vdots & \vdots & \ddots & \vdots\\
        x_{1, \nruns} & \ldots & x_{n, \nruns } & k_r
      \end{matrix} \right)\in \left(\{0, 2^{m+3}-1\}^n\times\{1,2^{2m}\}\right)^\nruns \mid    G^v (p, m, n, d, \epsilon, z, \nruns) \land H(p, 0)
   \right\}\right|\le  \frac 1 3 2^{\nruns\cdot(n(m+3)+2m)}$ %
}\\[2ex]
  The set on the left is:\\[2ex]
  \resizebox{\textwidth}{!}{%
  $\left\{      \left (\begin{matrix}
                     x_{1, 1} & \ldots & x_{n, 1 } & k_1\\
                     x_{1, 2} & \ldots & x_{n, 2 } & k_2\\
                     \vdots & \vdots & \ddots & \vdots\\
                     x_{1, \nruns} & \ldots & x_{n, \nruns } & k_r
                   \end{matrix}\right)
                 \in  \left(\{0, 2^{m+3}-1\}^n\times \{1, 2^{2m}\}\right)^\nruns \mid\\
                 \forall  1\le j\le 48m. \mathbb Z \models p(\vec x_j) =0 \lor (\mathbb Z \models p(\vec x_j) \neq 0 \land  k_j\text{ does divide }p(\vec x_j))\right\}$ %
}\\[2ex]
  Thus, it is equal to:
  \begin{multline*}
    \{(x_{1}, \ldots, x_{n}, k) \in  \left(\{0, 2^{m+3}-1\}^n\times \{1, 2^{2m}\}\right) \mid \\
    \mathbb Z \models p(\vec x) =0 \lor (\mathbb Z \models p(\vec x) \neq 0 \land  k\text{ does divide }p(\vec x))\}^\nruns\subseteq
  \end{multline*}
  \begin{multline*}
    \big(\{(x_{1}, \ldots, x_{n}, k) \in \{0, 2^{m+3}-1\}^n\times \{1, 2^{2m}\} \mid \mathbb Z \models p(\vec x) =0\} \cup\\
         \{(x_{1}, \ldots, x_{n}, k) \in \{0, 2^{m+3}-1\}^n\times \{1, 2^{2m}\} \mid \mathbb Z \models p(\vec x) \neq 0 \land  k\text{ does divide }p(\vec x)\}\big)^\nruns
  \end{multline*}
  Thus, its size is smaller than:
  \begin{align*}
      \big(|&\{(x_{1}, \ldots, x_{n}, k) \in \{0, 2^{m+3}-1\}^n\times \{1, 2^{2m}\} \mid \mathbb Z \models p(\vec x) =0\}| +\\
      |&\{(x_{1}, \ldots, x_{n}, k) \in \{0, 2^{m+3}-1\}^n\times \{1, 2^{2m}\} \mid \mathbb Z \models p(\vec x) \neq 0 \land  k\text{ does divide }p(\vec x)\}|\big)^\nruns.
  \end{align*}
  Our goal is equivalent to showing the ratio between that value and $2^{\nruns \cdot (n(m+3)+2m)}$ is smaller than $\frac 1 3$. Which, in turn, resolves to:\\[2ex]
  \resizebox{\textwidth}{!}{%
  $\begin{aligned}
      \Bigg(&\frac {|\{(x_{1}, \ldots, x_{n}, k) \in \{0, 2^{m+3}-1\}^n\times \{1, 2^{2m}\} \mid \mathbb Z \models p(\vec x) =0\}|}{2^{n(m+3)+2m}} +\\
      &\frac{|\{(x_{1}, \ldots, x_{n}, k) \in \{0, 2^{m+3}-1\}^n\times \{1, 2^{2m}\} \mid \mathbb Z \models p(\vec x) \neq 0 \land  k\text{ does divide }p(\vec x)\}|}{2^{n(m+3)+2m}}\Bigg)^\nruns\le \frac 1 3.
  \end{aligned}$%
  }\\[2ex]
  From Lemma \ref{lemma:sz} and point $(e)$, it is smaller than:\\[2ex]
  \resizebox{\textwidth}{!}{%
  $\begin{aligned}
    \Bigg(&\frac {2^m \cdot 2^{(m+3)\cdot(n-1)}\cdot  2^{2m}}{2^{n(m+3)+2m}} +\\
          &\frac{|\{(x_{1}, \ldots, x_{n}, k) \in \{0, 2^{m+3}-1\}^n\times \{1, 2^{2m}\} \mid \mathbb Z \models p(\vec x) \neq 0 \land  k\text{ does divide }p(\vec x)\}|}{2^{n(m+3)+2m}}\Bigg)^\nruns=\\
    \Bigg(&\frac 1 8 - h +
            \frac{|\{(x_{1}, \ldots, x_{n}, k) \in \{0, 2^{m+3}-1\}^n\times \{1, 2^{2m}\} \mid \mathbb Z \models p(\vec x) \neq 0 \land  k\text{ does divide }p(\vec x)\}|}{2^{n(m+3)+2m}}\Bigg)^\nruns\le
  \end{aligned}$%
  }\\[2ex]
  Applying Proposition \ref{prop:ARargument}, we can find an upper bound to this value, which is:
  \small
  \begin{align*}
    \Bigg(&\frac 1 8 - h+\left (1- \frac 1 8 +h\right)\frac{2^{n(m+3)+2m} - \frac{2^{n(m+3)+2m}}{16m}}{2^{n(m+3)+2m}}d\Bigg)^\nruns=\\
    \Bigg(&\frac 1 8 - h+\left(\frac 7 8 +h\right)\frac{2^{n(m+3)+2m} - \frac{2^{n(m+3)+2m}}{16m}}{2^{n(m+3)+2m}}\Bigg)^\nruns=\\
    \Bigg(&\frac 1 8 - h+\left (\frac 7 8 +h\right)\left( 1-\frac 1 {16m}\right)\Bigg)^\nruns=
           \Bigg(\frac 1 8 - h+\frac 7 8 -\frac {7} {8\cdot16m} +h-\frac h {16m}\Bigg)^\nruns=\\
    \Bigg(&1 -\frac {7} {8\cdot16m}-\frac h {16m}\Bigg)^\nruns=\Bigg(1 -\frac 1 {16m} \left(\frac 7 8 + h\right)\Bigg)^\nruns\le\Bigg(1 -\frac 1 {\frac 8 7 16m} \Bigg)^\nruns\le\\
    \le  \Bigg(&1 -\frac 1 {\frac 8 7 16m} \Bigg)^{{\frac 8 7 16m}} \cdot \Bigg(1 -\frac 1 {\frac 8 7 16m} \Bigg)^{{\frac 8 7 16m}}\le \frac 1 4
  \end{align*}
  \normalsize
  The very last observation is a consequence that $\forall n \ge 2.\left( 1+\frac 1 n\right)^n\ge 2$. This, in turn, is done expanding the binomial $\left( 1+\frac 1 n\right)^n$ and obtaining:
  \begin{align*}
    \left(1+\frac 1 n\right)^n &=\sum_{k=0}^n\frac {n!}{(n-k)!k!}\frac 1 {n^k}=\sum_{k=0}^n\frac 1{k!}\prod_{i=1}^{k-1}\left(1-\frac i {n^k}\right)=1+1+\sum_{k=2}^n\frac 1{k!}\prod_{i=1}^{k-1}\left(1-\frac i {n^k}\right)
  \end{align*}
  Where the derivation is justified by:
  \begin{itemize}
  \item The binomial expansion, namely: $\forall a, b, n. (a+b)^n = \sum_{k=0}^n \frac{n!}{k!(n-k)!}a^{n-k} b^k$. The proof is by induction on $n$. We will omit the details, which can be found in many text-books.
  \item $\forall k. k!\ge 2^k$. The proof is by induction. The base case is trivial, the inductive one is done as follows:
    $(k+1)!= k! (k+1) \ge 2^k (k+1)$. Now we go by cases on $k$: if it is $0$, the claim is trivial, otherwise we observe that $(k+1)\ge2$ and thus $2^k(k+1)\ge 2^{k+1}$.
  \item $\forall n. \sum_{k=1}^n 2^{-k}= 1- 2^{-n}$, which is shown by induction on $n$. The base case is trivial. The inductive one is done as follows:$\sum_{k=1}^{n+1} 2^{-k}= \sum_{k=1}^{n} 2^{-k}+ 2^{-n-1}= 1- 2^{-n}+ 2^{-n-1}=1- 2^{-n-1}$.
  \end{itemize}
  These proof are completely syntactic thus, modulo an encoding for rational numbers and for their operations, they can be done in $\PA$ without much effort. The proof proceeds observing that
  $$
  \left(1+\frac 1 n\right)^n=\left(1+\frac 1 n\right)^{-1\cdot -1 \cdot n}=\left(\frac n {n+1}\right)^{-n}=\left(1-\frac 1 n\right)^{-n}.
  $$
  Therefore, $\forall n\ge 4. \left(1-\frac 1 n \right)^n\le \frac 1 2$, and so:
  $$
  \Bigg(1 -\frac 1 {\frac 8 7 16m} \Bigg)^{{\frac 8 7 16m}} \cdot \Bigg(1 -\frac 1 {\frac 8 7 16m} \Bigg)^{{\frac 8 7 16m}}\le \frac 1 4.
  $$
\end{proof}

As we have shown, even last result can be proved in $\PA$, using an encoding of finite sets and standard arithmetic on $\mathbb N$. 

\subsection{Closure of $\BPP_{\mathsf T}$ under polytime reduction}
\label{sec:appbpppaclosure}

We assume $\mathsf T$ is any theory containing $\RSE$.
The statement of Proposition \ref{prop:bpppaclosure} can be given more precisely as: 

\begin{prop}
  For every language $L \in \BPP_{\mathsf T}$ and every language $M\in \Bool^*$, if there is an polytime function $r_{M, L}: \Ss \longrightarrow \Ss$, such that for every string in $\sigma \in\Ss$,
  $
  x \in M \leftrightarrow r(x)\in L
  $,
  then $M$ is in $\BPP_{\mathsf T}$.
\end{prop}
\begin{proof}
  Assume that $L \in \BPP_{\mathsf T}$, and let $G_L$ be the $\Sigma^b_1$ $\Lpw$ formula characterizing $L$ as required in Definition \ref{defn:BPPT}. Since $r_{M, L}$ is poly-time, there is a $\Sigma^b_1$ $\Flip$-free formula of $\Lpw$ $R$ characterizing $r_{M, L}$ as well.
  This has a consequence that the formula $C(x, y) := \exists w\le t(x). R(x, w) \land G_L(w, y)$ characterizes the composition of the reduction $r_{M, L}$ and the function $\chi_L$. The function $C(x, y)$ is still a $\Sigma^b_1$ formula, and characterizes a function deciding $M$ due to the hypothesis on the correctness of the reduction $r_{M, L}$. Form this conclusion and $\LANG{G_L}=L$, we deduce $\LANG{C}=M$, we only need to show point (2) of Definition \ref{defn:BPPT} for $C$:
  \begin{align*}
  \mathsf T &\vdash \forall x. \exists!y. \MEAS[C]\\
  \mathsf T &\vdash \forall x. \exists!y. \MEAS[\exists w\le t(x). R(x, w) \land G_L(w, y)]\\
   \mathsf T&\vdash \forall x. \exists!y. \Thresh[\Flipless[\exists w\le t(x). R(x, w) \land G_L(w, y)]]\\
   \mathsf T &\vdash \forall x. \exists w\le t(x). R(x, w) \land \exists!y. \Thresh[\Flipless[G_L(w, y)]]
  \end{align*}
  This is a consequence of a $\Sigma^b_1$ $\Flip$-free formula of $\Lpw$ and the hypothesis on $G_L$.
\end{proof}


\end{document}